\pdfoutput=1
\documentclass[acmsmall,manuscript]{acmart}
\usepackage{mathtools,amsthm,latexsym,url}
\usepackage{amsfonts}
\usepackage{longtable}
\usepackage{calrsfs}
\usepackage{xcolor}
\usepackage{graphicx}

\makeatletter
\long\def\@makecaption#1#2{
   \vskip 10pt
   \setbox\@tempboxa\hbox{{\footnotesize {\bf #1.} #2}}
   \ifdim \wd\@tempboxa >\hsize         %
       {\footnotesize {\bf #1.} #2\par}%
     \else                              %
       \hbox to\hsize{\hfil\box\@tempboxa\hfil}
   \fi}
\makeatother

\def\reals{\mathbb{R}}
\let\eps\varepsilon
\let\bd\partial

\def\OS{\mathcal{O}}
\def\WDS{\varUpsilon}
\def\ZWDS{\Upsilon^0}
\def\ZDS{\Sigma}
\def\ZZDS{\Sigma^0}
\def\zpred{\breve{\Pi}}

\def\pregion#1{\widetilde{#1}}

\def\curve{\sigma}
\def\ph{\varphi}
\def\tA{{t_a}}
\def\tB{{t_b}}
\def\tO{{t_o}}
\def\tQ{{t_q}}
\def\plate{\Delta}

\def\obj{O}
\def\query{Q}
\def\I{\mathcal{I}}
\def\oset{\mathcal{O}}
\def\qset{\mathcal{Q}}
\def\ospace{\mathbb{O}}
\def\qspace{\mathbb{Q}}
\def\qpred{\Pi}
\def\xx{\mathbf{x}}
\def\yy{\mathbf{y}}
\def\qtree{\varUpsilon}

\def\Aset{\A}
\def\Bset{\B}
\def\Aobj{A}
\def\Bobj{B}
\def\Qout{\Phi}

\def\A{\mathcal{A}}
\def\B{\mathcal{B}}
\def\C{\mathcal{C}}

\def\E{\mathcal{E}}
\def\EP{\mathcal{E}}
\def\F{\mathcal{F}}
\def\FS{\mathcal{F}}
\def\G{\mathcal{G}}

\def\S{\mathcal{S}}
\def\T{\mathcal{T}}

\def\W{\mathcal{W}}
\def\X{\mathcal{X}}
 
\def\ES{\mathbb{E}}

\def\TT{\mathbb{T}}

\def\CAD{{\Xi}}
\def\ZCAD{\hat{\Xi}}
\def\fiber{\Omega}
\def\zfiber{\hat{\Omega}}
\def\zcell{\sigma}
\def\zCADcell{\chi}
\def\frag{\varphi}
\def\fragset{\Phi}

\def\zfragset{\Lambda}
\def\arr{\mathcal{A}}
\def\cell{\tau}
\def\xx{\mathbf{x}}
\def\etal{\textit{et~al.}}

\newtheorem{theorem}{Theorem}[section]
\newtheorem{lemma}[theorem]{Lemma}

\newtheorem{corollary}[theorem]{Corollary}
\theoremstyle{remark}
\newtheorem{remark}{Remark}
\newtheorem*{remark*}{Remark}
\newtheorem*{remarks*}{Remarks}

\title{Intersection Queries for Flat Semi-Algebraic Objects in Three Dimensions and Related 
  Problems}
\titlenote{An abridged preliminary version of this work appeared in the Proceedings of the 38th Symposium on Computational Geometry (SoCG), 2022 \cite{socg-version}.}

\author{Pankaj K. Agarwal}
\affiliation{%
  \institution{Department of Computer Science, Duke University}
  \city{Durham}
  \state{NC}
  \country{USA}
}
\email{pankaj@cs.duke.edu}
\orcid{0000-0002-9439-181X}

\author{Boris Aronov}
\affiliation{
  \institution{Department of Computer Science and Engineering, 
    Tandon School of Engineering, New York University}
  \city{Brooklyn}
  \state{NY}
  \country{USA}
}
\email{boris.aronov@nyu.edu}
\orcid{0000-0003-3110-4702}

\author{Esther Ezra}
\affiliation{
  \institution{School of Computer Science, Bar Ilan University}
  \city{Ramat Gan}
  \country{Israel}
}
\email{ezraest@cs.biu.ac.il}
\orcid{0000-0001-8133-1335}

\author{Matthew J. Katz}
\affiliation{
  \institution{Department of Computer Science, Ben Gurion University}
  \city{Beer Sheva}
  \country{Israel}
}
\email{matya@cs.bgu.ac.il}
\orcid{https://orcid.org/0000-0002-0672-729X}

\author{Micha Sharir}
\orcid{0000-0002-2541-3763}
\affiliation{
  \institution{School of Computer Science, Tel Aviv University}
  \city{Tel Aviv}
  \country{Israel}
}
\email{michas@tauex.tau.ac.il}

\allowdisplaybreaks%
\begin{document}
 
\begin{abstract}
  Let $\T$ be a set of $n$ flat (planar) semi-algebraic regions in $\reals^3$ of constant complexity 
  (e.g., triangles, disks), which we call \emph{plates}.  We wish to preprocess $\T$ into a data structure 
  so that for a query object $\gamma$, which is also a plate, we can quickly answer various \emph{intersection queries}, 
  such as detecting whether $\gamma$~intersects any plate of~$\T$, reporting all the plates intersected by $\gamma$, or 
  counting them. We also consider two simpler cases of this general setting: (i)~the input objects are plates
  and the query objects are constant-degree parametrized algebraic arcs in $\reals^3$ (\emph{arcs}, for short), or
  (ii)~the input objects are arcs and the query objects are 
  plates in $\reals^3$. Besides being interesting in their own right, the data structures for these 
  two special cases form the building blocks for handling the general case. 
  
  By combining the polynomial-partitioning technique with additional tools from real 
  algebraic geometry, we present many different data structures for intersection queries, which also provide 
  trade-offs between their size and query time.
  For example, if $\T$ is a set of plates and the query objects are algebraic arcs, 
  we obtain a data structure that uses $O^*(n^{4/3})$ storage 
  (where the $O^*(\cdot)$ notation hides factors of the form $n^\eps$, for an arbitrarily small $\eps>0$) and answers an 
	arc-intersection query in $O^*(n^{2/3})$ time.  
  This result is significant since the exponents do not depend on the specific shape of the input and query objects.
  We generalize and slightly improve this result: for a parameter $s\in [n^{4/3}, n^\tQ]$, where $\tQ\ge 3$ is 
	the number of real parameters needed 
	to specify a query arc, the query time can be decreased to~$O^*((n/s^{1/\tQ})^{\tfrac{2/3}{1-1/\tQ}})$ 
  by increasing the storage to $O^*(s)$. 
  Our approach can be extended to many additional intersection-searching problems in three 
  dimensions, even when the input or query objects are not flat. 
\end{abstract}

       

\maketitle

\section{Introduction}
\label{sec:intro}

The general intersection-searching problem asks to preprocess a set $\T$ 
of geometric objects in~$\reals^d$ into a data structure, so that one can quickly report or count 
all objects of $\T$ intersected by a query object~$\gamma$, or just test 
whether $\gamma$ intersects any object of~$\T$ at all.  
Motivated by applications in various areas such as robotics, computer-aided design, computer graphics, and 
solid modeling, intersection-searching problems have been studied since the 1980s.
The early work~\cite{DE87,AE99} on intersection searching in computational geometry mostly focused on those instances
which could be reduced to simplex range searching in 2D or 3D, and more recently 
on segment-intersection or ray-shooting queries amid triangles in $\reals^3$---see the survey 
by Pellegrini~\cite{Pel:surv}.
However, very little is known about more general intersection queries in $\reals^3$, at least from the theoretical perspective. For instance, how quickly can one answer arc-intersection queries amid triangles in 
$\reals^3$, or triangle-intersection queries amid arcs in $\reals^3$? 

In this paper we  make significant, and fairly comprehensive, progress on the design of efficient solutions 
to general intersection-searching problems in $\reals^3$. We mainly investigate
intersection-searching problems in $\reals^3$ where both input and query objects are flat 
(planar) semi-algebraic regions of constant complexity 
(e.g., triangles, disks), which we refer to as \emph{plates}\footnote{%
  Roughly speaking, a semi-algebraic set in $\reals^d$ is the set of points 
  in $\reals^d$ satisfying a Boolean predicate over a set of polynomial inequalities; 
  the complexity of the predicate and of the set is defined in terms of the number of 
  polynomials involved and their maximum degree; see \cite{BPR} for details. We call a semi-algebraic set in $\reals^3$ 
  \emph{flat} (or \emph{planar}) if it is contained in a (two-dimensional) plane.} and/or (not necessarily planar) arcs. 
  In particular, we study the following three broad classes of intersection-searching problems:
\begin{itemize}
	\item[(Q1)] the input objects are plates and the query objects are arcs in $\reals^3$, 
	\item[(Q2)] the input objects are arcs and the query objects are plates in $\reals^3$, and
	\item[(Q3)] both input and query objects are plates in $\reals^3$.
\end{itemize}
 Beyond these three classes of queries, we also study some cases when both query and input objects are non-planar. 

 These instances of intersection searching arise naturally in the applications mentioned above, and 
 no data structures are known for them that perform better than what one would obtain using 
 the recent semi-algebraic range-searching techniques as described below.
 (As far as we know, no one has yet applied these fairly recent techniques to the problems discussed in this work.)

\subsection{Related work}
\label{subsec:related-work}

Intersection searching is a generalization of range 
searching (in which the input objects are points) and point enclosure queries
(in which the query objects are points), so it is not surprising that range-searching techniques have been extensively used for intersection searching~\cite{AE99,AAEZ}. More precisely,
the intersection condition between an input object and a query object
can be written as a first-order formula involving polynomial equalities and inequalities.
Using quantifier elimination, intersection queries can be reduced to \emph{semi-algebraic range} queries,
by working in \emph{object space}, where each input 
object~$O$ is mapped to a point~$O^*$ and a query object~$\gamma$ is mapped to a 
semi-algebraic region~$\pregion{\gamma}$, such 
that $O^*\in\pregion{\gamma}$ if and only if $\gamma$ 
intersects the corresponding input object $O$. Alternatively, the 
problem can be reduced to a \emph{point-enclosure} query, by working 
in \emph{query space}, where now each input object~$O$ is mapped to a 
semi-algebraic region~$\pregion{O}$ and each query object~$\gamma$ is 
mapped to a point~$\gamma^*$, so that 
$\gamma^*\in \pregion{O}$ if and only if $\gamma$ intersects~$O$.
The first approach leads to a linear-size data structure with sublinear 
query time, and the second approach leads to a large-size data structure (albeit polynomial in the input size)
with logarithmic or polylogarithmic query time; see,~e.g.,~\cite{AAEZ,AM94,AMS,BHO*,Ch85,MP,YY85}, where this technique has been applied to simpler %
instances of intersection queries.
One can also combine the two approaches to obtain a query-time/storage trade-off. 
We refer to this approach as the \emph{range-searching based approach}.

The performance of these data structures depends on the number of parameters 
needed to specify the input and query objects. We refer to these numbers as 
the \emph{parametric dimension} 
of the input and query objects, respectively.
Sometimes (quite often in the present study) the performance can be improved using a 
\emph{multi-level data structure}, where the data structure at each level is constructed in a lower-dimensional 
subspace of the parametric space, using only some of the degrees of freedom that specify an object.  
We refer to the maximum dimension of a subspace over all levels of the data structure as the 
\emph{reduced parametric dimension}, which is equal to the maximum number of parameters of 
input/query objects used in a polynomial inequality in the Boolean formula describing the 
intersection condition; see Appendix~\ref{app:multi-level} for a more formal definition.
The performance of such a multi-level data structure depends on the reduced parametric 
dimension.  As an illustration, segments in the plane have four degrees of freedom, but the standard way of computing intersections between them uses a multi-level structure where each level considers only two parameters (the coordinates of one endpoint or the coefficients of the supporting line), so the reduced parametric dimension in this case is only two.
We note that while the parametric dimension of input (or query) objects is independent of the underlying intersection predicate, the reduced parametric dimension depends on the predicate, and one often expresses the intersection 
predicate in a way that minimizes the reduced parametric dimension.

If the reduced parametric dimensions of the input and query objects 
are $\tO $ and $\tQ $, respectively, and~$s\in[n,n^\tQ]$ 
is a \emph{storage parameter}, then using the recently developed techniques for semi-algebraic range searching based on 
the polynomial-partitioning method~\cite{AMS,AAEZ},
a query can be answered in $O^*((n/s^{1/\tQ})^\rho)$
time using $O^*(s)$ storage and expected preprocessing time,
where $\rho=\tfrac{1-1/\tO}{1-1/\tQ}$.
See Appendix~\ref{app:multi-level}.
(As in the abstract, the $O^*(\cdot)$ notation hides factors
of the form $O(n^{\eps})$, for arbitrarily small $\eps > 0$, and their 
coefficients, which depend on $\eps$.)
For example, a segment-intersection 
query amid $n$ triangles in $\reals^3$ can be answered in $O^*(n^{3/4})$ 
time using $O^*(n)$ storage, in $O(\log n)$ time using $O^*(n^4)$ storage, or in
$O^*(n/s^{1/4})$ time using $O^*(s)$ storage,
for $n \le s \le n^4$, by combining the first two solutions~\cite{Pel,Pel:surv}. 
A similar multi-level approach yields data structures in which a 
segment-intersection-detection query amid $n$ planes or spheres in~$\reals^3$ can be answered in
$O^*(n/s^{1/3})$ time using $O^*(s)$ storage, for $n\leq s \leq n^3$,~\cite{MoSh,Pel,Pel:surv,ShSh}---these 
data structures can be extended to answering reporting queries at an additional cost that is linear in the output size.
However, the extension 
to answering counting queries does not apply to the setting of spheres.  

A departure from this approach is the \emph{pedestrian} approach for 
answering ray-shooting queries.  For instance, given a simple polygon $P$ with $n$ edges, a Steiner triangulation of $P$ can be constructed so that a line segment lying inside $P$ intersects only $O(\log n)$ triangles. A query is answered by 
traversing the query ray through this sequence of triangles~\cite{HS95}. 
The pedestrian approach has also been applied 
to polygons with holes in~$\reals^2$~\cite{AAS95,HS95}, to a convex polyhedron 
in~$\reals^3$~\cite{DK83}, and to polyhedral subdivisions in~$\reals^3$~\cite{AAS95,AF99}.
Some of the ray-shooting data structures combine the pedestrian 
approach with the above range-searching tools~\cite{Ag92,AS96,BHO*}.

Recently, Ezra and Sharir~\cite{ES} proposed 
a new approach for answering ray-shooting queries amid triangles in $\reals^3$.
Roughly speaking, 
they first construct a partition tree in $\reals^3$  using the polynomial-partitioning scheme of Guth~\cite{Guth} on the input triangles and their edges, and then reduce the problem to segment-intersection queries amid a set of planes in $\reals^3$.
The latter can 
be formulated as a three-dimensional simplex range-searching problem, and thus their approach
leads to a data structure with~$O^*(n^{3/2})$~storage and $O^*(n^{1/2})$~query time,
which improves upon the earlier solution~\cite{Pel}. 
They also obtain an improved trade-off between storage and query time by augmenting their approach with the previous approaches that  build a data structure in the four-dimensional parametric space of lines.
In addition, their approach supports segment-intersection reporting queries in~$O^*(n^{1/2} + k)$ 
time, where $k$ is the output size (with the same amount of storage).
However, this approach neither supports counting queries nor can it answer 
intersections queries with non-straight arcs (e.g.\ circular arcs), or handle non-flat input objects.  

\subsection{Our results} 
\label{subsec:results}

We present efficient data structures for many instances of the intersection searching tasks (Q1)--(Q3) that perform significantly better than the best known methods. For most settings considered here, the best known method is 
what one obtains using the general semi-algebraic range-searching based technique mentioned above.
The bottleneck in developing efficient data structures for answering (Q1)--(Q3) queries is that
the number of parameters needed to represent a plate   
 is typically much larger than three because of the parameters needed to specify its boundary.
 We aim to circumvent this challenge by first constructing a partition tree in the \emph{ambient space}
 $\reals^3$, which enables us to 
 express the arc-plate intersection condition (for Q1-instances, say) as a semi-algebraic predicate in which 
 each atomic predicate uses only a few of the parameters of a plate or of an arc,
 thereby attaining a small reduced parametric dimension. 
 For instance, if the input (resp. query) objects are plates, then ideally (albeit too ambitiously at first sight) we would like to reduce 
 $\tO$ (resp.\ $\tQ$) to $3$, the number of parameters needed to represent the plane containing the plate,
 and eliminate the dependence on the parameters needed to represent the boundary.
 Unlike segment-triangle intersection studied in~\cite{ES},
 several technical challenges arise in accomplishing this goal because the arc-plate intersection predicate 
 could be quite complex. For example, an arc may intersect a plate $\triangle$ several times and thus 
 may intersect $\triangle$ even if the arc lies above (or below) all of the edges of $\triangle$. 

 One of the main technical contributions of this work is to develop a general approach that leads to 
 a small reduced parametric dimension by combining the polynomial-partitioning technique with tools 
 from real algebraic geometry.
These tools enable us to reduce the original intersection-searching problem to simpler ones and to eliminate the 
asymptotic dependence on the boundary parameters completely in many cases, thereby improving $\tO$ and $\tQ$ 
significantly, e.g., to three for the case of plates. 
The most interesting among these tools is the 
\emph{cylindrical algebraic decomposition} (CAD) (see~\cite{BPR,Col,SS2}) 
that is used to construct a carefully tailored stratification
of a suitable parametric space. We 
exploit for this purpose the full power of CAD (see below). Other more efficient 
decomposition schemes, such as vertical decomposition, do not seem to work.

We describe, in Appendix~\ref{app:multi-level}, the multi-level data structure (mentioned above) for a fairly general 
\emph{semi-algebraic predicate} searching that allows space/query-time trade-off,
using the recent techniques for semi-algebraic range searching based on the 
polynomial-partitioning method~\cite{AMS,AAEZ}. Although this approach is 
similar to the multi-level partition trees described in~\cite{AE99,Mat94}, extending it 
to polynomial-partitioning based data structures is nontrivial and, as far as we know,
has not been described in any previous paper.
 We first describe (in Appendix~\ref{subsec:Pi-lin}) a data structure of $O^*(n)$ size, by applying the polynomial-partioning 
  technique in the object space, that answers a query in $O^*(n^{1-1/\tO})$ time. Next, we obtain 
  a space/query-time trade off
  (see Appendix~\ref{subsec:Pi-tradeoff}) by first applying the 
  polynomial-partitioning technique in the query space and then switching to the first data structure when 
  the size of a subproblem becomes sufficiently small. As mentioned above, for a storage parameter $s\in[n,n^\tQ]$,
  a query can be answered in $O^*\left( (n/s^{1/\tQ})^{\tfrac{1-1/\tO}{1-1/\tQ}}\right)$ time using $O^*(s)$ space.
We regard the careful and detailed presentation of this general machinery as another major contribution of this paper.
In this paper, we augment the aforementioned polynomial-partitioning scheme in the three-dimensional ambient space with 
these data structures to obtain space/query-time trade offs. These multi-level data structures already have  found other 
applications~\cite{ES-22}, and we believe they will continue to find additional applications.

In the rest of this introduction, we describe the specific results that we have obtained, compare them with the best-known
bounds or rather with best bounds achievable with the range-searching based approach (mentioned above), and briefly sketch the key ideas
that lead to the improved bounds. See Table~\ref{table:results} below for a summary of the main results. 
For simplicity, we mostly focus on answering \emph{intersection-detection} queries,
where we want to determine whether a query object intersects any input object of $\T$. 
Our data structures extend to answering \emph{intersection-reporting} queries, where we wish to
report all objects of~$\T$ that the query object intersects.
By combining our intersection-detection data structures with the 
parametric-search framework of Agarwal and Matou\v{s}ek~\cite{AM}, we can also answer \emph{extremal} intersection 
queries with one degree of freedom. For example, we can answer \emph{arc-shooting} queries, an extension of the well studied \emph{ray-shooting} queries, amid a set of plates, where 
for a directed query arc $\gamma$, the goal is to find the first intersection point of $\gamma$ 
with the plates as we walk along $\gamma$ (or report that there is no such a point). 
Most of the data structures extend to answering \emph{intersection-counting} queries as well, 
within the same asymptotic time bound.\footnote{%
        More generally, we can use the so-called semi-group model, i.e., the setup where given a semigroup $(S,+)$, each 
	input object is assigned a weight, which is an element of $S$. For a query object $\gamma$, the query procedure 
	returns the ``sum'' of weights of the
        input objects intersected by the query object. For example, if the 
	semigroup is $(\reals, \max)$, it returns the maximum-weight input plate/arc intersected by a query arc/plate.
        For counting queries, the data structures count the number of connected components of the \emph{intersections} 
        between the input objects and the query object (which may be a set of points or a set of arcs), 
        and not the number of input objects intersected by the query object (except when both input and query 
        objects are triangles).
        
        Throughout this paper, when we say that an intersection query can be 
        answered in $O^*(t(n))$ time, we mean
        that detection, counting, and extremal queries can be answered in $O^*(t(n))$ time,
        and reporting queries in $O^*(t(n))+O(k)$ time, where $k$ is the output size.
        We use an enhanced version of the \emph{real RAM} model of 
        computation~\cite{PS85} for our algorithms, in which we assume that various operations on 
        polynomials of fixed constant degree with a constant number of variables can be performed in $O(1)$ time; see~\cite{BPR}.
}

Before we summarize our specific results, we need
the following definition: We refer to a connected path $\gamma$ as a \emph{parametrized (algebraic) arc} if it is the restriction of a real algebraic curve $\sigma\colon I \rightarrow \reals^3$, where $I$ is either the real axis or the unit circle,
to a subinterval $[a,b]\subseteq I$. 
The \emph{parametric dimension}~$t$ of~$\gamma$ is the number of real parameters needed 
to describe $\gamma$.
Two of these parameters, namely $a$, $b$, specify the respective endpoints $\sigma(a)$ and $\sigma(b)$.
(For instance, the parametric dimension of a line segment is $6$: $3$ parameters to define each of its endpoints, or alternatively,
$4$ parameters to define the line containing the segment and $2$ parameters to define its endpoints.)
We assume that the degree of the polynomials defining $\sigma$ is bounded by 
some constant, the dependence on which will only show up in the constants hiding in the $O^*(\cdot)$ notation.
See Section~\ref{sec:main} for a more detailed discussion on the parametrization of algebraic arcs.\footnote{%
  Many of the data structures described in this paper work even with an implicit representation of the query arcs, where
  the curve supporting a query arc is the one-dimensional locus of real solutions to an equation of the form
  $F_1^2(x,y,z) + F_2^2(x,y,z)=0$, where $F_1, F_2 \in \reals[x,y,z]$ are  polynomials.
  It is known that any irreducible algebraic curve in $\reals^3$ can be described as $f(x,y)=0$ and 
  $z=g_1(x,y)/g_2(x,y)$, where $f, g_1, g_2$ are polynomials and $g_2 \not\equiv 0$, and that such a 
  representation can be computed efficiently~\cite{AB89,Be97}. However, for simplicity, throughout this paper we
  assume that we have a uni-parametric representation of the query arcs.
}

\paragraph*{Intersection searching with arcs amid plates.} (Sections~\ref{sec:main}--\ref{sec:linst}.)
Let $\T$ be a set of $n$ plates in $\reals^3$ in general position.  That is, we assume that  
any plane contains only $O(1)$ plates of $\T$, and any line is contained in supporting planes of $O(1)$ plates of $\T$.
We note that the general-position assumption on the input plates is necessary to make our algorithms efficient, as discussed in the beginning of Section~\ref{sec:main}.
Let $\Gamma$ be a parametrized family of (not necessarily planar) algebraic arcs in~$\reals^3$ (e.g., the family of all line segments, or the family of all circular or parabolic arcs). Let $\tO, \tQ$ be the reduced parametric 
dimensions of $\T$ and $\Gamma$, respectively, for their intersection predicate, which as mentioned above can be considerably smaller than their parametric dimensions.  For instance, although the parametric dimension of triangles in $\reals^3$ is $9$, as we will see below, $\tO=5$ for the intersection predicate with respect to algebraic arcs.
The majority of the paper is devoted to presenting several data structures for answering arc-intersection 
queries amid $\T$ with arcs in $\Gamma$ under various settings.

We begin (see Section~\ref{sec:main}) by presenting an $O^*(n^{4/3})$-size data structure that can be constructed in 
$O^*(n^{4/3})$~expected time and that supports arc-intersection queries in $O^*(n^{2/3})$ time.
(In fact, the exponent in the query time can be reduced to  slightly less than $2/3$, as stated in~Theorem~\ref{thm:space-query-tradeoff}.)
Our key technical contribution is that the asymptotic query time bound depends neither on the parametric dimension of the query 
arc nor on that of the input plates, though the coefficients hiding in the $O^*$-notation do depend on them, which we found quite surprising.
If we follow the range-searching based approach outlined in Appendix~\ref{app:multi-level} 
(and use Theorem~\ref{thm:Pi-tradeoff}), an $O^*(n^{4/3})$-size data structure will answer a query in 
time $O^*(n^\rho)$, where 
$\rho= (1-\tfrac{1}{\tO })(1-\tfrac{1}{3(\tQ -1)})$. For instance, if the input is a set of triangles and the query arcs 
are circular arcs, then naively $\tO=9$, $\tQ=8$, and $\rho \approx 0.847$. It might be possible to reduce the value of $\tO$ and $\tQ$ by examining the intersection predicate carefully, but we are unaware of any such result.

As in~\cite{ES}, we also construct a polynomial partitioning in the ambient space $\reals^3$, i.e., we compute 
a trivariate partitioning polynomial~$F$ of degree~$O(D)$ for the input plates as well as their boundary arcs, 
where $D$ is  a sufficiently large constant, using the algorithm in~\cite{AAEZ}. 
The zero set~$Z(F)$ of~$F$ partitions~$\reals^3$ into \emph{cells}, which 
are the connected components of~$\reals^3\setminus Z(F)$. Each cell is crossed by at most $n/D$ input plates and the boundary arcs of at most $n/D^2$ plates.
For a cell~$\tau$ of $\reals^3\setminus Z(F)$, the plates whose relative boundaries intersect $\tau$ 
are called \emph{narrow} at $\tau$, and the other plates intersecting $\tau$ are called \emph{wide}.
For each cell~$\tau$, we recursively preprocess the narrow plates of $\tau$,
and construct a secondary data structure for the wide plates of $\tau$. A~query is answered by traversing all cells 
of $\reals^3\setminus Z(F)$ that are intersected by the query arc. (If the query arc is contained in $Z(F)$ then we 
use a two-dimensional data structure, described in Section~\ref{sec:zero-set}, to answer an intersection query.) 
At each such cell, intersection with narrow plates is handled recursively, and the intersection with wide plates is processed nonrecursively, using the secondary data structure.
Handling wide plates is significantly more challenging than in~\cite{ES} because the query object, as well 
the edges of the  plates, are arcs instead of (straight) line segments, and we use 
a completely different approach that not only works for algebraic arcs but
extends to answering intersection-counting queries.
Roughly speaking, we construct a carefully tailored
CAD of a suitable parametric space, where the CAD is induced by the partitioning polynomial. 
For a plate~$\Delta$ and a partitioning polynomial $F$, $\Delta \setminus Z(F)$ consists of several
connected components. The CAD is used to further subdivide each component into smaller pieces (pseudo-trapezoids)
and label each piece that is fully contained in the relative interior of~$\Delta$.
The label is an explicit semi-algebraic representation of that piece, of constant complexity, 
that \emph{depends only on the equation of~$h_\Delta$}, the plane supporting $\Delta$ (and not on the boundary of~$\Delta$), and \emph{on the fixed polynomial~$F$} (Sections~\ref{subsec:usecad} and~\ref{subsec:plate-partition}). 
These labels enable us to formulate an arc-intersection query on wide plates as a 
three-dimensional semi-algebraic range query (Section~\ref{subsec:rs}), which is how we get
the query time to be independent of the parametric dimension of the plates.  
In our view, using CAD to handle wide plates is the most interesting idea of the paper, which we later apply to other settings as well.

Next, we present data structures (in Sections~\ref{sec:wide-tradeoff} and~\ref{sec:trade-off}) 
for answering arc-intersection queries amid plates, that provide a trade-off 
between size and query time. We first present such data structures (see Section~\ref{sec:wide-tradeoff}) for wide 
plates by using the CAD labels and the general framework of space/query-time trade-off 
described in Appendix~\ref{app:multi-level}. 
In general, if the query arcs have reduced parametric dimension~$\tQ$, then, using $O^*(s)$ storage, 
$s \in [n,n^{\tQ}]$, an arc-intersection query amid wide plates can be answered in 
$O^*\biggl((n/s^{1/\tQ})^{\tfrac{2/3}{1 -1/\tQ}}\biggr)$
time (see~Lemma~\ref{lem:wide-tradeoff}).
If the query arcs are \emph{planar}, then 
by exploiting the geometry of planar arcs, $\tQ$ (for the intersection predicate with respect to wide plates) 
can be improved to the reduced parametric dimension of the curves supporting the arcs in $\Gamma$,
by effectively eliminating the dependence on the endpoints of the query arcs 
(see Sections~\ref{subsec:circ_imp} and~\ref{subsec:planar_imp}).
For example, $\tQ=8$ for circular arcs 
(three for specifying the supporting plane, three for specifying the containing circle in 
that plane, and two for the endpoints) while $\tQ=6$ for circles. Therefore the query time improves
from~$O^*(n^{13/21})$ (the bound for $\tQ =8$) to~$O^*(n^{3/5})$ (the bound for $\tQ =6$), 
for the same asymptotic storage complexity $O^*(n^{3/2})$.
We note that the query time of a data structure, with storage parameter $s\in [n,n^\tQ]$,  based on 
the range-searching approach would be $O^*((n/s^{1/\tQ})^{\tfrac{1-1/\tO}{1-1/\tQ)}})$,
which is larger because $\tO$ is typically much larger than $3$.

With the space/query-time trade-off for wide plates at our disposal, we present
a trade-off for general plates in Section~\ref{sec:trade-off} by combining the partition tree described in Section~\ref{sec:main} with the range-searching technique of Appendix~\ref{app:multi-level}.  
For $s \le n^{3/2}$, we combine it with the $O^*(n)$-size data structure (see Section~\ref{subsec:Pi-lin}), and for 
$s>n^{3/2}$, we combine it with the $O^*(1)$ query-time data structure (see Section~\ref{subsec:Pi-tradeoff}).
For a storage parameter $s\in [n,n^{\tQ}]$,
an arc-intersection query can be answered in time
\[
O^* \biggl ( \frac{n^{2- 3/\tO}}{s^{1-2/\tO}} + 
	\biggl ( \frac{n}{s^{1/\tQ}}\biggr )^{\tfrac{2/3}{1-1/\tQ}} \biggr ) .
\]
If $\Gamma$ is a family of planar arcs, then $\tQ$ in the above bound is the reduced parametric dimension of the curves supporting the arcs in $\Gamma$. 

We next present in Section~\ref{sec:zero-set} a data structure 
for arc-intersection queries for the case when the query arcs lie on a fixed two-dimensional algebraic
surface of constant degree. Such a data structure is needed as a subroutine for the main 
data structure described in Section~\ref{sec:main}. Again, we combine the polynomial-partitioning technique with CAD 
for this task.  Using the fact that the query arcs lie on a fixed two-dimensional surface, we obtain a data structure of $O^*(n)$ size with $O^*(n^{2/3})$ query time.

We conclude the first part of the paper 
in Section~\ref{sec:linst} by considering a special case of plates, namely,
when $\T$ is a set of \emph{triangles} in $\reals^3$. We show that
the intersection condition of a triangle with a parametrized algebraic arc of constant complexity
can be expressed as a semi-algebraic predicate in which each 
polynomial inequality uses at most five of the nine parameters that specify a triangle, namely,
$\tO=5$ in this case.   This is accomplished by constructing a CAD induced by a suitable polynomial
in the joint space of query arcs and 
the space of planes in $\reals^3$, i.e., in $\reals^{t+3}$ if the parametric dimension of the 
curves supporting the query arcs is $t$, and by building a separate data structure for each cell of the CAD.
This leads to an $O^*(n)$-size data structure for triangles, that answers arc-intersection queries in
$O^*(n^{4/5})$ query time, a significant improvement over the best achievable %
query time of $O^*(n^{8/9})$ using Theorem~\ref{thm:Pi-tradeoff}. We also obtain a trade off between the query time and the size of the data structure.

\paragraph*{Intersection searching with plates amid arcs.}
Next, we present data structures (in Section~\ref{sec:lines}) for the complementary setup where the input
objects are arcs and we query with a plate.
We first show that we can preprocess a set $L$ of $n$ lines in $\reals^3$,
in expected time $O^*(n^{3/2})$, into a data structure  of size $O^*(n^{3/2})$,
by constructing a polynomial partitioning in $\reals^3$ on the input lines,
so that a plate-intersection query 
amid $L$ can be answered in~$O^*(n^{1/2})$~time. 
The algorithm constructs a CAD in $\reals^3$ induced by the partitioning polynomial, and uses a topological result 
to reduce the problem to plane-intersection searching amid a set of segments. The latter can be formulated 
as a three-dimensional simplex range-searching problem. The best achievable query time
for a data structure of size $O^*(n^{3/2})$, based on the range searching approach, is 
$O^*(n^\rho)$, where $\rho=\tfrac{3}{4}\biggl(1-\tfrac{1}{2(\tQ-1)}\biggr )$, where $\tQ$ is the parametric 
dimension of the query plates.
More generally, for a parameter $s\in [n, n^{\tQ}]$, an intersection query can be answered in time
$O^* \biggl (n^{5/4}/s^{1/2} + (n/s^{1/\tQ})^{\frac{1}{2 -3/\tQ}} \biggr )$ using $O^*(s)$ space by combining the $O^*(n^{3/2})$-size data structure with the range-searching based technique.
This data structure easily extends, with the same asymptotic performance,
to the case where the input is a set of line segments rather than full lines.

Finally, we consider the case where the input consists of a set of $n$ algebraic 
arcs of reduced parametric dimension $\tO$ in general position, i.e., any plane contains only $O(1)$ input arcs.
The query objects remain plates of reduced parametric dimension $\tQ$. Since a query plate may intersect an input arc multiple times, our data structure for line segments does not extend to arcs. 
Instead, we follow an approach similar to that in Section~\ref{sec:main} and combine polynomial partitioning with CAD. 
For a storage parameter $s\in [n,n^{\tQ}]$, the query time is 
$O^* \biggl ( (n/s^{1/3})^{\frac{1-1/\tO}{2/3}}+(n/s^{1/\tQ})^{\frac{1}{2-3/\tQ}} \biggr )$.
See Theorem~\ref{thm:arcs}.  In particular, for $s=n^{3/2}$, the query time is 
$O^*(n^{\frac{3}{4}(1-1/\tO)})$, which, interestingly, is asymptotically independent of 
$\tQ$. As a comparison point, the exponent in the query time with the range-searching approach would be 
$(1-\tfrac{1}{\tO})(1-\tfrac{1}{2(\tQ-1)})$.
In other words, the new approach reduces~$\tQ$ to~$3$, eliminating altogether the dependence on the 
complexity of the boundary of 
the query plate.

\paragraph*{Intersection searching with plates amid plates.\footnote{%
  This study was inspired by a question of Ovidiu Daescu related to collision detection in robotics.}}
The above results can be used to provide simple solutions for the case where both
input and query objects are plates (see Section~\ref{sec:plate-plate}).
For simplicity, assume first that both input and query objects are triangles in~$\reals^3$.
For $s \in [n,n^4]$, using $O^*(s)$ storage, a detection/reporting/extremal 
query can be answered in time $O^*( n^{5/4}/s^{1/2} + n^{4/5}/s^{1/5})$.
For counting queries, the query time is $O^* (n^{5/4}/s^{1/2} + n^{8/9}/s^{2/9})$. 
(This is the only case in this paper in which intersection-counting queries are
more expensive than intersection-detection or reporting queries.)

The technique can be extended to the case where both input and query objects are arbitrary plates. 
In this case, the boundary of a plate consists of $O(1)$ algebraic arcs of constant complexity. Let
$\tO $ and $\tQ $ be the reduced parametric dimensions of the input and the query plates, respectively.
We obtain a data structure of $O^*(n^{3/2})$ size with query time $O^*(n^\rho)$, where
$\rho = \max \left\{\frac{2\tQ -3}{3(\tQ -1)}, \frac{3(\tO -1)}{4\tO } \right\}$. If $\tO=\tQ=t\ge 3$, then 
$\rho=\frac{3}{4}(1-1/t)$.

Our data structure for the plate-plate case also works if the input and query objects 
are constant-complexity, not necessarily convex three-dimensional polyhedra.
This is because an intersection between two polyhedra occurs when their boundaries meet,
unless one of them is fully contained in the other, and the latter situation can be 
easily detected. We can therefore just triangulate the boundaries of both input and 
query polyhedra and apply the triangle-triangle intersection-detection machinery.

\paragraph*{The case of spherical caps.}
Finally, to show the versatility of our approach, we present (in Section~\ref{sec:caps}) an application of our technique to an 
instance where the input objects are not flat. Specifically, we show how to answer 
segment-intersection queries amid spherical caps (each being the intersection of a 
sphere with a halfspace) in $O^*(n^{27/40})$ using $O^*(n^{3/2})$ space. 
By combining this data structure with the general range-searching techniques,
for a storage parameter $s\in [n,n^6]$, a query can be 
answered in $O^*(n^{11/7}/s^{5/7}+n^{9/10}/s^{3/20})$ time.

\begin{table}
  \small
  \caption{Summary of results.  
    For simplicity, we state the bounds for fixed storage parameters.
    Storage is $O^*(n^\alpha)$, and query time is $O^*(n^\beta)$ using the range-searching approach
    in Appendix~\ref{app:multi-level} and $O^*(n^\gamma)$ using our new technique.
    We specify the values of $\alpha$, $\beta$, and $\gamma$ for each result.
    Here $\tO, \tQ$ are the reduced parametric dimensions of the input and the query objects, respectively.
    For the case of triangles vs\@. triangles, the exponent $5/8$ in the query time is obtained 
    by setting the parametric dimension to $4$ after some processing (without further processing the parametric
    dimension is $9$).
  }
  \label{table:results}
  \centering
  \begin{tabular}{|c|c|c|c|c|c|}
    \hline
    Input & Query & Storage & \multicolumn{2}{c}{Query Time Exponent}&Reference\\
    \cline{4-5}
	  &&Expon\@. ($\alpha$)&Old~($\beta$)&New~($\gamma$)&\\
    \hline\hline
    Plates & Arc/Curve & $4/3$&$(1-\tfrac{1}{\tO})(1-\tfrac{1}{3(\tQ-1)})$&$2/3$&Theorem~\ref{thm:trimain}\\
    Plates & Arc/Curve & $3/2$&$(1-\tfrac{1}{\tO})(1-\tfrac{1}{2(\tQ-1)})$&$\tfrac{2}{3}(1-\tfrac{1}{2(\tQ-1)})$&Theorem~\ref{thm:space-query-tradeoff}\\
    Triangles & Arc/Curve & $1$&$8/9$&$4/5$&Theorem~\ref{thm:linsto}\\
    Triangles & Arc/Curve & $11/9$&$\tfrac{8}{9}(1-\tfrac{2}{9(\tQ-1)})$&$2/3$&Theorem~\ref{thm:trimainx}\\
     \hline
     Segments& Plate&$3/2$&$\tfrac{3}{4}(1-\tfrac{1}{2(\tQ-1)})$&$1/2$&Theorem~\ref{thm:segs-tradeoff}\\
     Arcs/Curves & Plate&$3/2$&$(1-\tfrac{1}{\tO})(1-\tfrac{1}{2(\tQ-1)})$&$\tfrac{3}{4}(1-\tfrac{1}{\tQ})$&Theorem~\ref{thm:arcs}\\
     \hline
     Triangles&Triangle&$3/2$&$5/8$&$1/2$ {\small (report)} &Theorem~\ref{thm:trixx}\\
     Triangles&Triangle&$3/2$&$5/8$&$5/9$ {\small (count)} &Theorem~\ref{thm:tri-count}\\
     Plates&Plate&
     $3/2$&$(1-\tfrac{3}{2\tQ})$&$\max\left\{\frac{2\tQ -3}{3(\tQ -1)},\frac{3(\tO -1)}{4\tO }\right\}$&Theorem~\ref{thm:ext}\\
     \hline
     Spherical caps&Segment&$5/4$&$11/14$&$3/4$&Theorem~\ref{thm:caps}\\
     Spherical caps&Segment&$3/2$&$5/7$&$27/40$&Theorem~\ref{thm:caps}\\
     \hline
  \end{tabular}
\end{table}

\section{Intersection Searching with Query Arcs amid Plates}
\label{sec:main}

	We begin by describing the parametrization/representation of query arcs and input plates.
A parametrized (infinite) family $\Gamma$ of algebraic arcs in $\reals^3$ is defined
as follows. Recall that a function $y=f(x_1, \ldots, x_m)$ in $m$ variables is called an \emph{algebraic function} if it 
solves a polynomial equation in $m+1$ variables of the form $P_f(y,x_1,\ldots, x_m)=0$. We say that 
$P_f$ \emph{defines} the function $f$, and that the degree of $P_f$ is the \emph{degree} of $f$.
For some fixed constant $t>0$, let 
$\hat{x}(u,\alpha), \hat{y}(u,\alpha), \hat{z}(u,\alpha) \colon \reals^t\times\reals \rightarrow \reals$ be 
$(t+1)$-variate algebraic functions of bounded degree each. For a point $\delta\in\reals^t$,
let $x_\delta(\alpha) = \hat{x}(\delta,\alpha)$, $y_\delta(\alpha) = \hat{y}(\delta,\alpha)$, and 
$z_\delta(\alpha) = \hat{z}(\delta,\alpha)$ be the corresponding univariate algebraic functions
obtained by fixing $\delta$.
Then $\curve_\delta(\alpha) \coloneqq (x_\delta (\alpha), y_\delta(\alpha), z_\delta(\alpha))$, $\alpha\in\reals$,
is a parametrized algebraic curve defined by the parameter $\delta\in\reals^t$.\footnote{%
We note that we allow a more general parametrization than the commonly used rational parametrization (where each of 
$\hat{x}, \hat{y}, \hat{z}$ is a ratio of two polynomials) that parametrizes only zero-genus 
algebraic curves~\cite{Wa62}. There is also some work on parametrizing algebraic curves (of genus at most $6$) using 
radicals~\cite{SS11}. Fast algorithms are known for computing rational and radical parametrization of 
algebraic curves if they exist~\cite{AB88,AB89,SS11}. We are unaware of (efficient) algorithms for computing a 
more general (global) uni-parametrizations of algebraic curves, though a local parametrization can be computed using 
Puiseux series~\cite{Mau80} (which is not very useful in our setting) or an approximate parametrization can 
be computed~\cite{PSR10}.
\newline
		In general, the parametrization of an algebraic curve may define it at all but finitely many points.  
		For example,  the rational parametrization $x(\alpha)=\tfrac{1-\alpha^2}{1+\alpha^2}, 
		y(\alpha)=\tfrac{2\alpha}{1+\alpha^2}$ of the unit circle $x^2+y^2=1$ defines the circle at all 
		points except $(-1,0)$. In such cases, we handle these finite sets of points separately.
	}
Let $\ES_t \coloneqq \reals^t$ denote the space of these curves.
For a pair of real values $\alpha^-, \alpha^+$, 
$\gamma (\delta, \alpha^-, \alpha^+)$ defines the algebraic arc 
$\gamma \coloneqq \{ \curve_\delta(\alpha) \mid \alpha^- \le \alpha \le \alpha^+ \}$ 
in $\reals^3$; $\gamma$ is an empty arc if $\alpha^+<\alpha^-$.
We assume that $\alpha^-, \alpha^+$ are chosen such that $\gamma$ is a portion of an irreducible component of $\curve_\delta$, that $\gamma$ is connected, and that $\gamma$ does not contain any singular point of $\curve_\delta$.
Set 
\begin{equation}
	\label{eq:arc-family}
\Gamma \coloneqq \{ \gamma (\delta, \alpha^-, \alpha^+) \mid \delta \in \ES_t\; \mbox{and}\; 
\alpha^-, \alpha^+ \in \reals \} .
\end{equation}
The \emph{parametric dimension} of $\Gamma$  is $t+2$. We identify the space of arcs in $\Gamma$ with $\reals^{t+2}$, which we refer to as the \emph{query space} and in which an arc $\gamma(\delta, \alpha^-,\alpha^+)$ is mapped to the point $(\delta, \alpha^-, \alpha^+)$. 
Thus a query arc $\gamma(\xi)$ is given by specifying a point 
$\xi\in\reals^{t+2}$. We remark that although the functional form of $\hat x, \hat y, \hat z$, and thus of the query arc, is known in advance, the point $\xi \in \reals^{t+2}$ of a 
query arc $\gamma(\xi)$ is provided in an on-line manner when a query is issued.

Next, we describe the space of input plates. Recall that a plate is a planar semi-algebraic set of constant complexity in~$\reals^3$. For simplicity, without loss of generality, we assume that the Boolean formula describing the plate consists only of 
conjunctions, as disjunctions can be handled  by decomposing each plate into $O(1)$ plates, each described by a 
conjunction of polynomial inequalities.\footnote{
Note that a plate may be reported $O(1)$ times because of the decomposition step. Since a counting query counts the number of intersection points between the query arc and the 
input plates (and not necessarily the number of input plates intersected by the query arc), the decomposition step does not affect the result of a counting query.
}
Let $k, r_1, \ldots, r_k > 0$ 
be constant integers. For each $1 \le i \le k$, 
let $f_i (u_1, \ldots, u_{r_i}, x,y, z)$ be an $(r_i+3)$-variate polynomial 
of constant degree. For a 
point $\xi = (a_0,b_0,c_0)\in\reals^3$, let $h_\xi$ be the (non-vertical) plane $z=a_0x+b_0y+c_0$. 
For given points $\delta_1 \in \reals^{r_1}, \ldots, \delta_k \in \reals^{r_k}, \xi \in \reals^3$, we define the 
plate 
\[
  \Delta(\delta_1, \ldots, \delta_k, \xi) \coloneqq \Bigl\{ p\in\reals^3 \mid 
					\bigwedge_{i=1}^k (f_i (\delta_i, p)\ge 0) \wedge p\in h_\xi\Bigr \}  .
\]
We define the resulting family $\TT$ of plates as 
\begin{equation}
	\label{eq:plate-family}
	\TT \coloneqq \{ \Delta(\delta_1, \ldots, \delta_k, \xi) \mid \delta_1 \in \reals^{r_1}, \ldots, 
			\delta_k \in \reals^{r_k}, \xi \in \reals^3 \} .
\end{equation}
The parametric dimension of $\TT$ is $r \coloneqq 3 + \sum_{i=1}^k r_i$, and we identify the space of plates in $\TT$ with $\reals^r$, which we refer to as the \emph{object space} and in which a plate $\Delta(\delta_1, \ldots, \delta_k, \xi)$
is mapped to the point $(\delta_1, \ldots, \delta_k, \xi)$.
For example, we can define a family of planar \emph{lunes}, each formed by the intersection of two balls and a plane in $\reals^3$, as follows. Let $f(u_1,u_2,u_3, u_4, x,y,z) = -(x-u_1)^2-(y-u_2)^2-(z-u_3)^2+u_4^2$. For $\delta_1, \delta_2 \in \reals^4$ and $\xi\in\reals^3$,  we define a ``parametric lune'' as
\[ 
  \Delta(\delta_1, \delta_2, \xi) \coloneqq \{ p\in\reals^3 \mid 
					(f (\delta_1, p)\ge 0) \wedge (f (\delta_2,p) \ge 0) \wedge (p\in h_\xi) \}  .
\] 
The family of lunes is $\TT \coloneqq \{ \Delta(\delta_1, \delta_2, \xi) \mid \delta_1, \delta_2 \in \reals^4, \xi\in\reals^3\}$, and its parametric dimension is $11$.
For a pair $\Delta \in \TT$ and $\gamma \in \Gamma$, let $\Pi (\gamma, \Delta)$ be the semialgebraic predicate that is $1$ 
if and only if $\gamma \cap \Delta \ne \emptyset$. The \emph{reduced  parametric dimensions} of $\TT$ and $\Gamma$ (with respect to $\Pi$), denoted by $\tO$ and $\tQ$, respectively, are the maximum number of the 
$t+2$ (resp.\ $r$) parameters of $\Gamma$ (resp.\ $\TT$) 
being used in a polynomial inequality in the description of $\Pi$. 
We note that $\tO, \tQ$ depend on the predicate $\Pi$. 
A given parametrized family of objects, in general, may have multiple plausible parameterizations, some of which are more 
appropriate for ``querying against a particular family of objects.'' Different parameterizations may yield 
different reduced parametric dimensions, and 
can be much smaller than the parametric dimensions. We choose an appropriate representation  that aims to
minimize $\tO, \tQ$.
For instance, the parametric dimension of segments in $\reals^2$ is $4$ but $\tO,\tQ=2$ for 
segment-segment intersections (see~\cite{Ag:rs}); the parametric dimensions of triangles and segments in $\reals^3$ are $9$ and $6$ respectively, but $\tO=\tQ=4$ for segment-triangle intersections in $\reals^3$ (see~\cite{AM94}).

Let $\T \subset \TT$ be a set of $n$ plates in $\reals^3$ in general position,
in the sense that any plane contains only $O(1)$ plates  of $\T$,
and that any line is contained in the supporting planes of $O(1)$ plates of $\T$. 
We refer to boundary arcs of a plate as its \emph{edges}.
The first general-position assumption  on $\T$---any plane contains only $O(1)$ plates of $\T$---is critical for 
our data structures (the second assumption is made only for the sake of simplicity) because otherwise, for example, one has
to handle intersection queries between an arc of $\Gamma$ and boundary arcs of many plates, 
all lying in the same plane, say. 
The recent lower bounds on semi-algebraic range searching~\cite{AC21,AC22} 
imply that the time to answer such a query depends on the parametric dimension of the boundary arcs and an 
$O^*(n^{4/3})$-size data structure with $O^*(n^{2/3})$ query time is not feasible if the 
parametric dimension is large.

We present algorithms for preprocessing $\T$ into a data structure that can answer 
arc-intersection queries with arcs $\gamma \in \Gamma$ efficiently.
We begin by describing a basic data structure, and then show how its performance can be improved. 

\subsection{The overall data structure}
\label{subsec:data-structure}

Our primary data structure consists of a partition tree $\Psi$ on $\T$ in $\reals^3$, which is constructed using the
polynomial-partitioning technique of Guth~\cite{Guth}. More precisely, let $\X\subseteq\T$ be a subset of $m$ plates
and let $D>1$ be a parameter.
The analysis of Guth implies that there exists a real polynomial $F \in \reals[x_1,x_2,x_3]$ of degree at most $c_1D$,
where $c_1>0$ is a constant,
such that each open connected component (called a \emph{cell}) of $\reals^3\setminus Z(F)$
is crossed by at most $m/D$ plates of $\X$ and by the 
boundary arcs of at most $m/D^2$ of these plates;
the number of cells is at most $c_2D^3$ for another constant $c_2>0$. 
We refer to $F$ as a \emph{partitioning polynomial} for~$\T$.
Agarwal et al.~\cite{AAEZ} showed that such a partitioning polynomial can be constructed 
in $O(m)$ expected time if $D$ is a constant, turning Guth's existential result into an
efficient algorithm. Using such a polynomial partitioning, 
$\Psi$ can be constructed recursively in a top-down manner as follows.

Each node $v\in \Psi$ is associated with a cell $\tau_v$ of some polynomial partitioning 
and a subset~$\T_v\subseteq\T$. If $v$ is the root of~$\Psi$, then $\tau_v=\reals^3$ and $\T_v=\T$.  
Set $n_v \coloneqq |\T_v|$. We set a threshold parameter $n_0>0$, which may depend on $n$, and we fix 
a sufficiently large constant~$D$ for the partitioning.
For the basic data structure described here, we set $n_0\coloneqq \max\{ n^{1/3}, (4c_2D)^{1/\eps}\}$, where $\eps>0$ is an arbitrarily small constant to be chosen later; 
the value of $n_0$ will change when we later modify the structure.
Suppose we are at some node $v$. If $n_v\le n_0$ then $v$ is a leaf and we store $\T_v$ at $v$.
Otherwise, we construct, in time $O(|\T_v|)$, a partitioning polynomial $F_v$ for $\T_v$ of 
degree at most $c_1D$, as described above, and store~$F_v$ at $v$. If a plate $\plate$ is contained in $Z(F_v)$ then the plane supporting $\plate$ is also contained in $Z(F_v)$. The number of planes contained in $Z(F_v)$ is bounded by the degree of $F_v$.
By our general-position assumption, each such plane contains at most $O(1)$ plates, so at most $O(D)=O(1)$ plates
lie on $Z(F)$.
Let $\T_v^0 \subset \T_v$ be the subset of these plates. We store $\T_v^0$ at $v$.
We construct a secondary data structure $\ZDS_v$ on $\T_v\setminus \T_v^0$ for answering 
arc-intersection queries with the arcs of $\gamma\in \Gamma$ that are contained in $Z(F_v)$. 
Using Lemma~\ref{lem:onzf}, presented later in Section~\ref{sec:zero-set}, $\ZDS_v$ requires 
$O^*(n_v)$ storage and answers a query in $O^*(n_v^{2/3})$ time. It can be constructed in $O^*(n_v)$ time.

Next, we compute (semi-algebraic representations of) all cells of $\reals^3\setminus Z(F_v)$~\cite{BPR}. 
For each such cell $\tau$, we create a child $w_\tau$ of $v$ associated with $\tau$. 
We classify each plate $\Delta\in\T_v$ that crosses $\tau$ as \emph{narrow} (resp., \emph{wide})
at $\tau$ if an edge of $\Delta$ crosses $\tau$ (resp., $\Delta$ crosses $\tau$, but none of its edges does).
Let $\W_\tau$ (resp., $\T_\tau$) denote the subset of the plates in $\T_v\setminus \T_v^0$ that are 
wide (resp., narrow) at $\tau$.
We construct a secondary data structure $\WDS_\tau$ on $\W_\tau$,
as described in Section~\ref{sec:wide} below, for answering 
arc-intersection queries with arcs of $\Gamma$ amid the plates of $\W_\tau$ (within $\tau$).
$\WDS_\tau$ is stored at the child~$w_\tau$ of~$v$.
The construction of $\WDS_\tau$ for handling the wide plates is the main technical step in our algorithm.
By Lemma~\ref{prop:wide} in Section~\ref{sec:wide}, $\WDS_\tau$ uses $O^*(|\W_\tau|)$ space,
can be constructed in $O^*(|\W_\tau|)$ expected time, and answers an arc-intersection query 
in $O^*(|\W_\tau|^{2/3})$ time.
Finally, we set $\T_{w_\tau} \coloneqq \T_\tau$, and
recursively construct a partition tree for $\T_{w_\tau}$ and attach it as the subtree rooted at
$w_\tau$. 
Note that two secondary structures are attached at each node $v$, namely, $\WDS_v$ and $\ZDS_v$, 
for handling wide plates and for handling query arcs that are contained in $Z(F_v)$, respectively.

Denote by $S(m)$ the maximum storage used by the data structure for a subproblem involving 
at most $m$ plates. 
For $m\le n_0$, $S(m)=O(m)$. For $m > n_0$, the aforementioned
Lemmas~\ref{prop:wide} and~\ref{lem:onzf} imply that the secondary structures for a subproblem of size $m$ require $O^*(m)$ space. Therefore $S(m)$ 
obeys the recurrence:
\begin{equation}
  \label{eq:storage}
  S(m) \le
  \begin{cases*}
	  c_2D^3 \cdot S(m/D^2) + c_3 m^{1+\delta} & for $m\ge n_0$, \\[1mm]
    	  c_4 m & for $m\le n_0$,
  \end{cases*}
\end{equation}
where $c_2$ is the constant as defined above, $\delta>0$ is an arbitrarily small constant, $c_3>0$ is a constant that depends on $\delta$ and $D$, and $c_4>0$ is a constant. 
	We claim that the solution to the above recurrence, for our particular choice of $n_0$,  is
\begin{equation}
	\label{eq:size-bound}
	S(m) \le  A ( m^{3/2+\eps}{n_0^{-1/2}}+m ) = A ( m^{3/2+\eps}{n^{-1/6}}+m ),
\end{equation}
where $\eps\ge \delta$ is an arbitrarily small constant, $n$ is the original input size, and $A$ is a sufficiently large 
constant, assuming that $D \coloneqq D(\eps)$ is chosen suitably (see below). Indeed, the  bound trivially holds for 
$m \le n_0$. Using induction hypothesis for $m>n_0$ and plugging~\eqref{eq:size-bound} into~\eqref{eq:storage}, we obtain
\begin{align*}
	S(m) & \le  c_2D^3 \cdot A  \left ( \biggl ( \frac{m}{D^2} \biggr )^{3/2+\eps}\cdot \frac{1}{n^{1/6}}+ 
		\frac{m}{D^2}\right ) + c_3 m^{1+\delta} \\
		& \le  A \biggl ( \frac{m^{3/2+\eps}}{n^{1/6}} \biggr) \left [ \frac{c_2}{D^{2\eps}} + 
			\frac{c_2Dn^{1/6}}{m^{1/2+\eps}} + \frac{c_3n^{1/6}}{Am^{1/2+\eps-\delta}}\right ]\\
			&\le   A \biggl ( \frac{m^{3/2+\eps}}{n^{1/6}} \biggr) 
			     \left [ \frac{c_2}{D^{2\eps}}+\frac{c_2D}{m^\eps}+\frac{c_3}{Am^{\eps-\delta}}\right]
			     		\quad \mbox{(because $m > n_0 =n^{1/3}$)}\\
			&\le   A  m^{3/2+\eps}n^{-1/6},
\end{align*}
provided that we choose $D \ge (2 c_2)^{1/2\eps}$,  $A \ge 4 c_3$, and $\delta \le \eps$; recall that $n_0 \ge (4c_2D)^{1/\eps}$.
This establishes the induction step and thus proves \eqref{eq:size-bound}.
Initially, $m=n$, so the overall size of the data structure is $S(n) = O^*(n^{4/3})$.
A similar analysis shows that the expected preprocessing time is also $O^*(n^{4/3})$.

\subsection{The query procedure}
\label{subsec:plate-query}

Let $\gamma\in\Gamma$ be a query arc. We answer an arc-intersection query, say, intersection detection, for 
$\gamma$ by searching through~$\Psi$ in a top-down manner. 
Suppose we are at a node~$v$ of~$\Psi$. Our goal is to determine whether 
$\gamma_v \coloneqq \gamma\cap\tau_v$ intersects any plate of $\T_v$. 
For simplicity, assume that $\gamma_v$ is connected, otherwise we query with each of the $O(1)$ connected components of $\gamma_v$. 

If $v$ is a leaf, we answer the intersection query na\"ively, in $O(n_0)$ time, by inspecting all plates in $\T_v$. 
So assume that $v$ is an interior node. We first check, in $O(1)$ time, whether any of the plates in $\T_v^0$, the set of plates lying in $Z(F_v)$, 
intersects $\gamma$. If the answer is yes, we have detected an intersection and stop.
If $\gamma_v \subset Z(F_v)$, we query the secondary data structure $\ZDS_v$ 
with $\gamma_v$ (see Section~\ref{sec:zero-set}) and return the answer. 
(In this case there is no need to further recurse down the tree from $v$.)
Otherwise we compute all cells of $\reals^3\setminus Z(F_v)$ that $\gamma_v$ 
intersects.
Since $\gamma$ is a connected portion of an irreducible component of a 
parametrized algebraic curve of constant degree, $F_v$ is a polynomial of degree at most $c_1D$, 
and $\gamma \not\subset Z(F_v)$, we conclude that $\gamma \cap Z(F_v)$ consists of  $O(D)$ points~\cite{BPR}, and therefore
$\gamma$ crosses at most $c_5 D$ cells, for some constant $c_5>0$.
Let $\tau$ be such a cell. We first use the secondary data structure $\WDS_\tau$
to detect whether $\gamma_v$ intersects any plate of~$\W_\tau$, the set of wide plates at $\tau$.
Again, if we detect an intersection, we stop. Otherwise we recursively query at the child $w_\tau$ to detect an intersection between $\gamma$ and $\T_\tau$, the set of
narrow plates at $\tau$.

For intersection-detection queries, the query procedure stops as soon as an intersection between $\gamma$ and 
$\T$ is found.
For reporting/counting  queries (or more generally, semi-group queries), we follow the above recursive scheme, and 
at each node $v$ visited by the query
procedure, we either report all the plates of $\T_v$ intersected by the (appropriate portion of the) query arc, or add up 
the intersection counts returned by the various secondary structures and recursive calls.
By our general-position assumption, if the query arc $\gamma$ is planar then there might be $O(1)$ plates 
whose supporting plane may contain $\gamma$, and none otherwise. 
These plates are either detected at the leaves of $\T$ or at the secondary structures. We keep track of these plates,
compute their intersections with $\gamma$, and report/count these intersections.
Since we clip $\gamma$ at each node $v$ within $\tau_v$, we note that each intersection point of $\gamma$ with an 
input plate $\Delta$ is reported/counted exactly once.

Denote by $Q(m)$ the maximum query time for a subproblem involving at most $m$ plates. 
Then $Q(m) = O(m)$ for $m \le n_0$. For $m>n_0$, Lemmas~\ref{prop:wide} and~\ref{lem:onzf} imply 
that the query time of the auxiliary data structures for subproblems of size $m$ is $O^*(m^{2/3})$. 
Therefore $Q(m)$ obeys the recurrence:
\begin{equation}
  \label{eq:query}
  Q(m) \le
  \begin{cases*}
	  c_5 D\cdot Q(m/D^2) + c_6 m^{2/3+\delta} & for $m> n_0$, \\[1mm]
    	  c_7 m & for $m\le n_0$,
  \end{cases*}
\end{equation}
where $c_5$ is the constant as defined above, $\delta>0$ as above is an arbitrarily 
small constant, $c_6$ is a constant that depends on $\delta$, and $c_7>0$ is an absolute constant. 
We claim that the solution to the recurrence is 
\begin{equation}
	\label{eq:qtime-bound}
	Q(m) \le B m^{1/2+\eps}n^{1/6} 
\end{equation}
for any constant $\eps\ge\delta$, where $B>\max\{2c_6, c_7\}$ is a sufficiently large constant and $n\ge m$ is the size of the original problem. 
Since $n_0=n^{1/3}$, for $m \le n_0=n^{1/3}$, we have
\[
B m^{1/2+\eps}n^{1/6} \ge B m^{1/2+\eps}m^{1/2} \ge c_7 m,
\] 
implying the claim for $n\le n_0$. For $n > n_0$,
using induction hypothesis and plugging~\eqref{eq:qtime-bound} into~\eqref{eq:query}, we obtain
\begin{align*}
	Q (m) &\le c_5 D \cdot B \left (\frac{m}{D^2}\right)^{1/2+\eps} n^{1/6}+c_6m^{2/3+\delta} \\
	&\le Bm^{1/2+\eps}n^{1/6} \left [ \frac{c_5}{D^{2\eps}} + \frac{c_6}{B} \frac{m^{1/6+\delta-\eps}}{n^{1/6}} \right ]\\
	&\le Bm^{1/2+\eps}n^{1/6} \left [ \frac{c_5}{D^{2\eps}} + \frac{c_6}{B} m^{\delta-\eps} \right ] 
			\quad \mbox{(because $m\le n$)}\\
		&\le  Bm^{1/2+\eps}n^{1/6},
\end{align*}
provided we choose $D \ge (2c_5)^{1/2\eps}$,
$B\ge 2c_6$, and $\delta\le\eps$. This establishes the induction step and thus yields $Q(n) = O^*(n^{2/3})$.
Putting everything together we obtain:

\begin{theorem}
  \label{thm:trimain}
  Let $\T$ be a set of $n$ plates of constant complexity in $\reals^3$ in general position, and let $\Gamma$ be a family
  of parametrized algebraic arcs of constant degree. $\T$ can be preprocessed, in expected time $O^*(n^{4/3})$, 
  into a data structure of size $O^*(n^{4/3})$, so that an arc-intersection query with an arc of $\Gamma$
  amid the plates of $\T$ can be answered in $O^*(n^{2/3})$ time, or in $O^*(n^{2/3}) + O(k)$ time for reporting queries, where $k$ is the output size. The constants of proportionality hiding
  in these bounds depend on the degree of the arcs of $\Gamma$ and on the complexity of the plates of $\T$.
\end{theorem}

\begin{remark*}
If we do not assume $\T$ to be in general position, the set $\T_v^0$ could be arbitrarily large.
The above machinery would then require a data structure to answer a point-enclosure or arc-intersection 
query on $\T_v^0$ in $O^*(n^{2/3})$ time using linear space. 
The recent lower bounds of~\cite{AC21,AC22} suggest that such a data structure is infeasible for general algebraic 
arcs. 
\end{remark*}

\subsection{Improving the storage slightly}
\label{subsec:slight-improv-size}

As mentioned in the introduction, if the reduced parametric dimension of $\T$ 
is $\tO$ then 
using a multi-level data structure based on the partition-tree technique by Matou\v{s}ek and Pat\'akov\'a~\cite{MP}, $\T$ 
can be preprocessed, in $O^*(n)$ time, into a data structure of size $O^*(n)$, so that an arc-intersection query 
can be answered in $O^*(n^{1-1/\tO})$ time; see Appendix~\ref{app:multi-level} for the details. 
Using this data structure, we can modify our main structure $\Psi$, as follows: 
Assume that $\tO \ge 3$, and set $\alpha \coloneqq \tfrac{1}{3(\tO -2)}$ and $n_0 \coloneqq n^{1/3+2\alpha}$,
i.e., a node $v$ is a leaf if $n_v\le n_0$.
We construct a  near-linear-size partition tree on $\T_v$, as described in Appendix~\ref{app:multi-level} (see Theorem~\ref{lem:Pi-lin}),
 at each leaf $v$ of $\Psi$ that answers an arc intersection query in $O^*(n_v^{1-1/\tO})$ time.
The recurrence for storage remains the same except that we now have the previously stated new value of $n_0$. 
The solution to the recurrence~\eqref{eq:storage}, with the new value of $n_0$, is easily seen to be
\[ 
S(m) \le B \left ( \frac{m^{3/2+\eps}}{n^{1/6+\alpha}} + m \right ) ,
\]
for any $\eps>0$, where $B$ depends on $\eps$.
Hence, the overall size of the data structure becomes $O^*(n^{4/3-\alpha})$.

The recurrence for the query time is now
\begin{equation}
  \label{eq:imp-query}
  Q(m) \le
  \begin{cases*}
	  c_5 D \cdot Q(m/D^2) + c_6 m^{2/3+\delta} & for $m\ge n_0$, \\[1mm]
	  c_7 m^{1-1/\tO +\delta} & for $m\le n_0$ ,
  \end{cases*}
\end{equation}
for an arbitrarily small $\delta>0$ where the constant $c_7$ depends on $\delta$. The solution to this recurrence is still
$$Q(m) \le B m^{1/2+\eps}n^{1/6}$$
for any constant $\eps\ge\delta$
and $B \ge \max\{c_7,2c_6\}$. Indeed for $m \le n_0$, with $\delta\le \eps$,
\[
Q(m) \le c_7 m^{1-1/\tO +\delta} \le B m^{1/2+\eps} n_0^{1/2-1/\tO }  
	= B m^{1/2+\eps}n^{\bigl(\tfrac{1}{3}+2\alpha \bigr) \bigl(\tfrac{1}{2}-\tfrac{1}{\tO } \bigr)} .
\]
By our choice of $\alpha$, the exponent of $n$ in the above inequality becomes
\[
\left (\frac{1}{3}+\frac{2}{3(\tO -2)}\right ) \left( \frac{1}{2}-\frac{1}{\tO } \right ) = \frac{1}{6},
\]
implying that the claim holds for $m \le n_0$. For $m>n_0$, we follow the same analysis as above
(as is easily verified, this part of the analysis is independent of the choice of $n_0$).
Hence, the overall query time remains  $Q(n) = O^*(n^{2/3})$.

We show in Section~\ref{sec:linst} that the reduced parametric dimension of triangles
is (at most) $5$ when $\Gamma$ is a set of algebraic arcs of constant complexity
(and it reduces to $4$ if $\Gamma$ is a set of lines~\cite{AM94,ShSh,MoSh}), even though
one needs $9$ parameters to specify a triangle in $\reals^3$. 
This immediately leads to an $O^*(n)$-size data structure with $O^*(n^{4/5})$ query time. 
Plugging this bound into~\eqref{eq:imp-query} and observing that $\alpha=1/9$ in this case, we obtain a data structure 
of size $O^*(n^{11/9})$ with $O^*(n^{2/3})$ query time for arc-intersection queries amid triangles.

\begin{theorem}
  \label{thm:trimainx}
Let $\Gamma$ be a family of algebraic arcs of constant parametric dimension, and let $\T$ be a set 
of $n$ plates in~$\reals^3$ of reduced parametric dimension $\tO \ge 3$ (with respect to $\Gamma$). 
$\T$ can be preprocessed, in expected time $O^*(n^{4/3-\alpha})$, into a data structure of size 
$O^*(n^{4/3-\alpha})$, where $\alpha = \frac{1}{3(\tO -2)}$, so that an
arc-intersection query amid the triangles of $\T$ can be answered in $O^*(n^{2/3})$ time.
If $\T$ is a set of $n$ triangles in $\reals^3$, then $\tO \le 5$ and thus the size and the 
expected preprocessing time are $O^*(n^{11/9})$.
\end{theorem}
\section{Handling Wide Plates}
\label{sec:wide}

Let $\T$ be a set of $n$ plates in $\reals^3$, $\Gamma$ a parametrized family of algebraic arcs, and $F$ a 
partitioning polynomial, as described in Section~\ref{sec:main}. This section describes the algorithm for 
preprocessing the set of wide plates, $\W_\tau$, for each cell $\tau$ of $\reals^3\setminus Z(F)$, for intersection 
queries with arcs of $\Gamma$.
Fix a cell $\tau$. 
Our goal is to write the intersection predicate between $\W_\tau$ and $\Gamma$ in a way so that it does not 
depend on the boundary arcs of the wide plates. It is tempting to replace a wide plate 
$\plate\in\W_\tau$ with the plane $h_\plate$ supporting $\plate$ and argue that $\gamma \cap \tau$, 
for a query arc $\gamma \in \Gamma$, intersects $\plate$ if and only if it intersects $h_\plate$. 
But this does not necessarily hold because $\tau$ is non-convex, so we proceed in a more careful manner, as follows.
Since $\plate$ is wide at $\tau$, each connected component of $\plate\cap\tau$ 
is also a connected component of $h_\plate\cap\tau$ (though  some connected components of $h_\plate\cap\tau$ 
may be disjoint from $\plate$). 
By a careful construction of a
\emph{cylindrical algebraic decomposition} (CAD) $\CAD$ in a five-dimensional parametric space 
(presented in detail in Section~\ref{subsec:usecad}),
we decompose $\plate\cap\tau$ into $O(1)$ pseudo-trapezoids, each of which has constant complexity
and is contained in a single connected component of $\plate\cap\tau$. 
We collect these pseudo-trapezoids of all wide plates at $\tau$ and cluster them, using $\CAD$,
into $O(1)$ families, so that each cluster $\fragset_C$ is defined by a cell $C$ of $\CAD$ and 
so that an arc-intersection query amid $\fragset_C$ can be reduced to
a three-dimensional semi-algebraic range query. See Sections~\ref{subsec:plate-partition} and~\ref{subsec:rs}. 
\subsection{An overview of cylindrical algebraic decomposition} 
\label{subsec:CAD}

We begin by giving a brief overview of \emph{cylindrical algebraic decomposition} (CAD), 
also known as Collins' decomposition, after its originator Collins~\cite{Col}. 
A detailed 
description can be found in \cite[Chapter 5]{BPR}; a possibly more accessible 
treatment is given in~\cite[Appendix A]{SS2}.

Let $\FS = \{f_1, \ldots, f_s\}$ be a finite set of $d$-variate polynomials. 
The \emph{arrangement} of $\FS$, denoted by $\A(\FS)$, is the decomposition of $\reals^d$ into 
maximal connected relatively open pairwise disjoint cells of all dimensions, so that all points within a cell have the 
same number of real roots of each polynomial $f_i \in \FS$. For another polynomial $g$, let 
$\A(\FS; g)$ be the arrangement $\A(\FS\cup\{g\})$ restricted to $Z(g)$, i.e., the collection of the cells of 
$\A(\FS\cup\{g\})$ contained in $Z(g)$, the zero set of $g$. If $\F=\{F\}$, we simply 
use the notation $\A(F)$ and $\A(F;g)$.

A cylindrical algebraic decomposition induced by $\FS$, denoted by $\CAD(\FS)$,
is a (recursive) decomposition of $\reals^d$ into a finite collection of relatively 
open simply-shaped semi-algebraic cells of dimensions $0,\ldots, d$, each homeomorphic 
to an open ball of the respective dimension. $\CAD(\FS)$ is a  refinement of the arrangement $\arr(\F)$.

Set $F=\prod_{i=1}^s f_i$. For $d=1$, let $\alpha_1 < \alpha_2 < \cdots < \alpha_t$ 
be the distinct real roots of~$F$. Then $\CAD(\FS)$ is the collection of cells 
$\{(-\infty,\alpha_1), \{\alpha_1\}, (\alpha_1,\alpha_2), \ldots, \{\alpha_t\}, (\alpha_t, +\infty)\}$. 
For $d>1$, regard $\reals^d$ as the Cartesian product $\reals^{d-1}\times\reals$. 
For simplicity of the description, here we
assume that $x_d$ is a \emph{good} direction, meaning that for any fixed $a\in\reals^{d-1}$, 
$F(a,x_d)$, viewed as a polynomial in $x_d$, has finitely many roots. 
The good-direction assumption is not needed if the recursive construction of the CAD is defined more carefully, as 
in~\cite[Chapter~5]{BPR}, but it can be made without loss of generality.

$\CAD(\FS)$ is defined recursively from a ``base'' $(d-1)$-dimensional CAD $\CAD_{d-1}$, 
as follows. One constructs a suitable set $\EP \coloneqq \EP(\FS)$ of the polynomials in $x_1,\ldots,x_{d-1}$ 
(denoted by~$\textsc{Elim}_{X_k}(\FS)$ in \cite{BPR} and by~$Q_b$ in \cite{SS2}). Roughly 
speaking, the zero sets of polynomials in $\EP$, viewed as subsets of $\reals^{d-1}$,
contain the projection onto $\reals^{d-1}$ of all intersections $Z(f_i) \cap Z(f_j)$, 
$1 \le i < j \le s$, as well as the projection of the loci in each $Z(f_i)$ where $Z(f_i)$ 
has a tangent hyperplane parallel to the $x_d$-axis, or a singularity of some kind. 
The actual construction of $\EP$, based on \emph{subresultants} of $\FS$, is somewhat 
complicated, and we refer to \cite{BPR,SS2} for more details.

One recursively constructs $\CAD_{d-1}=\CAD(\EP)$ in $\reals^{d-1}$, which is a refinement 
of $\A(\EP)$ into topologically trivial open cells of dimensions $0,1,\ldots,d-1$. 
For each cell $\cell\in \CAD_{d-1}$, the sign of each polynomial in $\EP$ is constant 
(zero, positive, or negative) and the (finite) number of distinct real $x_d$-roots of $F(\xx,x_d)$ 
is the same for all $\xx\in\cell$.  Moreover, each of these roots varies continuously with~$\xx\in\cell$.  $\CAD(\FS)$ is then defined in terms of $\CAD_{d-1}$, as follows. 
Fix a cell $\cell \in \CAD_{d-1}$. Let $\cell \times \reals$ denote the \emph{cylinder} 
over $\cell$. There is an integer $t\ge 0$ (depending on $\cell$) such that for all $\xx\in\cell$, there are exactly 
$t$ distinct real roots $\psi_1 (\xx) < \cdots < \psi_t(\xx)$ of $F(\xx,x_d)$ (regarded 
as a polynomial in $x_d$), and these roots are algebraic functions that vary continuously 
with~$\xx\in\tau$. Let $\psi_0, \psi_{t+1}$ denote the constant functions $-\infty$ and $+\infty$, 
respectively. Then we create the following cells that decompose the cylinder over $\cell$:

\begin{itemize}
\item 
$\sigma = \{(\xx,\psi_i(\xx)) \mid \xx\in\cell\}$, for $i=1, \ldots, t$;
$\sigma$ is a section of the graph of $\psi_i$ over $\cell$, and
\item 
$\sigma = \{(\xx, y) \mid \xx\in\cell,\; y\in (\psi_i(\xx), \psi_{i+1}(\xx))\}$,
for $0 \le i \le t$; $\sigma$ is a portion (``layer'') of the cylinder $\cell\times\reals$ 
between the two consecutive graphs $\psi_i$, $\psi_{i+1}$.
\end{itemize}

\medskip
The main property of $\CAD$ is that, for each cell $\cell\in \CAD$, the sign of each polynomial 
in $\FS$ is constant for all $\xx\in\cell$. 
Omitting all further details (for which see \cite{BPR,Col,SS2}), we have the following lemma:

\begin{lemma} \label{lem:CAD}
Let $\FS=\{f_1,\ldots,f_s\}$ be a set of $s$ $d$-variate polynomials of degree at 
	most~$D$ each. Then, assuming that all coordinates are good directions (with respect to the corresponding sets of polynomials),
$\CAD(\FS)$ consists of $O(Ds)^{2^d}$ cells, and each cell can be represented 
semi-algebraically by $O(D)^{2^{d}}$ polynomials of degree at most $O(D)^{2^{d-1}}$. 
$\CAD(\FS)$ can be constructed in time $(Ds)^{2^{O(d)}}$ in a suitable standard model
of algebraic computation.
\end{lemma}
\subsection{Constructing a CAD of the partitioning polynomial} 
\label{subsec:usecad}

Let $\ES_3$ denote the space of all (non-vertical) planes in $\reals^3$. More precisely,
$\ES_3$ is the (dual) three-dimensional space where each plane $h \colon z=ax+by+c$ is mapped to the
point $(a,b,c)$. For a point $\xi=(a_0,b_0,c_0)\in \ES_3$, we use $h_\xi$ to denote the plane $z = a_0x+b_0y+c_0$. With a slight abuse of notation, we will also use $h_\xi$ to denote the linear function $a_0x+b_0y+c_0$ as well as the linear 
polynomial $z-a_0x-b_0y-c_0$. It will be clear from the context what $h_\xi$ is referring to.
We consider the five-dimensional parametric space $\ES \coloneqq \ES_3\times\reals^2$ with coordinates $(a,b,c,x,y)$.
Let $F \in \reals[x,y,z]$ be a partitioning polynomial constructed in Section~\ref{sec:main}.
We construct a CAD of $\ES$ induced by a single five-variate polynomial $\hat{F} \in \reals[a,b,c,x,y]$, defined by $\hat{F} (a,b,c,x,y) \coloneqq F(x,y,ax+by+c)$.
Roughly speaking, $\hat F$ encodes the polynomial $F$ restricted to any plane in $\reals^3$ in the sense that for any $\xi\in\ES_3$, $\hat F(\xi,x,y)$ is the function $F$ restricted to the plane $h_\xi$.

The construction of the CAD recursively eliminates the variables in the
order $y,x,c,b,a$. That is, unfolding the recursive definition given in
Section~\ref{subsec:CAD}, each cell of the CAD is given by a sequence
of equalities or inequalities (one from each row) of the form:
\begin{align}
  \label{eq:cad}
a &= a_0 & \text{or}&& a_0^- &< a <  a_0^+ \nonumber \\
b &= f_1(a) & \text{or} && f_1^-(a) &< b < f_1^+(a) \nonumber \\
c &= f_2(a,b) & \text{or} && f_2^-(a,b) &< c < f_2^+(a,b) \\
x &= f_3(a,b,c) & \text{or} && f_3^-(a,b,c) &< x < f_3^+(a,b,c) \nonumber \\
y &= f_4(a,b,c;x) & \text{or}&& f_4^-(a,b,c;x) &< y < f_4^+(a,b,c;x) \nonumber ,
\end{align}
where $a_0$, $a_0^-$, $a_0^+$ are real parameters, and $f_1,f_1^-,f_1^+,\ldots,f_4,f_4^-,f_4^+$
are constant-degree continuous algebraic functions (any of which can be $\pm\infty$), 
so that, whenever we have an inequality involving two reals or two functions,
we then have $a_0^- < a_0^+$, and/or $f_1^-(a) < f_1^+(a)$, $f_2^-(a,b) < f_2^+(a,b)$,
$f_3^-(a,b,c) < f_3^+(a,b,c)$, and $f_4^-(a,b,c;x) < f_4^+(a,b,c;x)$, over the cell defined by the preceding 
set of equalities and inequalities in~\eqref{eq:cad}.

We illustrate the structure of this CAD by considering a special case in which
only horizontal planes of the form $z=c$ are considered. Let $\ES_1$ be the one-dimensional space 
of horizontal planes. Set $\ES = \ES_1 \times \reals^2$. We construct a 3-dimensional CAD $\CAD$ of $\ES$ 
induced by the trivariate polynomial $\hat{F}\in \reals[c,x,y]$ with $\hat{F}(c,x,y) = F(x,y,c)$.
$\CAD$ induces a partition $\CAD_1$ of $\ES_1$ into intervals and delimiting points. 
For each point $c_0 \in \ES_1$, the cross-section of $\CAD$ over $c_0$, denoted by $\fiber (c_0)$ 
and called the \emph{fiber} of $\CAD$ over $c_0$, is 
the CAD of the $xy$-plane induced by $F(x,y,c_0)$. $\fiber(c_0)$ is a refinement of (the projection of) 
$\A(F;h_{c_0})$ into pseudo-trapezoids, where $h_{c_0}\colon z=c_0$. Each pseudo-trapezoid of~$\fiber(c_0)$
is given by a simpler version of the 
set of the last two equations or inequalities in~\eqref{eq:cad}.
As we vary $c_0$, the combinatorial structure of~$\fiber(c_0)$ remains the same as long as $c_0$ lies in the same interval $\gamma$ of the partition $\CAD_1$ of $\ES_1$.
In other words, the topology of~$\fiber(c_0)$ does not change as $c_0$ varies within $\gamma$. 
The combinatorial structure of the fiber changes at a delimiting endpoint of $\CAD_1$,
which implies a change in the topology of the fiber. 
Readers familiar with Morse's theory~\cite{Morse} should note the close relationship between the breakpoints of $\CAD_1$ and the critical points of a Morse function defined over $Z(F)$ that gives the $z$-value of each point of $Z(F)$.

\begin{figure}[htb]
  \centering
  \scalebox{0.8}{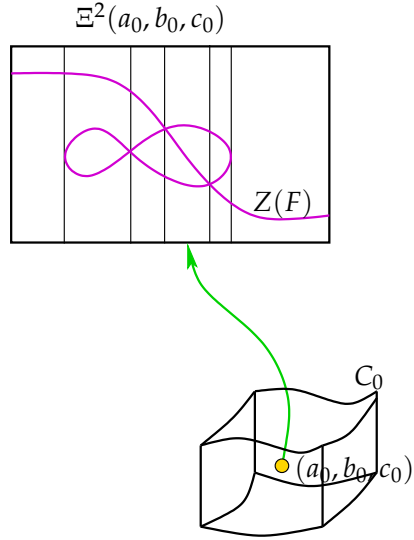}
  \caption{An illustration of the CAD construction. $C_0$ is a three-dimensional cell of $\CAD_3$. 
    For a point $(a_0,b_0,c_0)\in C_0$, its two-dimensional fiber $\fiber(a_0,b_0,c_0)$ is shown.
    Formally, the purple curve is the $xy$-projection of $Z(F)\cap h(a_0,b_0,c_0)$.} 
  \label{fig:cadfiber}
\end{figure}

Returning to the construction of the CAD for the general case of all non-vertical planes, 
let $\CAD_5 = \CAD_5(F)$ denote the five-dimensional CAD just
defined. Let $\CAD_3$ denote the projection of $\CAD_5$ onto $\ES_3$, which
we refer to as the \emph{base} of $\CAD_5$ and  which itself is a CAD of a suitable set of polynomials
in $a,b,c$. Each base cell of $\CAD_3$ is given by a set of 
equalities and inequalities from the first three rows of~\eqref{eq:cad}, one per row. 
For a cell $C\in\CAD_5$, let $C^\downarrow\in\CAD_3$ denote the \emph{base cell} of $C$, the projection of $C$ onto $\ES_3$.

For a point $\xi=(a_\xi, b_\xi, c_\xi)\in \ES_3$, let $\fiber(\xi)$ 
denote the cross-section of $\CAD_5$ over $\xi$, which is a
decomposition of the $xy$-plane into pseudo-trapezoids induced by $\CAD_5$ over 
$\xi$. In fact, $\fiber(\xi)$  is a CAD of the $xy$-plane induced by the bivariate polynomial 
$F_\xi (x,y) \coloneqq F(x,y,h_\xi(x,y))$.  
We refer to $\fiber(\xi)$ as the two-dimensional \emph{fiber} of $\CAD_5$ over $\xi$.
Each pseudo-trapezoid of $\fiber(\xi)$ is specified by equalities and/or inequalities from the last 
two rows of~\eqref{eq:cad}, with $a=a_\xi$, $b=b_\xi$, $c=c_\xi$. 
For a cell $C\in\CAD_5$ and for a point $\xi\in C^\downarrow$, let $C(\xi)$ denote the cross-section of $C$ over $\xi$, i.e., $C(\xi)$ is the pseudo-trapezoid in $\fiber(\xi)$ corresponding to the cell $C$.

The \emph{lifting} of $\fiber(\xi)$ to the plane $h_\xi$, denoted by $\fiber^\uparrow (\xi)$, is defined as 
lifting of each pseudo-trapezoid $\frag\in\fiber(\xi)$ to 
$\frag^\uparrow = \{(x,y,h_\xi(x,y)) \mid (x,y)\in\frag\}$. $\fiber^\uparrow(\xi)$ is a CAD of $h_\xi$ induced by 
$F$, and thus a refinement of the planar arrangement $\A(F;h_\xi)$ into pseudo-trapezoids 
(i.e., each pseudo-trapezoid of $\fiber^\uparrow(\xi)$ lies in a cell of $\A(F;h_\xi)$).
See Figure~\ref{fig:cadfiber} for an illustration.

As in the example mentioned above, the combinatorial structure of $\fiber(\xi)$, as well as of its lifting 
$\fiber^\uparrow(\xi)$, is the same for all points $\xi$ in a base cell $\psi \in \CAD_3$.
It changes only when we move from one base cell to another cell of $\CAD_3$. 
Hence, each cell $C$ of $\CAD_5$ can be associated with a fixed cell of $\A(F)$,
denoted as $\tau_C$, such that for all points $\xi$ in the base cell $C^\downarrow\in \CAD_3$, $C^\uparrow (\xi)$, 
the lifting of $C(\xi)$ to $h_\xi$, is a pseudo-trapezoid of $\fiber^\uparrow (\xi)$ that lies in $\tau_C$.
For a cell $\tau \in \A(F)$, let $\CAD_\tau \coloneqq \{ C\in\CAD \mid \tau_C=\tau\}$ be the subset of CAD cells associated with $\tau$.

\begin{figure}[htb]
  \centering
  \scalebox{0.7}{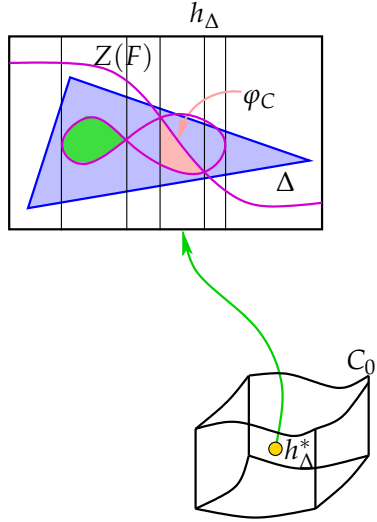}
  \caption{The encoding scheme provided by the CAD (the plate depicted in this figure is a triangle).
    The cell $C$ labels, by an explicit semi-algebraic expression,
    the highlighted inner pseudo-trapezoidal subcell $\varphi_C$ within the plate $\plate$.
    Another inner subcell, with a different label, in a different partition cell $\tau$, is also highlighted.
  }
  \label{fig:cadlabel}
\end{figure}

We conclude this discussion with the following crucial observation, which is the main rationale for the CAD construction:
The semi-algebraic representation of the cell $C\in\CAD$ provides 
a \emph{fixed} constant-size \emph{semi-algebraic encoding} of  the pseudo-trapezoids $C^\uparrow (\xi)$, 
for all $\xi\in C^\downarrow$.
Namely, each such pseudo-trapezoid $C^\uparrow(\xi)$
is represented by equalities and inequalities of the form
\begin{equation}
	\label{eq:encoding}
	\begin{array}{rclcc}
		x &=& f_3(\xi)&\quad \mathrm{or}\quad  &f_3^-(\xi) < x < f_3^+(\xi);\\
		y &=& f_4(\xi,x)&\quad \mathrm{or}\quad & f_4^-(\xi,x) < y < f_4^+(\xi,x);\\
		z &=&f_5(\xi,x,y).
	\end{array}
\end{equation}
Here $f_3, f_3^-, f_3^+, f_4, f_4^-, f_4^+$ are fixed constant-degree continuous algebraic 
functions over the corresponding domains, as in~\eqref{eq:cad}, and $f_5(\xi,x,y) = h_\xi(x,y) = ax+by+c$,
where $\xi = (a,b,c)$. We note that these functions are the same for all pseudo-trapezoids $C(\xi)$, 
$\xi\in C^\downarrow$, and thus the encoding 
is \emph{independent}\footnote{%
  More precisely, its dependence on $\xi$ is only in terms of its coordinates being 
  substituted in the fixed semi-algebraic predicate given above.}
of $\xi$; see Figure~\ref{fig:cadlabel}.

\subsection{Decomposing wide plates into pseudo-trapezoids}
\label{subsec:plate-partition}

We describe the decomposition of $\plate\cap\tau$ into pseudo-trapezoids, induced by the CAD,
and the clustering of the resulting pseudo-trapezoids.
For a plate $\plate\in\T$, let $\plate^*$ denote the point in the $abc$-subspace 
$\ES_3$ dual to the plane $h_\plate$ supporting $\plate$.

Let $\cell$ be a cell of~$\reals^3\setminus Z(F)$,  and
let $\plate \in \W_\tau$ be a plate that is wide at $\cell$, and let $\psi\in\CAD_3$ be the base cell containing 
$\plate^*$. Recall that $\fiber^\uparrow(\plate^*)$ is the lifting of the fiber $\fiber(\plate^*)$ to $h_{\plate}$.
Let $\frag\coloneqq C^\uparrow(\plate^*)$ be a pseudo-trapezoid in $\fiber^\uparrow(\plate^*)$ 
that is contained in $\tau$ (i.e., $C^\downarrow=\psi$ and $C\in \CAD_\tau$).
Since $\plate$ is wide at $\tau$, either $\frag\subseteq\plate$ or $\frag\cap\plate=\emptyset$. Let 
$\fragset_{\plate,\tau}\subseteq\fiber^\uparrow(\plate^*)$ be the subset of such pseudo-trapezoids that are contained 
in $\tau\cap\plate$.
That is,
\[
\fragset_{\plate,\tau} \coloneqq \{ C^\uparrow(\plate^*) \mid \plate^*\in C^\downarrow,\, 
		C\in\CAD_\tau,\, C^\uparrow(\plate^*) \subset \tau\cap\plate\} .
\]
$\fragset_{\plate,\tau}$ is a decomposition of $\plate\cap\tau$ into pseudo-trapezoids. Set
$\fragset_\tau = \bigcup_{\plate\in \W_\tau} \fragset_{\plate,\tau}$.
An arc $\gamma \in\Gamma$ intersects a wide plate $\plate\in\W_\tau$ within $\tau$ if and only if it intersects a 
pseudo-trapezoid of $\fragset_{\plate,\tau}$. Hence an intersection query with $\gamma$ on $\W_\tau$ 
(within $\tau$) reduces to an intersection query in $\fragset_\tau$. To facilitate the latter task, 
we compute a clustering of $\fragset_\tau$ into $O(1)$ clusters, where the constant depends on the degree $D$ of the partitioning polynomial $F$,
 and build a separate data structure 
for each cluster. Roughly speaking, all pseudo-trapezoids of $\fragset_\tau$ corresponding to a single 
cell $C$ of $\CAD_5$ form one cluster $\fragset_C$. More precisely, for each cell $C\in \CAD_\tau$, 
we define $\fragset_C \subseteq\fragset_\tau$ to be
\[
\fragset_C =\{ C^\uparrow (\plate^*) \mid \plate \in \W_{\tau} \; \wedge \;
C^\uparrow (\plate^*) \in \fragset_{\plate,\tau} \} .
\]
By definition, $\fragset_\tau = \bigcup_{C\in\CAD_\tau} \fragset_C$.
As mentioned above, a crucial property of $\fragset_C$ is that all of its pseudo-trapezoids have 
the same constant-complexity \emph{semi-algebraic encoding} of the form described in~\eqref{eq:encoding}.
The coefficients in the polynomials of the encoding depend only on $F$ 
(and not on the parameters of the plates containing these pseudo-trapezoids).
Some of the variables in the polynomials in~\eqref{eq:encoding}
encode the coefficients of the plane  containing a pseudo-trapezoid, 
but none of the variables encode the plate-boundary parameters.
The latter property will be 
crucial in constructing the data structure for $\fragset_C$.
\subsection{Reduction to semi-algebraic predicate queries}
\label{subsec:rs}

For each cell $C\in\CAD_5$, we build a data structure  $\Sigma_C$ to answer an arc-intersection query on 
$\fragset_C$ with an arc $\gamma \in \Gamma$ that is contained in (i.e., clipped to within) the cell 
$\tau_C$ of $\reals^3\setminus Z(F)$. Set $n_C = |\fragset_C|$.
To this end, we fix a cell $C$ of $\CAD_5$ and
define a predicate
$\Pi_C\colon \Gamma \times \ES_3 \rightarrow \{0,1\}$ 
that is $1$ for a pair $(\gamma,\xi) \in\Gamma \times \ES_3$ if and only if $\xi\in C^\downarrow$ and 
an intersection point of $\gamma$ and $h_\xi$ lies in the pseudo-trapezoid $C^\uparrow (\xi)$, i.e.,
\begin{equation}
	\label{eq:intersect-pred}
	\Pi_C(\gamma; \xi) = \left \{
		\begin{array}{ll} 
			1 & \mbox{if } \xi\in C^\downarrow \wedge 
			\exists  (x_p,y_p,z_p)\in \gamma \cap h_\xi \, \mbox{s.t.}\, (x_p,y_p)\in C(\xi), \\[1mm]
			0 & \mbox{otherwise.} 
		\end{array}
		\right .
\end{equation}
By construction, if $(x_p,y_p)\in C(\xi)$ then $(\xi,x_p,y_p) \in C$
and $(x_p,y_p,z_p)\in C^\uparrow(\xi) \subseteq \tau_C$, where $z_p = a_\zeta x_p + b_\zeta y_p + c_\zeta$.
Since $C$ is a semi-algebraic set of constant complexity, 
$\Pi_C(\gamma; \xi)$ is a semi-algebraic predicate of constant complexity (the complexity depends on $D$ and the 
parametric dimension of $\Gamma$). 
We refer to $\Pi_C$ as a \emph{$C$-intersection predicate}.

The following lemma, which follows from the construction, lies at the heart of the data structure:
\begin{lemma}
	For a pseudo-trapezoid $\frag\in\fragset_C$, let $\frag^*\in\ES_3$ be the point dual to the plane supporting $\frag$. An arc $\gamma\in \Gamma$ crosses a pseudo-trapezoid $\frag\in\fragset_C$ if and only if $\Pi_C(\gamma,\frag^*)=1$.
\end{lemma}

\smallskip
\begin{remark}
\label{rem:enhanced-pred}
  The semi-algebraic predicate $\Pi_C$ can be replaced by $b$ predicates $\Pi_C^{(i)}$, 
  for $i=1,\dots,b$, where $b$ is the maximum number of intersections of a query arc 
  with a plane ($b$ is at most the maximum degree of the arcs of $\Gamma$), 
  so that $\Pi_C^{(i)}(\gamma; \xi)$ asserts that (is equal to $1$ when)
  $\gamma$ intersects $h_\xi$ at least $i$ times and the $i$-th intersection point 
  along $\gamma$ (here we assume that $\gamma$ is directed) belongs to the pseudo-trapezoid $C(\xi)$.
  These predicates, which are formed using quantifiers that can then be eliminated,
  are also of constant complexity, albeit of larger complexity than $\Pi_C$.
  This enhancement is used for answering intersection-counting queries as well as for answering 
  intersection queries with planar arcs (see Sections~\ref{subsec:circ_imp} and~\ref{subsec:planar_imp}).
\end{remark}

We preprocess $\fragset_C$, in $O^*(n_C)$ expected time, into a data structure $\Sigma_C$ of size $O^*(n_C)$,
as described in Appendix~\ref{subsec:Pi-lin}, for answering $\Pi_C$-queries. For any $\frag\in\fragset$, since the predicate $\Pi_C$ depends only on the three 
parameters that define the plane supporting $\frag$, $\tO$, the reduced parametric dimension of $\fragset_C$, is $3$. Therefore, by Theorem~\ref{lem:Pi-lin}, the query time is $O^*(n_C^{2/3})$.

We construct $\Sigma_C$ for every cell $C\in\CAD_5$.
For a cell $\tau$ of $\reals^3\setminus Z(F)$, 
we store at $\tau$ the structures $\Sigma_C$, for all $C\in\CAD_\tau$, as the secondary structure $\WDS_\tau$. 
To test whether an arc $\gamma\in\Gamma$, which lies 
inside $\tau$, intersects a plate of $\W_\tau$, we query each of the structures $\Sigma_C$ stored at 
$\tau$ with $\gamma \cap \tau$ and return yes if any of them returns yes.  
Putting everything together, we obtain the following result.
\begin{lemma}
  \label{prop:wide}
  A set $\W$ of $n$ wide plates at some cell $\tau$ of $\reals^3\setminus Z(F)$ can be preprocessed into 
  a data structure of size $O^*(n)$, in $O^*(n)$ expected time, so that an arc-intersection 
  query on $\W$, for intersections within $\tau$, can be answered in $O^*(n^{2/3})$ time.
\end{lemma}
\section{Space/Query-Time Trade-Offs for Wide Plates}
\label{sec:wide-tradeoff}

In this section we show that the query time of arc-intersection searching amid wide plates can be improved by increasing
the size of the data structure. 
Let $\T, \Gamma, F$ be as defined above.
Let $\tau$ be a cell of $\reals^3\setminus Z(F)$, and let
$\W_\tau$ be the set of wide plates for $\tau$. 
Following the same approach as in Section~\ref{subsec:rs} but using Theorem~\ref{thm:Pi-tradeoff}, we obtain a 
space/query-time trade-off for $C$-intersection queries. Let $t \coloneqq \tQ$ be the reduced parametric dimension of $\Gamma$ with respect to $\Pi_C$. Then for a storage parameter $s\in [n,n^t]$, by Theorem~\ref{thm:Pi-tradeoff},
the query time is 
$O^* \biggl((n/s^{1/t})^{\tfrac{2/3}{1-1/t}}\biggr)$. Putting everything together as above, we obtain the following:
\begin{lemma}
  \label{lem:wide-tradeoff}
  Let $\W$ be a set of $n$ wide plates at some cell $\tau$ of $\reals^3\setminus Z(F)$, let $\Gamma$ 
  be a family of algebraic arcs whose reduced parametric dimension with respect to $\Pi_C$ is a constant $t\ge 3$, 
  and let $s \in [n,n^t]$ 
  be a storage parameter. Then $\W$ can be preprocessed, in $O^*(s)$ expected time,
  into a data structure of size $O^*(s)$, so that an intersection query on $\W$ (within $\tau$) 
  with an arc in $\Gamma$ can be answered in
  $O^*\biggl ((n/s^{1/t})^{\tfrac{2/3}{1-1/t}} \biggr )$
  time.
\end{lemma}

In the remainder of this section, we improve the query time for the case where $\Gamma$ is a family of planar arcs,
by eliminating the effect of the endpoints of the query arcs. That is, $t$ becomes the reduced parametric dimension of the curves supporting the query arcs. 
For simplicity, we first describe this improvement, in Section~\ref{subsec:circ_imp}, 
for circular arcs,
and then extend it to general planar arcs, in Section~\ref{subsec:planar_imp}.

\subsection{The case of circular query arcs}
\label{subsec:circ_imp}

In this subsection we present the improvement in the query time for the case
of circular arcs.
We describe the intersection condition between a circular arc and a 
pseudo-trapezoid of $\Phi_C$, for a cell $C\in\CAD$, as a 
semi-algebraic predicate in which each polynomial inequality uses 
at most six of the eight parameters specifying a query arc $\gamma$, thereby improving the reduced parametric dimension of circular arcs from $8$ to $6$. 

Let $\Gamma$ be the family of all circular arcs in $\reals^3$. For technical reasons that will become clear shortly, 
we assume that each circular arc $\gamma\in\Gamma$ is directed. Let $c_\gamma$ (resp., $\pi_\gamma$) denote the circle 
(resp., plane) containing $\gamma$. For a circle $c_\gamma$, let $\lambda_\gamma$ be the 
minimal point on $c_\gamma$ in the lexicographic order, i.e., if $\pi_\gamma$ is not parallel to the $yz$-plane then 
$\lambda_\gamma$ is the point with the minimum $x$-coordinate on $c_\gamma$; otherwise it is the point 
on $c_\gamma$ with the minimum $y$-coordinate. We partition $c_\gamma$ into two ``canonical'' semi-circles, 
by splitting $c_\gamma$ at $\lambda_\gamma$ and at its antipodal point.
By splitting the query arc $\gamma$ into at most three arcs and querying with each of them separately, 
we can assume that $\gamma$ is fully contained in one of the canonical semi-circles of $c_\gamma$; 
we denote this semi-circle by $\hat\gamma$. Let $p_\gamma, q_\gamma$ (resp. $p_{\hat\gamma}, q_{\hat\gamma}$)
be the  \emph{initial} and \emph{terminal} endpoints of $\gamma$ (resp., $\hat\gamma$). 
Without loss of generality, we assume that $p_{\hat\gamma}=\lambda_\gamma$ and that $\gamma,\hat\gamma$ are oriented 
from $p_\gamma$ toward $q_\gamma$.
We note that once 
$c_\gamma$ is fixed, so are $p_{\hat\gamma}, q_{\hat\gamma}$. As such we do not need two additional 
parameters to specify them, and thus need only six parameters to specify $\hat\gamma$. 
Note that $p_{\hat\gamma}, p_\gamma, q_\gamma, q_{\hat\gamma}$ appear in this order along $\hat\gamma$. 

Let $\frag \coloneqq C^\uparrow(\xi) \in \fragset_C$ be a pseudo-trapezoid corresponding to $C$ in the decomposition 
of $h_\xi$ induced by $\CAD$, where $\xi\in C^\downarrow$ is the point that represents the plane $h_\xi$ supporting 
$\frag$. Without loss of generality, assume that $\frag$ does not lie in the plane $\pi_\gamma$.\footnote{%
  By our general-position assumption, at most $O(1)$ input plates lie in $\pi_\gamma$.
  We can extract these $O(1)$ plates using a simple hash table and test each of them separately whether it intersects $\gamma$.}
We consider two different cases to express the intersection condition of $\gamma$ with $\frag$.

\smallskip 

\noindent
\textbf{\textit{Case I: $p_\gamma, q_\gamma$ lie on the same side of $h_\xi$.}}
Let $h_\xi^+$ be the halfspace of $h_\xi$ containing $p_\gamma, q_\gamma$, and let $h_\xi^-$ be the other halfspace.
Let $\tau_\gamma(p_\gamma)$, $\tau_\gamma(q_\gamma)$ be the tangents to $\gamma$ 
at $p_\gamma$ and $q_\gamma$, respectively, oriented toward $\gamma$ and lying in the plane~$\pi_\gamma$.
Let $\ell$ be the intersection line of $h_\xi$ and $\pi_\gamma$, and let $u_{\gamma,\xi}$ be the normal vector of 
$\ell$ within the plane $\pi_\gamma$, pointing away from $p_\gamma, q_\gamma$ (i.e., pointing toward $h_\xi^-$). 
We say that $\tau_\gamma(p_\gamma)$ (resp., $\tau_\gamma(q_\gamma)$) \emph{points toward} $h_\xi$ if the angle
between $\tau_\gamma(p_\gamma)$ (resp., $\tau_\gamma(q_\gamma)$) and $u_{\gamma,\xi}$ is 
acute; otherwise we say that it \emph{points away} from $h_\xi$. See Figure~\ref{fig:cross}.
The following lemma is the main ingredient of the intersection condition in this case:

\begin{figure}[htb]
  \centering
  \includegraphics[scale=0.6]{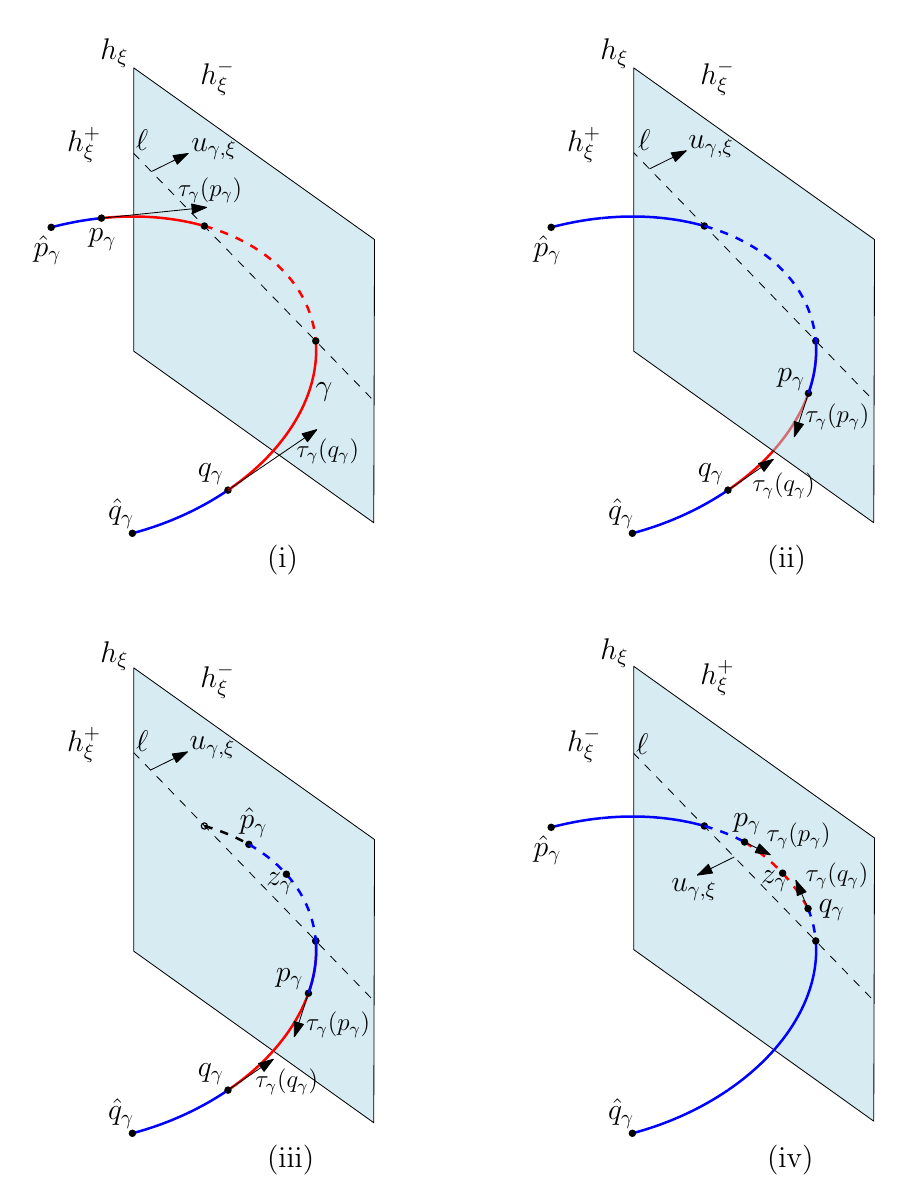}
  \caption{Illustration of the proof of Lemma~\ref{lem:circ-cross}; (i)~$\gamma$ intersects $h_\xi$,
    (ii)~one arc of $\hat\gamma\setminus\gamma$ contains both intersection points of $\hat\gamma\cap h_\xi$,
    (iii)~each of the arcs of $\hat\gamma\setminus\gamma$ contains at most one intersection 
    point of $\hat\gamma\cap h_\xi$ and $z_\gamma\in h_\xi^-$, and (iv)~$z_\gamma \in h_\xi^+$.}
    \label{fig:cross}
\end{figure}

\begin{lemma}
  \label{lem:circ-cross}
	The arc $\gamma$ intersects $h_\xi$ if and only if (a) both $\tau_\gamma(p_\gamma)$ and 
	$\tau_\gamma(q_\gamma)$ point toward $h_\xi$, and 
	(b) $\hat\gamma$ intersects $h_\xi$.
\end{lemma}

\begin{proof}
	Refer to Figure~\ref{fig:cross}(i). Assume first that $\gamma$ intersect $h_\xi$.  Property (b) holds trivially, so it suffices to prove that (a) 
	also holds. Since $p_\gamma$ and $q_\gamma$ lie on the same side of the line 
	$\ell$, $\gamma \cap h_\xi^-$ contains a point $z_\gamma$ at which the tangent to $\gamma$ is 
	parallel to $\ell$ (and thus normal to the vector $u_{\gamma,\xi}$). 
	For a point $z\in\gamma$, let $\overrightarrow{\tau}(z)$ be the tangent of $c_\gamma$ at $z$ oriented in 
	in the same direction as $\gamma$ (i.e., toward $q_\gamma$). For $z=p_\gamma$, 
	$\overrightarrow{\tau}(p_\gamma) = \tau_\gamma(p_\gamma)$.
	Since $z_\gamma \in h_\xi^-$, if we trace $z$ back from $z_\gamma$
	toward the initial point $p_\gamma$, the angle between $\overrightarrow{\tau}(z)$
	and $u_{\gamma,\xi}$ decreases continuously, starting from $+\pi/2$, so $z$ has to turn by more 
	than $\pi$ till
	$\overrightarrow{\tau}(z)$ forms an obtuse angle with $u_{\gamma,\xi}$.
	But $\gamma$ turns by at most $\pi$, as it is contained in one of the canonical semi-circles of $c_\gamma$.  
	Hence, $p_\gamma$ points toward $h_\xi$. A similar argument shows that $q_\gamma$ also points toward $h_\xi$, 
	thereby establishing (a).

        Conversely, assume that (a) and (b) hold but $\gamma$ does not cross $h_\xi$, 
	i.e., $\gamma \subset h_\xi^+$. 
	Then $\hat\gamma \setminus \gamma$ intersects $h_\xi$ in one or two points.
	Assume that $\hat\gamma\setminus \gamma$ consists of two arcs $\gamma^-$ and $\gamma^+$,
        which are the respective portions of $\hat\gamma$ between $p_{\hat\gamma}$ and $p_\gamma$ 
        and between $q_\gamma$ and $q_{\hat\gamma}$. 
	(If one of the arcs is empty, then the argument is simpler.) 
	Assume first that one of these two arcs, say, $\gamma^-$, intersects $h_\xi$ at two points; refer to Figure~\ref{fig:cross}(ii).
	Then the above argument implies that 
	$\tau_{\gamma^-}(p_\gamma)$ points toward $h_\xi$, so %
	$\tau_\gamma(p_\gamma) = -\tau_{\gamma^-}(p_\gamma)$ points away from $h_\xi$, thereby
	contradicting (a). A similar contradiction occurs if $\gamma^+$ intersects $h_\xi$ twice.

        Assume then, as above, that each of $\gamma^-, \gamma^+$ intersects $h_\xi$ at most once, and one of them
	intersects $h_\xi$ exactly once; refer to Figure~\ref{fig:cross}(iii).
	Let $z_\gamma$ be the point on $\hat\gamma$ at which the tangent to $c_\gamma$ is parallel to $\ell$. 
	Since $c_\gamma$ intersects $h_\xi$, the outer normal of $c_\gamma$ at $z_\gamma$ is $u_{\gamma,\xi}$ 
	if $z_\gamma\in h_\xi^-$ and $-u_{\gamma,\xi}$ otherwise. 
	Assume first that $z_\gamma \in h_\xi^-$, in which case the outer normal at $z_\gamma$ is $u_{\gamma,\xi}$ 
	and $z_\gamma \in \hat\gamma \setminus \gamma$. 
	Suppose $z_\gamma \in \gamma^-$. 
	As above, for a point $z\in\gamma^-$, let $\overrightarrow{\tau}(z)$ be the 
	tangent of $c_\gamma$ at $z$ oriented in the same direction as $\hat\gamma$ (toward $p_\gamma$). 
	Note that 
	$\overrightarrow{\tau}(p_\gamma)=\tau_\gamma(p_\gamma)$ and the angle between $u_{\gamma,\xi}$ and 
	$\overrightarrow{\tau}(z_\gamma)$ is $\pi/2$.
	As we trace $z$ forward from $z_\gamma$ toward $p_\gamma$, the angle between $u_{\gamma,\xi}$ and 
	$\overrightarrow{\tau}(z)$ increases. Since $\gamma^-$ turns by at most $\pi$, $\overrightarrow{\tau}(p_\gamma)$ 
	makes an obtuse angle with $u_{\gamma,\xi}$, i.e., $\tau_\gamma(p_\gamma)$ points away from $h_\xi$, 
	contradicting condition (a). A similar contradiction can be attained if $z_\gamma \in \gamma^+$.

	Next, assume that $z_\gamma \in h_\xi^+$, in which case the outer normal of $c_\gamma$ at $z_\gamma$ is
	$-u_{\gamma,\xi}$; refer to Figure~\ref{fig:cross}(iv).
	We note that at least one of the following two properties holds: 
	$p_\gamma$ lies between $p_{\hat\gamma}$ and $z_\gamma$ 
	(if $z_\gamma \in \gamma\cup\gamma^+$), or $q_\gamma$ lies between 
	$z_\gamma$ and $q_{\hat\gamma}$ (if $z_\gamma\in \gamma^-\cup\gamma$); both conditions hold if 
	$z_\gamma \in \gamma$.  Without loss of generality, assume 
	that $p_\gamma$ appears between $p_{\hat\gamma}$ and $z_\gamma$. 
	Let $\overrightarrow{\tau}(z)$ be the same as above; again, 
	$\overrightarrow{\tau}(p_\gamma)=\tau_\gamma(p_\gamma)$.
	As we trace $z$ backward along $\hat\gamma$
	from $z_\gamma$ toward $p_\gamma$, the angle between $\overrightarrow{\tau}(z)$ and $u_{\gamma,\xi}$, 
	which at $z_\gamma$ is $\pi/2$, increases. 
	Since $\gamma$ turns by at most $\pi$, we can conclude that
	the angle between $u_{\gamma,\xi}$ and $\overrightarrow{\tau}(p_\gamma)$ is obtuse, i.e.,
	$\tau_\gamma(p_\gamma)$ points away from $h_\xi$. 
	A similar argument shows that if $z_\gamma$ lies before $p_\gamma$ (in which case $q_\gamma$ lies 
	between $z_\gamma$ and $q_{\hat\gamma}$), 
	then $\tau_\gamma(q_\gamma)$ points away from $h_\xi$. 

	Since condition (a) is violated in all cases,  we conclude that $\gamma$ intersects $h_\xi$.
\end{proof}

In view of Lemma~\ref{lem:circ-cross} and the discussion in Section~\ref{subsec:rs}, the condition that 
$\gamma$ intersects the pseudo-trapezoid $\frag\in\fragset_C$, under the setup in Case I, can be described as: 
\begin{itemize}
	\item[(i)] The endpoints $p_\gamma$ and $q_\gamma$ lie on the same side of $h_\xi$;
	\item[(ii)] both $\tau_\gamma(p_\gamma)$ and $\tau_\gamma(q_\gamma)$ point toward $h_\xi$; and
	\item[(iii)] $\Pi_C(\hat\gamma,\xi)=1$, where $\Pi_C$ is the predicate defined in~\eqref{eq:intersect-pred}.
\end{itemize}
Each of these conditions can be expressed as a constant-size semi-algebraic predicate in which 
each polynomial inequality uses at most six parameters of $\gamma$: Condition (i) needs three parameters, 
condition (ii) needs four parameters---two to describe the tangent vector 
$\tau_\gamma(p_\gamma)$ (or $\tau_\gamma(q_\gamma)$) and two to denote the normal of $\pi_\gamma$ 
(which is needed to specify $u_{\gamma,\xi}$), and condition (iii) needs six parameters to describe $\hat\gamma$.  

\smallskip

\noindent
\textbf{\textit{Case~II: $p_\gamma, q_\gamma$ lie on opposite sides of $h_\xi$.}}
In this case $\gamma$ intersects $h_\xi$ at exactly one point, say, $w_\gamma$. But the semi-circle~$\hat\gamma$ may 
intersect $h_\xi$ also at another point. There are three subcases depending on the halfplanes of $h_\xi$ that contain
the endpoints $p_{\hat\gamma}$ and $q_{\hat\gamma}$ of $\hat\gamma$ (which also determines how many times 
$\hat\gamma$ intersects $h_\xi$).
\smallskip

\textbf{\textit{Case~II.a: $p_{\hat\gamma}$ and $q_{\hat\gamma}$ lie on opposite sides of $h_\xi$.}}
Since $\gamma\subseteq\hat\gamma$, $p_\gamma$ and $p_{\hat\gamma}$ lie on one side of $h_\xi$, and 
$q_\gamma$ and $q_{\hat\gamma}$ lie on the other side. Furthermore $w_\gamma$ is the only intersection point of $\hat\gamma$ and $h_\xi$. Therefore $\gamma$ intersects the pseudo-trapezoid~$\frag$ if and only if $\hat\gamma$ intersects $\frag$, and the intersection condition for this case can be specified as:
\begin{enumerate}
	\item[(i)] $p_\gamma, q_\gamma$ lie on opposite sides of $h_\xi$; 
	\item[(ii)] $p_\gamma$ and $p_{\hat\gamma}$ lie in the same halfplane of $h_\xi$, and the same holds 
		for $q_\gamma$ and $q_{\hat\gamma}$; and 
	\item[(iii)] $\Pi_C(\hat\gamma,\xi)=1$.
\end{enumerate}
\smallskip

\textbf{\textit{Case~II.b: $p_{\hat\gamma}$ and $q_{\hat\gamma}$ lie on same side of $h_\xi$ as $p_\gamma$ does.}} 
In this case $\hat\gamma$ intersects $h_\xi$ at two points, one of which is~$w_\gamma$. 
Let $\bar{w}_\gamma$ be the other intersection point. Since $p_\gamma$ and $p_{\hat\gamma}$ lie 
in the same halfplane of $h_\xi$, $\bar{w}_\gamma$ lies on $\hat\gamma$ between $q_{\gamma}$ and~$q_{\hat\gamma}$.
That is, assuming that $\hat\gamma$ is oriented from $p_{\hat\gamma}$ to $q_{\hat\gamma}$, 
$w_\gamma$ is the first intersection point of $\hat\gamma$ with $h_\xi$ and $\bar{w}_\gamma$ is the second one.
In order to detect whether $\gamma$ intersects the pseudo-trapezoid $\frag$, unlike the previous case, 
we cannot simply use the predicate $\Pi_C(\hat\gamma,\xi)$, because $\bar{w}_\gamma$ might lie in 
$\frag$ while $w_\gamma$ does not (see Figure~\ref{fig:second}). 
Hence, we use the extension $\Pi_C^{(1)}(\hat\gamma,\xi)$ of $\Pi_C(\hat\gamma,\xi)$, defined in 
Remark~\ref{rem:enhanced-pred} (see Section~\ref{subsec:rs}), which asserts that the first intersection 
point of $\hat\gamma$, which is $w_\gamma$, lies in $\frag$. The intersection condition can now be expressed as:
\begin{enumerate}
	\item[(i)] $p_\gamma, q_\gamma$ lie on opposite sides of $h_\xi$; 
	\item[(ii)] both $p_{\hat\gamma}$ and $q_{\hat\gamma}$ lie in the same halfplane of $h_\xi$ as $p_\gamma$; and
	\item[(iii)] $\Pi_C^{(1)} (\hat\gamma,\xi)=1$.
\end{enumerate}

\begin{figure}[htb]
  \centering
  \scalebox{0.8}{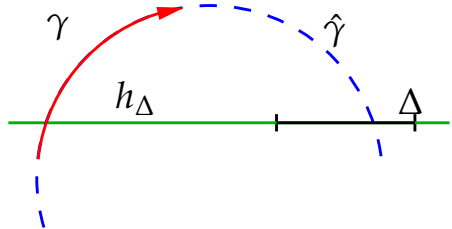}
  \caption{An instance where only the second crossing point of $\hat{\gamma}$ with 
    $h_\Delta$ lies in $C$ but this point does not belong to $\gamma$.}
  \label{fig:second}
\end{figure}

\smallskip

\textbf{\textit{Case~II.c: $p_{\hat\gamma}$ and $q_{\hat\gamma}$ lie on same side of $h_\xi$ as $q_\gamma$.}} 
This case is symmetric to the previous case, except that $\bar{w}_\gamma$ now lies between $p_{\hat\gamma}$ 
and $p_\gamma$ on $\hat\gamma$. To handle this case we simply reverse the direction of $\hat{\gamma}$,
and use the same analysis as in the previous case.
\smallskip

In summary, 
the intersection condition between $\gamma$ and $\frag$ in each of the three subcases of Case~II 
is specified as a constant-complexity semi-algebraic predicate in which each polynomial inequality uses 
at most six parameters of $\gamma$. Combining with Case~I, we conclude that the intersection 
condition between a circular arc and a pseudo-trapezoid of $\Phi_C$ can be written as a Boolean predicate with $O(1)$ polynomial inequalities, 
so that each polynomial inequality 
uses at most six of the eight (real) parameters specifying $\gamma$.
Hence, the reduced parametric dimension of circular arcs with respect to the $C$-intersection predicate is $6$.
By invoking Theorem~\ref{thm:Pi-tradeoff} with $\tO=3$ and $\tQ=6$, 
for a given storage parameter $s\in [n,n^6]$, 
$\fragset_C$ can be preprocessed, in $O^*(s)$ expected time, into a data structure of size $O^*(s)$,
so that an intersection query for $\fragset_C$ with a circular arc can be answered in 
$O\left((n/s^{1/6})^{4/5}\right) = O^*\left(n^{4/5}/s^{2/15}\right)$ time.
Putting everything together, we obtain the following result.
\begin{lemma}
  \label{prop:circular-arc-trade-off}
	Let $\W$ be a set of $n$ wide plates at some cell $\tau$ of $\reals^3\setminus Z(F)$, 
	and let $s\in [n, n^6]$ be a prespecified storage parameter.
	$\W$ can be preprocessed, in $O^*(s)$ expected time,
	into a data structure of size $O^*(s)$, so that an intersection query on $\W$ within $\tau$ with a circular 
	arc in $\reals^3$ can be answered in $O^*(n^{4/5}/s^{2/15})$ time.
\end{lemma}
\subsection{The case of planar query arcs}
\label{subsec:planar_imp}
We now let $\Gamma$ be a family of arbitrary constant-degree planar algebraic arcs.
Let $\gamma\in\Gamma$ be a query arc, and let $\sigma_\gamma$ be the curve 
containing the query arc $\gamma$.
We  break the curve $\sigma_\gamma$, as well as the arc
$\gamma$, at its singular points (see~\cite{Hil20} for the definition) and lexicographically extremal points, 
so that each resulting subarc is convex and its turning angle is at most $\pi$.
The overall number of these breakpoints is a constant that depends on %
the degree of $\gamma$.
We then apply the intersection-searching algorithm to each portion separately. Since each portion
is convex, it can intersect a plane in at most two points. A straightforward modification of the algorithm presented
above for circular arcs, and its analysis, extends to such 
general arcs; the routine details are omitted.

Let $t$ be the reduced parametric dimension of the curves supporting the query arcs. Following the same analysis as 
above and using Theorem~\ref{thm:Pi-tradeoff}, for a storage parameter $s \in [n,n^t]$, a 
$C$-intersection query on $\W$ (within $\tau$) can be answered in 
$O^*\biggl( (n/s^{1/t})^{\tfrac{2/3}{1-1/t}} \biggr)$ 
time, using an $O^*(s)$-size data structure. That is, we have

\begin{lemma}
  \label{prop:wide-trade-off}
	Let $\W$ be a set of $n$ wide plates at some cell $\tau$ of $\reals^3\setminus Z(F)$, let $\Gamma$ be a family 
	of constant-degree planar algebraic arcs such that the reduced parametric dimension
	of their supporting curves with respect to the $C$-intersection predicate is
	$t\ge 3$, and  let $s\in [n, n^t]$ be a prespecified storage 
	parameter. $\W$ can be preprocessed, in $O^*(s)$ 
	expected time, into a data structure of size $O^*(s)$, so that an intersection query on $\W$ within $\tau$ 
	with an arc in $\Gamma$ can be answered in
        $O^*\biggl((n/s^{1/t})^{\tfrac{2/3}{1-1/t}}\biggr)$ 
        time.
\end{lemma}
\section{Space/Query-Time Trade-Offs for the Overall Data Structure}
\label{sec:trade-off}

We are now ready to describe how we combine the data structure described in Section~\ref{sec:main} with the general technique for semi-algebraic predicate searching, outlined in Appendix~\ref{app:multi-level},
to obtain space/query-time trade-offs for arc-intersection queries. 
As above, let $\T$ be the set of input plates, and let $\Gamma$ 
be the family of query arcs. Let $\tO,\tQ\ge 3$ be the reduced parametric dimensions of $\T$ and $\Gamma$, respectively, with respect to their intersection predicate.
Given a storage parameter $s\in [n,n^{\tQ }]$, we construct, in $O^*(s)$ expected time, 
a data structure of size $O^*(s)$, that can answer an arc-intersection query in time 
\[
O^* \left ( \frac{n^{2- 3/\tO}}{s^{1-2/\tO}} + 
	\left ( \frac{n}{s^{1/\tQ}}\right )^{\tfrac{2/3}{1-1/\tQ}} \right ) .
\]
We note that the second term dominates for $s\ge n^{3/2}$ as $\tQ\ge 3$. We combine the data structure described in Section~\ref{sec:main} with a linear-size data structure constructed in the object space, described in Appendix~\ref{subsec:Pi-lin}, for $s\le n^{3/2}$, and with a large-size data structure constructed in the query space, described in Appendix~\ref{subsec:Pi-tradeoff}, for $s>n^{3/2}$, as follows.

We build the same partition tree $\Psi$ as described in Section~\ref{subsec:data-structure} 
but with a few twists. Let $D>0$ be a sufficiently large constant parameter as before. First, the
recursive subproblem at a node $v$ of $\Psi$ involves two parameters: a subset $\T_v \subseteq \T$ of plates of 
size $n_v$ (as before), and a storage parameter $s_v \ge n_v$ that specifies, in the asymptotic
$O^*(\cdot)$ sense, the size of the subtree $\Psi_v$ constructed at $v$. 
Initially, at the root, $\T_v=\T$ and $s_v=s$. Second, we set the threshold value $n_0 \coloneqq n_0 (s)$ for
termination of the recursion that depends on the storage parameter $s$. In particular,
\begin{equation}
  \label{eq:threshold}
  n_0(s) = \left \{ \begin{array}{cl}
    \frac{n^3}{s^2} & s \le n^{3/2}, \\[1mm]
    \left ( \frac{s}{n^{3/2}} \right )^{\frac{1}{\tQ - 3/2}} & s > n^{3/2},
  \end{array} \right .
\end{equation}

Suppose we are at a node $v$. If $n_v > n_0$, we construct a 
partitioning polynomial $F_v$ of degree at most $c_1D$, for a suitable constant $c_1>0$,
and the secondary data structures 
$\ZDS_v$ and $\WDS_v$ to handle the query arcs that are contained in $Z(F_v)$ and to answer 
intersection queries for wide plates, respectively, as before. However the size of each of 
$\ZDS_v, \WDS_v$  is now $O^*(s_v)$, which allows a query to be answered faster. 
We already have described in Section~\ref{sec:wide-tradeoff} how to 
construct $\WDS_v$ for a given storage parameter $s_v$, and we describe a construction
of $\ZDS_v$, with the same asymptotic bounds, in Section~\ref{sec:zero-set} (see Lemma~\ref{lem:onzf-tradeoff}).
Finally, for each cell $\tau$ of $\reals^3\setminus Z(F_v)$, we recursively construct the 
data structure on the subset $\T_\tau \subset \T_v$ of the narrow plates at $\tau$ with storage parameter 
set to $s_v/D^3$. This storage parameter at the children of $v$ is chosen so that the total space used at each level of $\Psi$ is $O^*(s)$.

If $n_v \le n_0$, then we regard $v$ as a leaf of $\Psi$. For $s \le n^{3/2}$, we construct a data structure of size 
$O^*(n_v)$, using Theorem~\ref{lem:Pi-lin}, so that an intersection query on $\T_v$ can be answered in $O(n_v^{1-1/\tO})$ time. For $s > n^{3/2}$, we construct a data structure of size $O^*(n_v^\tQ)$, using Theorem~\ref{thm:Pi-tradeoff},
so that an intersection query on $\T_v$ can be answered in $O^*(1)$ time.

An intersection query with an arc $\gamma\in\Gamma$ is answered as described in Section~\ref{subsec:plate-query},
except that we use the procedures described in Sections~\ref{sec:wide-tradeoff} and~\ref{sec:zero-set} to query 
the respective secondary data structure $\WDS_v$ and $\ZDS_v$ at each node $v$. When we reach a leaf of $\Psi$, we use the query procedure described in Section~\ref{subsec:Pi-lin} or Section~\ref{subsec:Pi-tradeoff} depending on whether $s \le n^{3/2}$ or $s>n^{3/2}$, respectively.

Let $S(n_v,s_v)$ denote the maximum size of the subtree $\Psi_v$, and let $Q(n_v,s_v)$ 
denote the maximum query time on $\Psi_v$. We obtain the following recurrence for $S(n_v,s_v)$:
\begin{equation}
  \label{eq:moderate-storage}
  S(n_v,s_v) \le
  \begin{cases*}
    \displaystyle c_2D^3 \cdot S \left ( \frac{n_v}{D^2},\frac{s_v}{D^3} \right) + c_3 s_vn_v^{\delta} & for $n_v >  n_0$, \\[1mm]
    c_4 n_v^{1+\delta} & for $n_v\le n_0$ and $s\le n^{3/2}$,\\[1mm]
    c_4 n_v^{\tQ+\delta} & for $n_v\le n_0$ and $s> n^{3/2}$,\\[1mm]
  \end{cases*}
\end{equation}
where $c_2, c_3, c_4$, and $\delta$ are constants as defined in Section~\ref{subsec:data-structure}.
Using induction on $n_v$, as in Section~\ref{subsec:data-structure}, we can show that the solution to the recurrence is
\begin{equation}
	\label{eq:moderate-storage-sol}
	S(n_v,s_v) \le A \left ( s_v + s \left ( \frac{n_v}{n}\right )^{3/2} \right ) n_v^\eps\,,
\end{equation}
for any $\eps>\delta$, provided the constants $A$ and $D$ are chosen sufficiently large. 
Here is an intuitive explanation of why (\ref{eq:moderate-storage-sol}) is the solution to the above recurrence; we omit a formal proof by induction from here.
By the choice of the storage parameters for each subproblem, the secondary 
data structures use $O^*(s_v)$ storage at each level of $\Psi$. Furthermore, there are 
$O^*((n_v/n_0)^{3/2})$ leaves of $\Psi$. For $s\le n^{3/2}$, each leaf uses $O^*(n_0)$ space and
$n_0=n^3/s^2$, so the total space used at the leaves of $\Psi$ is $O^*((n_v/n)^{3/2}s)$.  For $s>n^{3/2}$, each leaf uses $O^*(n_0^\tQ)$ space and $n_0=(s/n^{3/2})^{\tfrac{1}{\tQ - 3/2}}$, so the total space used at the leaves of $\Psi$ is again $O^*((n_v/n)^{3/2}s)$.  Hence, the total size of the subtree $\Psi_v$, including the space used by 
the secondary structures, is
$O^*(s_v + s (n_v/n)^{3/2})$. Since $n_v=n$ and $s_v=s$ at the root of $\Psi$, 
the overall size of $\Psi$ is $S(n,s) = O^*(s)$.

Concerning the query cost, by adapting~\eqref{eq:query} and plugging the query-time bounds 
for the secondary data structures from Lemmas~\ref{prop:wide-trade-off} and~\ref{lem:onzf-tradeoff}, 
we obtain the following recurrence:
\begin{equation}
  \label{eq:moderate-query}
  Q(n_v,s_v) \le
  \begin{cases*}
	  \displaystyle c_5 D \cdot Q\left (\frac{n_v}{D^2},\frac{s_v}{D^3}\right ) + 
		c_6 \left (\frac{n_v}{s_v^{1/\tQ}}\right)^{\frac{2/3}{1-1/\tQ}}n_v^\delta & for $n_v> n_0$, \\[1mm]
	  c_7 n_v^{1-\tfrac{1}{\tO }+\delta} & for $n_v\le n_0$ and $s\le n^{3/2}$,\\[1mm]
	  c_7 n_v^{\delta} & for $n_v\le n_0$ and $s> n^{3/2}$,
  \end{cases*}
\end{equation}
where $c_5, c_6$ and are constants as defined in Section~\ref{subsec:plate-query}, $c_7>0$ is a constant, and $\delta$ is an arbitrarily small constant. 
For our choice of storage parameter, \eqref{eq:moderate-query} implies that the time spent in querying the secondary 
structures at the children of a node $v$, namely the sum of the nonrecursive overhead terms at the children of $v$
recursively visited by the query, is 
\[ c_5D \cdot c_6 \left ( \frac{(n_v/D^2)}{(s_v/D^3)^{1/\tQ}} \right )^{\frac{2/3}{1-1/\tQ}} = 
c_5 c_6 \left (\frac{n_v}{s_v^{1/\tQ}} \right )^{\frac{2/3}{1-1/\tQ}} \cdot D^\alpha ,
\]
where 
\[ \alpha = 1 + \left (\frac3\tQ-2\right) \frac{2/3}{1-1/\tQ}  = \frac{1/\tQ-1/3}{1-1/\tQ} \le 0
\] for $\tQ\ge 3$.
Therefore, for $\tQ \ge 3$, the total time spent in querying the secondary data structures at all descendents of 
a node $v$ is $O^*\biggl ( (n_v/s_v^{1/\tQ})^{\tfrac{2/3}{1-1/\tQ}}\biggr)$. 
For $s\le n^{3/2}$, the procedure visits $O^*(s/n)$ leaves and
spends $O^*(n_0^{1-1/\tO})$ time at each such leaf, thus spending a total of $O^*(n^{2-3/\tO}/s^{1-2/\tO})$ time at the leaves.
For $s\ge n^{3/2}$,  the query time is dominated by the time spent on the secondary structures.
Summing these bounds and setting $n_v=n$ and $s_v=s$, 
the overall query time is 
\begin{equation}
	\label{eq:moderate-qtime-bound}
	Q(n,s) = O^* \biggl ( \frac{n^{2- 3/\tO }}{s^{1-2/\tO }} + 
	\biggl ( \frac{n}{s^{1/\tQ}}\biggr )^{\tfrac{2/3}{1-1/\tQ}} \biggr ).
\end{equation}
It can be checked that for $s\ge n^{3/2}$ (assuming $\tQ\ge 3$), the second term dominates.
In particular, for $s=n^{3/2}$, the query time is $O^*\biggl(n^{\tfrac{2\tQ -3}{3(\tQ -1)}}\biggr)$.
If $\Gamma$ is a family of planar arcs, then $\tQ$ is the reduced parametric dimension of the curves supporting the arcs in $\Gamma$.
Putting everything together, we obtain the following result:
\begin{theorem}
  \label{thm:space-query-tradeoff}
  Let $\Gamma$ be a family of parametrized algebraic arcs of constant complexity, let $\T$ be a set 
	of $n$ constant-complexity plates in $\reals^3$, let $\tO$ be the reduced parametric dimension of $\T$, let $\tQ$ be the reduced 
	parametric dimension of $\Gamma$ or of the curves supporting the arcs in $\Gamma$ if they are planar,  and let 
	$s\in [n,n^{\tQ }]$ be a storage parameter. $\T$ can be preprocessed, in expected time $O^*(s)$, into a 
	data structure of size $O^*(s)$, so that an arc-intersection query amid the plates of $\T$ 
	with an arc in $\Gamma$ can be answered in time
        \[
          O^* \left ( \frac{n^{2- 3/\tO }}{s^{1-2/\tO }} + 
            \left ( \frac{n}{s^{1/\tQ}}\right )^{\tfrac{2/3}{1-1/\tQ}} \right ).
        \]
\end{theorem}

If $\T$ is a set of $n$ triangles in $\reals^3$, then $\tO =5$ for general algebraic arcs (namely, this is
the best bound we have so far) and $\tO =4$ for lines; $\tQ=4$ in the latter case. We thus obtain the following 
corollary.

\begin{corollary}
  \label{cor:space-query-triangle}
		Let $\T$ be a set of $n$ triangles in $\reals^3$, let $\Gamma$ be a family of parametrized algebraic arcs.
	Let $\tQ$ be the reduced parametric dimension of $\Gamma$ or of the curves supporting the arcs of $\Gamma$ 
	if they are planar, and let $s\in[n,n^{\tQ}]$ be a storage parameter. $\T$ can be preprocessed, 
	in $O^*(s)$ expected time, into a data structure of size $O^*(s)$, so that an intersection query with an arc 
	in $\Gamma$ can be answered in 
	$O^* \biggl ( n^{7/5}/s^{3/5} + (n/s^{1/\tQ})^{\tfrac{2/3}{1-1/\tQ}} \biggr )$ time.
	If $\Gamma$ is a set of line segments, then an intersection 
	query can be answered in $O^* (n^{5/4}/s^{1/2} + n^{8/9}/s^{2/9})$ time,
	for $s\in [n,n^4]$.
\end{corollary}
\section{Handling the Zero Set}
\label{sec:zero-set}

Let $\Gamma$ and $\T$ be as above, and let $F$ be a partitioning polynomial of degree at most $D$, 
for some constant $D>0$, 
as described in Section~\ref{sec:main}. We assume that no plate in $\T$ lies in $Z(F)$, 
i.e., the intersection of a plate with $Z(F)$ is a collection of algebraic arcs, all contained
in a single algebraic curve (of constant degree).
This assumption is justified because, by assumption,
$Z(F)$ (or rather its planar components) contains only $O(1)$ input plates, which are
handled separately; see Section~\ref{subsec:data-structure}. 
In this section, we present a data structure for answering 
intersection queries amid $\T$ with an arc $\gamma\in\Gamma$ that is contained in $Z(F)$. 

Our data structure is based on a hierarchical polynomial-partitioning technique proposed 
in~\cite{AAEZ} (see Lemma~\ref{lem:aaez} in the appendix).
For a subset $\X\subseteq \T$ 
and for a parameter $E\gg D$, this technique constructs a partitioning polynomial 
$G \in \reals[x,y,z]$ of degree at most $c'_1 E$, for a constant $c'_1>0$ that depends 
(polynomially) on $D$, such that any (two-dimensional) cell of $Z(F)\setminus Z(G)$ is crossed by 
at most $m/E$ plates of $\X$. Furthermore, each such cell contains at most $m/E^2$ endpoints of the intersection 
arcs of $\X$ and $Z(F)$.
Let $\A(G;F)$ denote the decomposition of $Z(F)$ induced by $Z(G)$.
The number of cells in $\A(G;F)$ is at most $c'_2E^2$, for a constant $c'_2>0$ that depends on $D$. 
With this polynomial partitioning at our disposal, we build a data structure similar to the one
presented in Section~\ref{subsec:data-structure}. Namely, we construct a partition tree $\ZDS$ 
on $\T$. Each node of $\ZDS$ stores a secondary data structure, analogous to the one in 
Section~\ref{sec:wide}, to handle wide plates at the corresponding cell. We describe the 
overall data structure briefly, highlighting the differences from the earlier structure.

\subsection{Overall data structure}
\label{subsec:z-overall}

Each node $z\in \ZDS$ is associated with a constant-complexity semi-algebraic cell $\zcell_z \subseteq Z(F)$ 
and a subset~$\T_z\subseteq\T$ of $n_z$ plates. If $z$ is the root of $\ZDS$, then $\zcell_z=Z(F)$ and 
$\T_z=\T$. Set $n_z =|\T_z|$. We fix two sufficiently large constants $n_1 \ge 0$ (a threshold parameter) 
and $E = D^{O(1)}$.

If $n_z\le n_1$ then $z$ is a leaf and we simply store $\T_z$ at $z$. Otherwise, we construct a 
partitioning polynomial $G_z$ for $\T_z$ (relative to $F$) of degree at most $c'_1E$, as described 
above, and store~$G_z$ at $z$. By our general-position assumption on the input plates stated in
Section~\ref{sec:main}---only $O(1)$ plates lie in any plane and any line is contained in supporting planes of 
$O(1)$ input plates---there are only $O(1)$ input plates $\Delta$ for which $\dim (\Delta\cap Z(F)\cap Z(G))=1$.
Let $\T_z^0 \subseteq \T_z$ be the subset of these plates. We store $\T_z^0$ at $z$.
Let $\sigma$ be a one- or two-dimensional cell%
\footnote{%
We also have to consider zero-dimensional cells of $\A(G;F)$, but they are trivial to handle and we ignore them  for the sake of simplicity of presentation.
}
of $\A(G;F)$. If $\sigma$ is a one-dimensional cell, which is a connected arc of the intersection curve 
$Z(G)\cap Z(F)$, we compute the set $\T_\sigma \subseteq \T_z\setminus \T_z^0$ of plates that intersect 
$\sigma$. Note that $\dim (\Delta\cap Z(F)\cap Z(G))=0$ for all $\Delta\in\T_\sigma$.
Using a one-dimensional range-searching data structure, we preprocess $\T_\sigma$ into a 
linear-size data structure $\ZZDS_\sigma$ that supports intersection queries for $\gamma\cap\sigma$, 
$\gamma\in\Gamma$, in $O(\log n)$ time; see also Appendix~\ref{subsec:Pi-lin}. 
We store $\ZZDS_\sigma$ at $z$. 
Assume next that $\sigma$ is a two-dimensional cell. We now call an input plate $\plate$, which crosses $\sigma$, 
\emph{narrow} if $\sigma$ contains an endpoint of the intersection arc $\plate\cap Z(F)$, 
and \emph{wide} otherwise. 
We create a child $u_\sigma$ of $z$ associated with $\sigma$. 
Let $\W_\sigma$ (resp., $\T_\sigma$) denote the set of the wide (resp., narrow) plates at $\sigma$. By construction,
$|\W_\sigma| \le n_z/E$ and $|\T_\sigma| \le n_z/E^2$.
We construct a secondary data structure $\ZWDS_\sigma$ on $\W_\sigma$, as described in 
Section~\ref{subsec:z-wide} below, for answering arc-intersection queries amid the plates 
of $\W_\sigma$ (for intersections that occur within $\sigma$) with the arcs of $\Gamma$ that lie in $Z(F)$.
$\ZWDS_\sigma$ is stored at the child $u_\sigma$ of $z$. Finally, we set $\T_{u_\sigma} = \T_\sigma$, 
recursively construct a partition tree for $\T_{u_\sigma}$, and attach it as the subtree rooted at $u_\sigma$. 

Let $S(n_z)$ be the maximum size of the subtree $\ZDS_z$ constructed on a set of at most $n_z$ plates.
For $n_z\le n_1$, $S(n_z)=O(n_z)$.
For $n_z > n_1$, Lemma~\ref{lem:zwide-time} below implies that the 
secondary structure for handling wide plates requires $O^*(n_z)$ space. Therefore $S(n_z)$ obeys the recurrence:
\begin{equation}
  \label{eq:zstorage}
  S(n_z) \le
  \begin{cases*}
	  \displaystyle c'_2E^2 \cdot S\left (\frac{n_z}{E^2}\right) + c'_3 n_z^{1+\delta} & for $n_z\ge n_1$, \\[1mm]
	  c'_4 n_z & for $n_z\le n_1$,
  \end{cases*}
\end{equation}
where $c'_2, c'_3, c'_4,  \delta$ are constants analogous to those in~\eqref{eq:storage}, and $c'_3$ depends on $\delta$.
The solution to the above recurrence is $S(n_z) = O(n_z^{1+\eps})$ for any constant $\eps\ge\delta$,
provided that $E$ is chosen sufficiently large.  Hence, the overall size of $\ZDS$ is $O^*(n)$.

The query procedure is similar to Section~\ref{subsec:plate-query}. 
Using Lemma~\ref{lem:zwide-time} again for the query cost amid the wide plates, which is $O^*(n^{2/3})$, we obtain 
the following recurrence for the query cost $Q(n_z)$:
\begin{equation}
  \label{eq:zquery}
  Q(n_z) \le
  \begin{cases*}
	  \displaystyle c'_5 E \cdot Q\left (\frac{n_z}{E^2} \right) + c'_6 n_z^{2/3+\delta} & for $n_z\ge n_1$, \\[1mm]
	  c'_7 n_z & for $n_z\le n_1$,
  \end{cases*}
\end{equation}
where $c'_5, c'_6, c'_7, \delta$ are constants analogous to those in~\eqref{eq:query}, and $c'_6$ depends on $\delta$.
The solution to the above recurrence is~$Q(n_z) = O(n_z^{2/3+\eps})$ for any constant $\eps\ge\delta$. 
Hence, we obtain the following result.

\begin{lemma}
  \label{lem:onzf}
  Let $\Gamma$ be a family of constant-degree parametrized algebraic arcs, let $\T$ be a set of $n$ plates 
  of constant complexity in $\reals^3$, and let $F$ be a partitioning polynomial of some constant degree $D$.
	$\T$ can be processed, in $O^*(n)$ expected time, into a data structure of $O^*(n)$ size, 
	so that an arc-intersection query amid the plates of $\T$ with an arc in $\Gamma$ that lies in~$Z(F)$ can be answered in $O^*(n^{2/3})$ time. The hidden constants and factors in the bounds depend on $D$.
\end{lemma}

To obtain a space/query-time trade-off, we proceed as in Section~\ref{sec:trade-off},  and combine this data 
structure with the general technique described in Appendix~\ref{subsec:Pi-tradeoff}.
Following the same analysis as in Section~\ref{subsec:Pi-tradeoff}, we obtain the following result.

\begin{lemma}
  \label{lem:onzf-tradeoff}
  Let $\Gamma$ be a family of constant-degree parametrized algebraic arcs, let $\T$ be a set 
	of $n$ plates in $\reals^3$, let $\tQ$ be the reduced parametric dimension of arcs in $\Gamma$ or of the curves supporting the arcs in $\Gamma$ if they are planar, let $F$ be a partitioning polynomial of constant degree, and let 
	$s\in [n,n^{\tQ}]$ be a storage parameter. $\T$ can be preprocessed, in expected 
	time $O^*(s)$, into a data structure of size $O^*(s)$, so that an arc-intersection  query amid the plates of 
	$\T$ with an arc in $\Gamma$ that lies in $Z(F)$ can be answered in 
        $O^* \biggl((n/s^{1/\tQ })^{\tfrac{2/3}{1-1/\tQ}} \biggr)$
        time.
\end{lemma}
\subsection{Handling wide plates}
\label{subsec:z-wide}

It remains to consider the case of wide plates.
Let $F$ be a partitioning polynomial as defined in Section~\ref{sec:main}, and let $G\in \reals[x,y,z]$ be a 
second partitioning polynomial 
(relative to $F$) as described in Section~\ref{subsec:z-overall}. In this subsection we present an algorithm for preprocessing the 
set $\W_\sigma$ of wide plates at a cell~$\sigma$ of $Z(F)\setminus Z(G)$, for intersection queries 
(within $\sigma$) with the arcs of $\Gamma$ that lie in $Z(F)$. The high-level approach closely follows
the approach of Section~\ref{sec:wide}. It decomposes $\sigma\cap\Delta$, for $\Delta\in \W_\sigma$, 
into subarcs and groups the resulting subarcs into $O(1)$ clusters using a CAD construction, so that each cluster has an
$O(1)$-size semi-algebraic encoding that depends only on the plane supporting $\Delta$ but not on
its boundary. 

We now describe the algorithm in more detail, highlighting the differences from Section~\ref{sec:wide}.
Let $\ES_3$ and $\ES$ be the same as in Section~\ref{sec:wide}. Define $\hat F, \hat G \in \reals[a,b,c,x,y]$ as
$\hat{F} (a,b,c,x,y) = F(x,y,ax+by+c)$ and $\hat G (a,b,c,x,y) = G(x,y,ax+by+c)$. We construct a five-dimensional 
CAD $\ZCAD$ of $\ES$ induced by $\{\hat F, \hat G\}$. Note that $\ZCAD$ is a refinement of $\CAD$, the CAD
constructed in Section~\ref{sec:wide} for $\hat{F}$. Here too, each cell of $\ZCAD$ is given by a sequence
of equalities or inequalities (one from each row) of the form in~\eqref{eq:cad}.
Let $\ZCAD_3$ denote the projection of $\ZCAD$ onto $\ES_3$, which
we again refer to as the \emph{base} of $\CAD$.
We will be interested only in those cells of $\ZCAD$ that lie in $Z(\hat F)$, and we denote by
$\ZCAD_F$ the subset of these cells.

For a point $\xi=(a_\xi, b_\xi, c_\xi)\in \ES_3$, let $\zfiber(\xi)$ 
denote the two-dimensional fiber of $\ZCAD$ over $\xi$ and $\zfiber^\uparrow (\xi)$ its lifting 
to the plane $h_\xi$.
Again, we are interested only in the one-dimensional cells of $\zfiber^\uparrow(\xi)$ that lie in $Z(F)$.
The combinatorial structure of $\zfiber(\xi)$, as well as of its lifting 
$\zfiber^\uparrow(\xi)$, remains the same for all points $\xi$ in the same base cell $\psi \in \ZCAD_3$.
As earlier, we associate each cell $\zCADcell$ of $\ZCAD_F$ with a fixed cell of $\A(G;F)$,
denoted as $\sigma_\zCADcell$, such that for all points $\xi$ in the base cell $\zCADcell^\downarrow\in \ZCAD_3$, 
$\zCADcell^\uparrow (\xi)$ is an edge of $\zfiber^\uparrow (\xi)$ that lies in $\sigma_\zCADcell$. 
For a cell $\zCADcell\in\ZCAD_F$ and a point $\xi\in\zCADcell^\downarrow$, we use $\zCADcell^\uparrow (\xi)$ 
to denote the arc of $\zfiber^\uparrow(\xi)$ corresponding to the cell $\zCADcell$.

Let $\sigma$ be a $2$-cell of $\A(G;F)$.
Let $\Delta \in \W_\sigma$ be a plate that is wide at $\sigma$, and let $\psi\in\ZCAD_3$ be the base cell containing 
$\Delta^*$. We define $\zfragset_{\Delta,\sigma}$ to be the decomposition of $\Delta\cap\sigma$ into subarcs 
induced by $\zfiber^\uparrow(\Delta^*)$, i.e.,
\[
\zfragset_{\Delta,\sigma} \coloneqq \{ \zCADcell^\uparrow (\Delta^*) \mid \zCADcell \in \ZCAD_F, \, 
	\Delta^* \in \zCADcell^\downarrow,\,
	\sigma_\zCADcell=\sigma,\,
	\zCADcell^\uparrow(\Delta^*) \subseteq \Delta\cap Z(F) 
	\}.
\]
Set $\zfragset_\sigma = \bigcup_{\Delta\in\W_\sigma} \zfragset_{\Delta,\sigma}$.
For a cell $\zCADcell\in\ZCAD_F$ with $\sigma_\zCADcell=\sigma$, we define 
$\zfragset_\zCADcell \subseteq \zfragset_\zcell$ to be the set of arcs defined by the cell 
$\zCADcell \in \ZCAD_F$, i.e., 
\[
\zfragset_\zCADcell \coloneqq \{ \zCADcell^\uparrow (\Delta^*) \mid \Delta \in \W_{\sigma}\; \wedge\; 
\zCADcell^\uparrow (\Delta^*) \in \zfragset_{\Delta,\sigma} \} .
\]
By definition, $\zfragset_\sigma = \bigcup_{\zCADcell\in\ZCAD_F:\zcell_\zCADcell=\zcell} \zfragset_\zCADcell$.
All arcs in $\zfragset_\zCADcell$ have a fixed constant-size 
semi-algebraic encoding of the form described in~\eqref{eq:encoding} 
that depends only on $F$ and $G$.
Furthermore, the polynomials defining the encoding do not even depend on the supporting planes--the  coefficients of supporting planes only appear as variables in these functions. 
We can therefore define a \emph{$\zCADcell$-intersection predicate}
$\zpred_\zCADcell\colon \Gamma \times \ES_3 \rightarrow \{0,1\}$, analogous to~\eqref{eq:intersect-pred}, such that 
$\zpred_\zCADcell(\gamma; \xi)=1$ if an intersection point of $\gamma$ and $h_\xi$ lies on the arc $\zCADcell(\xi)$. 
That is,
\begin{equation}
	\label{eq:zintersect-pred}
	\zpred_\zCADcell(\gamma; \xi) \coloneqq \left \{
		\begin{array}{ll} 
			1 & \text{if } \xi\in \zCADcell^\downarrow \wedge 
			\exists  (x_p,y_p,z_p)\in \gamma \cap h_\xi \; \text{s.t.}\; (x_p,y_p)\in \zCADcell(\xi), \\[1mm]
			0 & \text{otherwise.}
		\end{array}
		\right .
\end{equation}
For an arc $\zeta\in\zfragset_\zCADcell$  contained in a plate $\plate\in\W_\sigma$, by construction,
$\zeta$ intersects an arc $\gamma\in\Gamma$ lying in $Z(F)$  if and only if $\zpred_\zCADcell (\gamma,\plate^*)=1$. 
We preprocess $\zfragset_\zCADcell$ into a data structure of $O^*(n)$ size, as described in Appendix~\ref{subsec:Pi-lin},
so that an arc-intersection query with an 
arc $\gamma\in\Gamma$ lying in $Z(F)$ can be answered in $O^*(n^{2/3})$ time,
because the reduced parametric dimension of $\zfragset_\zCADcell$ is $3$. Building such a data structure for all cells of $\ZCAD_F$ and proceeding as in Section~\ref{subsec:rs}, we obtain the following:
\begin{lemma}
	\label{lem:zwide-time}
	Let $\W$ be a set of $n$ wide plates for some $2$-cell $\zcell$ of $Z(F)\setminus Z(G)$, 
	and let $\Gamma$ be a family of constant-degree parametrized algebraic arcs.
	$\W$ can be preprocessed, in $O^*(n)$ expected time, into a data structure of $O^*(n)$ size, 
	so that an intersection query on $\W$ within $\sigma$ with an arc in 
	$\Gamma$ that lies in $Z(F)$ can be answered in $O^*(n^{2/3})$ time. 
\end{lemma}
\section{Arc-Intersection Queries amid Triangles with Near-Linear Storage}
\label{sec:linst}

In this section we describe a near-linear-size data structure for answering arc-intersection queries amid a set~$\T$ of $n$~triangles in~$\reals^3$. 
Our main contribution is to show that, despite the default parametric 
dimension of a triangle in~$\reals^3$ being~$9$ (which results, e.g., by specifying the coordinates of each of 
its three vertices, but see below for a different parametrization), the reduced parametric dimension 
with respect to constant-degree parametrized 
algebraic arcs is only $5$, which immediately leads to an intersection-searching data structure
with $O^*(n^{4/5})$ query time, using $O^*(n)$ space.

We represent a non-vertical triangle\footnote{%
  Vertical triangles have only $8$ degrees of freedom, and can be handled by a similar, and somewhat simpler,
  technique; we omit here the easy details.}
$\Delta$ in $\reals^3$, with edges $e_1, e_2, e_3$, by the $9$-tuple 
$\Delta^* = (\xi, \mu, \nu) \in \reals^3\times\reals^3\times\reals^3$, 
where $\xi$ is the point dual to the plane $h_\Delta$ supporting $\Delta$, and
$\mu=(\mu_1, \mu_2, \mu_3)\in\reals^3$, $\nu = (\nu_1, \nu_2, \nu_3)\in\reals^3$ are defined so that,
for $i=1,2,3$, $(\mu_i, \nu_i)$ specifies the line supporting the edge $e_i$ (within the plane $h_\Delta$). 
As in Section~\ref{sec:wide}, we use $\ES_3$ to denote the set of all non-vertical planes, each
represented by its dual point. Recall that the family $\Gamma$ of query arcs is defined in 
Section~\ref{sec:main} (see eq.~(\ref{eq:arc-family})) as follows:
\[
\Gamma \coloneqq \{ \gamma (\delta, \alpha^-, \alpha^+) \mid 
\delta \in \ES_t\; \mbox{and}\; \alpha^-, \alpha^+ \in \reals \} ,
\]
where $t$ is the parametric dimension of curves supporting the arcs in $\Gamma$, and
$\ES_t$ is the $t$-dimensional parametric space of these curves.

Let $\Delta \in \T$ be a triangle, let $\gamma \coloneqq \gamma (\delta,\alpha^-,\alpha^+) \in \Gamma$ be an arc, and let $\sigma_\delta$ be the curve supporting the arc $\gamma$, where 
$\sigma_\delta(\alpha)=(x_\delta(\alpha), y_\delta(\alpha), z_\delta(\alpha))$ for $\alpha\in\reals$. 
Then $\gamma$ intersects $\Delta$ if and only if (i)
$\gamma$ intersects $h_\Delta$, and (ii)
one of the intersection points of $\gamma \cap h_\Delta$ lies, within $h_\Delta$,
on the positive side of each of the lines $\ell_1$, $\ell_2$, $\ell_3$
supporting the respective edges $e_1$, $e_2$, $e_3$ of~$\Delta$, where
the positive side of a line in~$h_\Delta$ is the halfplane of~$h_\Delta$ bounded by
the line and containing $\Delta$. For technical reasons, we rewrite these conditions as follows:
\smallskip

There exists $\alpha\in\reals$ such that
\begin{itemize}
\item[($I_1$)] 
	$\curve_\delta(\alpha) \in h_\Delta$.
\item[($I_2$)]
	$\alpha^- \le \alpha\le \alpha^+$ (this sub-condition is vacuous when $\gamma$ is a full algebraic curve).
\item[($I_3$)] 
	$\curve_\delta(\alpha)$ lies on the positive side of $\ell_1$.
\item[($I_4$)] 
	$\curve_\delta(\alpha)$ lies on the positive side of $\ell_2$.
\item[($I_5$)] 
	$\curve_\delta(\alpha)$ lies on the positive side of $\ell_3$.
\end{itemize}

A major technical issue that arises in expressing these conditions as semi-algebraic predicates is that 
$\curve_\delta$ may intersect $h_\Delta$ in 
several points, and we need to test ($I_2$)--($I_5$) together for each intersection point separately.
That is, for each of these tests, we need to specify which point (i.e., which 
value of $\alpha$) is to be used, and all four subpredicates ($I_2$)--($I_5$) must use the same value.
More specifically, ($I_1$) gives $\alpha$ as a root of the algebraic equation
$z_{\delta}(\alpha) = a x_{\delta}(\alpha) + b y_{\delta}(\alpha) +c$, where $(a,b,c)$ 
are the coefficients of $h_\Delta$, and we need to specify which root to use in 
testing for conditions ($I_2$)--($I_5$).

To address this problem, we partition the product space $\ES_3 \times \ES_t$ (in which a point $(\xi,\delta)$ 
corresponds to a pair of a plane $h_\xi$ and a curve $\sigma_\delta$)  into $O(1)$ semi-algebraic 
regions such that for all pairs $(\xi, \delta) \in \ES_3\times \ES_t$ lying in the same region, 
the intersection points of the plane $h_\xi$ and the curve $\sigma_\delta$ have the same topological structure,
and each of the intersection points has a fixed-size semi-algebraic encoding that is independent 
of $\xi$ and $\delta$. This encoding enables us to 
express each of the above conditions as a semi-algebraic predicate. We now describe the details.

Let $G \colon \ES_3\times\ES_t\times \reals \rightarrow \reals$ be the algebraic function 
\[
G(a,b,c; \delta; \alpha) \coloneqq z_\delta(\alpha) - a x_\delta(\alpha) - b y_\delta(\alpha) - c,
\]
where $x_\delta, y_\delta, z_\delta$ are the functions that define the curve $\sigma_\delta$
represented by the point $\delta\in\reals^t$. 

We construct a CAD $\CAD_{t+4}$ of $\ES_3 \times \ES_t \times \reals$ induced by $G$, 
where the order of coordinate elimination is first $\alpha$, then $\delta$, and then $c,b,a$.\footnote{%
  Technically, the functions $x_\delta$, $y_\delta$, $z_\delta$ are algebraic functions, not 
  necessarily polynomials, and the CAD construction is defined over polynomials. To address 
  this issue formally, we need to construct a CAD over 
  $\tilde{G}(a,b,c;x,y,z) \coloneqq z - ax - by - c$ and the three polynomials $P_x(x,\delta,\alpha)$, 
  $P_y(x,\delta,\alpha)$, $P_z(z,\delta,\alpha)$ that define the algebraic functions $x_\delta, y_\delta, z_\delta$, 
  respectively.
  This calls for a CAD in $t+7$ dimensions, but this does not affect the algorithm in any significant 
  way. For the sake of simplicity, we describe the construction under the assumption
  of $\delta$ having a polynomial parametrization.} 
Each cell of $\CAD_{t+4}$ has a constant-size semi-algebraic representation analogous to~\eqref{eq:cad}.

Many of the technical details of the structure of $\CAD_{t+4}$ are similar to those 
used for the CAD introduced in Section~\ref{sec:wide}, 
but we briefly sketch them because the setup is different here.
Let $\CAD_3$ denote the projection of $\CAD_{t+4}$ onto the object space $\ES_3$, and let $\CAD_{t+3}$
denote its projection onto the $(t+3)$-dimensional space $\ES_3 \times \ES_t$. 
The latter gives the desired decomposition of the product space $\ES_3 \times \ES_t$. 
For a cell $C\in\CAD_{t+4}$, let $C^\Downarrow$ (resp., $C^\downarrow$) denote 
its projection onto $\ES_3$ (resp., onto $\ES_3\times \ES_t$).
For each point $\xi\in\ES_3$, denote by $\fiber^{t+1}(\xi)$ the 
$(t+1)$-dimensional fiber of $\CAD_{t+4}$ over $\xi$, and
for each point $(\xi,\delta) \in \ES_3\times\ES_t$ denote by $\fiber^1(\xi,\delta)$ 
the one-dimensional fiber of $\CAD_{t+4}$ over $(\xi,\delta)$.
Each zero-dimensional cell (a point) in $\fiber^1(\xi,\delta)$ that lies in $Z(G)$ 
corresponds to (the value of $\alpha$ of) an intersection point of the curve $\curve_\delta$ 
with the plane $h_\xi$. For all pairs $(\xi,\delta)$ in a cell $\psi \in \CAD_{t+3}$, 
$\fiber^1 (\xi,\delta)$ has a fixed combinatorial structure. In particular, all of these
pairs have the same number of zero-dimensional cells in their fibers, and thus the same number 
of intersection points. Among the cells of $\CAD_{t+4}$ that lie in $Z(G)$ and in the cylinder over $\psi$, i.e.,
they project 
to $\psi$, the $i$th cell, say, $C$, corresponds to the $i$th intersection point of 
$\curve_\delta\cap h_\xi$ for $(\xi,\delta)\in\psi$. Therefore the $i$th intersection point 
of $\sigma_\delta$ abd $h_\xi$ (or rather the value of $\alpha$ corresponding to that point) has a 
fixed semi-algebraic encoding 
$\varphi_C(\xi;\delta)$ of constant complexity, over $(\zeta,\delta)\in \psi$, which will be used for expressing 
the intersection condition.

Let $\CAD^0$ be the collection of the cells of $\CAD_{t+4}$ that lie in $Z(G)$, 
and, for a base cell $\psi \in \CAD_3$, let $\CAD_\psi^0 \subset \CAD^0$ be the 
subset of cells $C$ with $C^\Downarrow = \psi$.
Fix a cell $C\in\CAD^0$ and an arc $\gamma = (\delta, \alpha^-,\alpha^+) \in \Gamma$.
We describe semi-algebraic predicates $\Pi_{C,\gamma}^{(i)}$, for $1 \le i \le 5$. For a triangle 
$\Delta=(\xi,\mu,\nu)$ with $\xi \in C^\Downarrow$, $\Pi_{C,\gamma}^{(i)}=1$ if and only if the condition 
($I_i$) holds for the intersection point 
of $h_\Delta$ and the curve $\curve_\delta$ corresponding to $C$. 
Set $\varphi_{C,\gamma} (\xi)  = \varphi_C(\xi, \delta)$.

\paragraph{The predicate $\Pi_{C,\gamma}^{(1)}$.}
Since $\xi\in C^\Downarrow$, there exists an $\alpha\in\reals$ such that $\curve_\delta(\alpha)$ is an intersection point
with $h_\xi$ corresponding to $C$ if and only if 
$(\xi,\delta)$ lies in $C^\downarrow \in \CAD_{t+3}$, i.e., $\delta = (\delta_1,\ldots,\delta_t)$ 
satisfies equalities or inequalities of the form
\begin{align} \label{eq:abc-fiber}
\delta_1 & = f_1(\xi) & \text{or} && f_1^-(\xi) &< \delta_1 < f_1^+(\xi) \nonumber \\
\delta_2 & = f_2(\xi;\delta_1) & \text{or} && f_2^-(\xi;\delta_1) &< \delta_2 < f_2^+(\xi;\delta_1) \\
& & \vdotswithin{\text{or}} & \nonumber \\
\delta_t & = f_t(\xi;\delta_1,\ldots,\delta_{t-1}) & \text{or} &&
f_t^-(\xi;\delta_1,\ldots,\delta_{t-1}) &< \delta_t < f_t^+(\xi;\delta_1,\ldots,\delta_{t-1}) \nonumber ,
\end{align}
for suitable constant-degree continuous algebraic functions $f_1,f_1^-,f_1^+,\ldots,f_t,f_t^-,f_t^+$ over $C^\Downarrow$. 
Since $\gamma$ and thus $\delta$ is fixed, each $f_i$ can be regarded as being defined over $\ES_3$.
Therefore this set of equalities and inequalities defines the desired
semi-algebraic predicate $\Pi_{C,\gamma}^{1}$ in $\ES_3$. 

\paragraph{The predicate $\Pi_{C,\gamma}^{(2)}$}
To test the condition ($I_2$) for the intersection point corresponding to $C$, we simply need to check whether
$\varphi_{C,\gamma}(\xi)$ lies between $\alpha^-$ and $\alpha^+$. So  we define the predicate $\Pi_{C,\gamma}^{(2)}$ as 
\[
\Pi_{C,\gamma}^{(2)} (\xi)  \coloneqq \alpha^- \le \varphi_{C,\gamma}(\xi) \le \alpha^+.
\]
Since $\varphi_{C,\gamma}$ is an algebraic function of bounded degree, $\Pi_{C,\gamma}^{(2)}$ is a constant-complexity 
semi-algebraic predicate over $\ES_3$.

\paragraph{The predicates $\Pi_{C,\gamma}^{(3)}$, $\Pi_{C,\gamma}^{(4)}$, $\Pi_{C,\gamma}^{(5)}$.}
In the preceding two predicates, we only used $\xi$, the three parameters that define the plane supporting $\Delta$. 
We now use, sparingly, the parameters that define its boundary edges, one edge per predicate. 
For a point $\xi\in \ES_3$, let $\chi_{C,\gamma} (\xi) = \curve_\delta (\varphi_{C,\gamma}(\xi))$ 
be the intersection point of $\curve_\delta$ and $h_\xi$ corresponding to $C$ (this point is 
well defined only if ($I_1$) holds). For $i=1,2,3$, let $(\mu_i,\nu_i)$ denote the two parameters 
that specify (the line supporting) the edge $e_i$ (within the plane $h_\Delta$). Hence, we need 
to test whether $\chi_{C,\gamma}(\xi))$ lies on the positive side of $\ell_i$, for each $i=1,2,3$. 
We thus define the predicate $\Pi_{C,\gamma}^{(i+2)}(\xi,\mu_i,\nu_i)$ that is $1$ if
$\chi_{C,\gamma}(\xi))$ lies on the positive side of $\ell_i$ and $0$ otherwise.
Since $\chi_{C,\gamma}$ is an algebraic function of bounded degree, 
$\Pi_{C,\gamma}^{(i+2)}$ is a semi-algebraic predicate of constant complexity 
in a five-dimensional parametric space, each point of which represents a plane 
$h_\xi$ and a line $\ell_i$ within that plane.

Set $\Pi_{C,\gamma} (\Delta^*) = \bigwedge_{i=1}^5 \Pi_C^{(i)} (\Delta^*)$. 
For a base cell $\psi \in \CAD_3$, let  
$\T_\psi \coloneqq \{ \plate \in  \T \mid \plate^* \in \psi\}$.
Putting all pieces together, we obtain the following lemma:
\begin{lemma}
	\label{lem:tri-rpd}
	For any base cell $\psi\in\CAD_3$ and for any cell $C\in\CAD^0_\psi$, 
	the intersection condition between an arc of $\Gamma$ and a triangle of $\T_\psi$ can be expressed as 
	a Boolean predicate in which each polynomial inequality uses at most five of the nine parameters 
	defining the triangle.
\end{lemma}

We construct an arc-intersection-searching data structure for $\T$ as follows. For each base cell 
$\psi \in \CAD_3$ and for every $C \in \CAD_\psi^0$, we preprocess $\T_\psi$ into a data structure of size $O^*(|\T_\psi|)$ for answering $\Pi_{C,\gamma}$-queries, as described in Section~\ref{subsec:Pi-lin}.

To answer an intersection query with an arc $\gamma\in \Gamma$, for each cell $C\in \CAD^0$, we query the data structure
$\Psi_C$ with $\gamma$. By Lemma~\ref{lem:tri-rpd} and Theorem~\ref{lem:Pi-lin}, the query time is $O^*(n^{4/5})$.
We note that for each intersection point of $\gamma$ and a triangle $\Delta$, there is a unique CAD cell at which 
this intersection point will be reported, so our data structure can also answer intersection-counting queries.
In summary, we obtain the following result.
\begin{theorem}
  \label{thm:linsto}
  A set $\T$ of $n$ triangles in $\reals^3$ can be processed, in expected $O^*(n)$ time, into a 
  data structure of size $O^*(n)$, so that an arc-intersection query amid the triangles of $\T$,
  with arcs from some family of constant-degree algebraic arcs, can be answered in $O^*(n^{4/5})$ time.
\end{theorem} 
\section{Plate-Intersection Queries amid Arcs}
\label{sec:lines}

We now move to the second type of intersection queries, in which the roles of
the input and query objects are interchanged: the input now consists of a 
collection of constant-degree algebraic arcs, and we query with plates of constant complexity.
Let $\Gamma$ be a set of $n$ constant-degree parametrized algebraic arcs as defined in~\eqref{eq:arc-family},
and let $\TT$ be a family of plates in $\reals^3$ as defined in~\eqref{eq:plate-family} in Section~\ref{sec:main}.
The goal is the same: answer intersection queries on the input arcs of $\Gamma$ with plates in $\TT$.
We begin by considering a simple case in which the input consists of a set of lines in $\reals^3$, 
then adapt this data structure for the case when the input consists of a set of line segments in $\reals^3$, 
and finally extend this data structure to handle the general case when the input consists of a set of 
constant-degree (not necessarily planar) algebraic arcs in $\reals^3$. As in the previous sections, for concreteness, we focus on 
intersection-detection queries, but the data structure can also answer reporting and counting queries as well,
with similar performance bounds. 

\subsection{The case of lines}
\label{subsec:line-input}

Let $L$ be a set of $n$ lines in $\reals^3$.
We construct a partition tree $\Psi$ on $L$ based on the polynomial partitioning technique 
of Guth~\cite{Guth}, as in Section~\ref{sec:main}. 
Each node $v\in \Psi$ is associated with a cell $\tau_v$ of some partitioning polynomial 
and a subset~$L_v\subseteq L$ of lines that cross~$\tau_v$. If $v$ is the root of $\Psi$, then 
$\tau_v=\reals^3$ and $L_v= L$. Set $n_v \coloneqq |L_v|$. We fix a sufficiently large constant~$D$ and
a threshold parameter $n_0 \le n$; for now $n_0$ is assumed to be a constant but later we will 
set it to depend on $n$. 

Suppose we are at a node $v$. If $n_v\le n_0$ then $v$ is a leaf and we store $L_v$ at $v$.
Otherwise we construct, in linear time, a partitioning polynomial $F_v$ of degree $O(D)$,
so that each cell of $\reals^3\setminus Z(F_v)$ is crossed by at most $n/D^2$ lines of $L$; 
see~\cite{AAEZ,AEZ}. The number of cells in $\reals^3\setminus Z(F_v)$ is $O(D^3)$.
Let $L_v^0 \subseteq L_v$ be the set of lines that are contained in~$Z(F)$.
We construct a secondary data structure $\ZDS_v$, described later in this section, to answer intersection queries:
(i)~on~the~lines of $L_v^0$ with any plate in $\TT$, and (ii)~on~$L_v$ with 
plates that lie in $Z(F)$. By Lemma~\ref{lem:line-zero} below, 
these structures require $O^*(n_v)$ space and can answer a query in $O^*(n_v^{1/2})$ time.

Let $\tau$ be a cell of $\reals^3\setminus Z(F_v)$. We create a child $w_\tau$ of $v$ associated 
with $\tau$. Let $L_\tau \subseteq L_v \setminus L_v^0$ be the set of lines of $L_v$ that intersect 
$\tau$. As in Section~\ref{sec:main}, we call a plate $\plate$ \emph{narrow} at 
$\tau$ if an edge of $\plate$ crosses $\tau$ and \emph{wide} if $\plate$ crosses $\tau$ but none of its edges does. 
We construct a secondary data structure $\WDS_\tau$ on the lines of $L_\tau$ for answering intersection 
queries (within $\tau$) with a plate that is wide at $\tau$. This structure too uses $O^*(n_v)$ space 
and answers a query in $O^*(n_v^{1/2})$ time (see Lemma~\ref{lem:line-wide} below). Finally, we 
recursively construct the subtree $\Psi_{w_\tau}$ on $L_\tau$ 
(to answer intersection queries with a plate that is narrow at $\tau$) and attach it to $w_\tau$.

Let $\plate$ be a query plate. An intersection query with $\plate$ is answered by visiting $\Psi$ in a 
top-down manner. Suppose we are at a node $v$. If $v$ is a leaf, then we test all lines of $L_v$ for 
intersection with $\plate$. If $v$ is an interior node, we first query $\ZDS_v$ to test whether $\plate$ 
intersects any line of $L_v^0$. If $\plate \subset Z(F)$, we again use $\ZDS_v$ to test whether $\plate$ 
intersects any line of $L_v\setminus L_v^0$. If $\plate\not\subset Z(F)$, 
we go over each cell $\tau$ of $\reals^3 \setminus Z(F)$ 
that $\plate$ crosses. If $\plate$ is wide at $\tau$ we query $\WDS_\tau$ with $\plate$. Otherwise $\plate$ 
is narrow at $\tau$, and we recursively search with $\plate$ in the subtree of $\Psi_v$ corresponding to $\tau$.
We stop as soon as one of the subprocedures detects an intersection.

This completes the description of the overall data structure $\Psi$ and the query procedure. It remains to
describe the two secondary data structures, and to analyze the performance of $\Psi$.

\paragraph{Querying with wide plates.}

Let $L$ be a set of $n$ lines in $\reals^3$, $F$ a partitioning polynomial, and $\TT$ the family of query plates.
For each cell $\tau$ of $\reals^3 \setminus Z(F)$, we construct a data structure $\WDS_\tau$ for preprocessing 
the subset $L_\tau$ of lines crossing $\tau$ that supports intersection queries (within $\tau$) with plates
in $\TT$ that are wide at $\tau$.

We construct a CAD $\CAD$ of $\reals^3$ induced by $F$
(see again \cite{BPR,Col,SS2} and Section~\ref{subsec:CAD} for details). 
Each cell $\pi$ of $\CAD$ is fully contained in some cell of $\A(F)$, so $\CAD$ is a refinement of $\A(F)$.
As in Section~\ref{sec:wide}, each three-dimensional cell $\pi$ is a
prism-like cell (we refer to it simply as a \emph{prism}) of the form
\begin{align*} 
\alpha_1 & < x < \alpha_2 \\ 
f_1(x) & < y < f_2(x) \\
g_1(x,y) & < z < g_2(x,y) ,
\end{align*}
where $\alpha_1$, $\alpha_2$ are reals, and $f_1$, $f_2$, $g_1$, $g_2$ 
are continuous algebraic functions of constant degree (which depends on~$D$), each defined over the range determined by the preceding inequalities; 
some of $\alpha_1$, $\alpha_2$ and these functions might be $\pm\infty$. 
(As before, some of these inequalities are equalities for lower-dimensional cells of $\CAD$.)
In general, $\pi$ has six two-dimensional faces,
contained respectively in the algebraic surfaces $x=\alpha_1$, $x=\alpha_2$,
$y = f_1(x)$, $y = f_2(x)$, $z = g_1(x,y)$ and $z = g_2(x,y)$. Each face
is simply connected, monotone (with respect to a suitable coordinate plane),
and of constant complexity (again, which depends on~$D$); $\pi$ has fewer faces 
when some of $\alpha_1$, $\alpha_2$, $f_1$, $f_2$, $g_1$, $g_2$ are $\pm\infty$.

We build a data structure $\WDS_\pi$ for each cell $\pi$ of $\CAD$, and denote by $\WDS_\tau$ 
the collection of the structures $\WDS_\pi$ over all cells $\pi\in\CAD$ that are contained in $\tau$. 
We focus on three-dimensional cells---the data structure is trivial for 
zero- and one-dimensional cells, and it is similar and simpler for two-dimensional cells. 
Let $\pi$ be a three-dimensional CAD prism that is contained in $\tau$. Let $L_\pi$ be the set of line segments 
obtained by clipping the lines of $L$ within $\pi$. Note that a line may contribute up to $O(1)$ 
segments to $L_\pi$. 

Let $\plate$ be a plate that is wide at $\tau$ (if it crosses $\pi$). Then clearly $\plate$ is also wide at $\pi$.
Let $\pi(\plate)$ be the subdivision of $\pi$ induced by $\plate$ (i.e., each face of $\pi(\plate)$ is the maximal 
connected portion of a face of $\pi$ that either lies in $\plate$ or does not intersect $\plate$).
It is important that we decompose $\pi$ by $\plate$ and not by the plane 
supporting $\plate$---see below. The complexity of $\pi(\plate)$ is a constant that depends 
on $D$. The following lemma, in which we assume general position of the segments, is the crucial observation:
\begin{lemma} \label{lem:cutcad} 
A segment $e$ of $L_\pi$ intersects $\plate$ if and only if its endpoints lie
on the boundary of different three-dimensional cells of $\pi(\plate)$.
\end{lemma} 
\begin{proof}
  Since $\plate$ is wide at $\pi$, its intersection $\plate\cap\pi$ consists
  of one or several connected regions, all fully contained in the relative interior of~$\plate$.
  Consequently, each of these pieces fully slices~$\pi$. Informally,
  the purpose of the lemma is to argue that a point on one side of such a slice 
  cannot reach a point on the other side, within $\pi$, without crossing the slice. 

  The segment $e$ lies inside $\pi$, so if the endpoints
  of $e$ lie in different cells of $\pi(\plate)$, then $e$ has to intersect 
  $\plate$ to go from one cell to another; it has to be through $\plate$ 
  because the relative interior of $e$ does not meet the boundary of~$\pi$. On the other 
  hand, assume that the two endpoints lie in the same three-dimensional 
  cell $\psi$ but $e$ intersects $\plate$. At each such intersection, 
  $e$ has to move from one three-dimensional cell of $\pi(\plate)$ to 
  another cell. See Aronov et al.~\cite{ArMS} for a (nontrivial) proof 
  of this seemingly obvious property. Informally, it holds because each
  cell of~$\pi(\plate)$ is topologically a ball. We note that (i) this 
  property also holds, as shown in \cite{ArMS}, for arbitrary connected 
  arcs $e$, and (ii) this property may fail for more general, non-CAD 
  cells $\pi$ (such as, e.g., a torus).
  See Figure~\ref{fig:epidelta} for an illustration for the case of a CAD cell.
  It follows that $e$ leaves $\psi$ when it crosses $\plate$, and has
  to return to $\psi$, necessarily at a second such crossing. Thus
  $e$ intersects $\plate$ (at least) twice, which is impossible. (Note 
  that this last part of the argument fails when $e$ is not straight, 
  as illustrated in Figure~\ref{fig:curveincad}.)
\end{proof}

\begin{figure}[htb]
  \centering
  \scalebox{0.85}{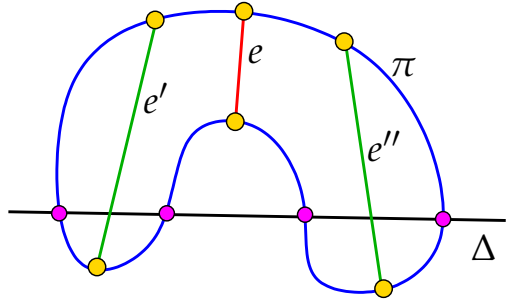}
  \caption{A two-dimensional rendering of the scenario analyzed in Lemma~\ref{lem:cutcad}. 
    The red segment $e$ has endpoints in the same three-dimensional cell of $\pi(\plate)$ 
    and does not cross $\plate$, whereas each of the green segments 
    $e'$, $e''$ has endpoints in different cells of $\pi(\plate)$ and 
    crosses $\plate$.}
  \label{fig:epidelta}
\end{figure}

\medskip

\noindent
\textbf{\textit{Data structure.}}
Exploiting Lemma~\ref{lem:cutcad}, we construct the data structure $\WDS_\pi$ as follows.
$\WDS_\pi$ consists of a family of $O(1)$ partition trees. For a pair of 
(not necessarily distinct) faces $\ph, \ph'$ of $\pi$, let $L_{\ph,\ph'} \subseteq L_\pi$ 
be the subset of segments of $L_\pi$ having one endpoint that lies on $\ph$ and the other on $\ph'$.
Let $E_\ph^{\ph'}$ (resp., $E_{\ph'}^\ph$) be the set of endpoints of the segments of $L_{\ph,\ph'}$ 
that lie on $\ph$ (resp., $\ph'$).\footnote{%
  If $\ph$ and $\ph'$ are the same face, then we regard them as two copies of that face and assign one of 
  the endpoints of each segment of $L_{\ph,\ph'}$ to $\ph$ and the other to $\ph'$.}
Set $n_{\ph,\ph'} = |L_{\ph,\ph'}|$. We define a predicate $\Pi_{\ph,\ph'}$ as follows.

Let $\psi \subseteq \ph, \psi' \subseteq \ph'$ be two faces of $\pi(\plate)$ for some plate $\plate \in \TT$ 
that is wide for $\pi$, and let  $e=pq \in L_{\ph,\ph'}$ be a segment with $p\in E_{\ph}^{\ph'}$ and $q\in E_{\ph'}^\ph$.
Then $\Pi_{\ph,\ph'} (e, (\psi,\psi'))=1$ if and only if $p\in\psi$ and $q\in \psi'$. 
Clearly, $\psi, \psi'$ are semi-algebraic sets of constant complexity.
Note that the reduced parametric 
dimension of $L_{\ph,\ph'}$ with respect to $\Pi_{\ph,\ph'}$ is $2$ since each polynomial inequality uses only one of the endpoints of $e$ and one requires two real parameters to describe each endpoint (as it lies on a fixed $2$-dimensional variety).
We preprocess $L_{\ph,\ph'}$, in $O^*(n_{\ph,\ph'})$ expected time, into a data structure $\WDS_{\ph,\ph'}$ of size 
$O^*(n_{\ph,\ph'})$ for answering $\Pi_{\ph,\ph'}$-queries, as described in Section~\ref{subsec:Pi-lin},
so that a query can be answered in $O^*(n_{\ph,\ph'}^{1/2})$ time (see Theorem~\ref{lem:Pi-lin}).
We construct such a data structure for each pair of faces of $\pi$. We then repeat this procedure 
for all cells of $\CAD$ that lie in $\tau$ and for all cells $\tau$ of $\reals^3 \setminus Z(F)$. 
The total size and expected preprocessing time of the data structure are $O^*(n)$.

\paragraph*{Query procedure.}
Let $\plate\in\TT$ be a plate that is wide at a cell $\tau \in \reals^3\setminus Z(F)$. For each CAD cell $\pi\subset\tau$, 
we do the following. First, we construct the decomposition $\pi(\plate)$ of $\pi$ induced by $\plate$. 
Let $\psi, \psi'$ be a pair of faces of $\pi(\plate)$
that lie on the boundaries of two different 
three-dimensional cells of $\pi(\plate)$. 
Let $\ph, \ph'$ be the faces of $\bd\pi$ containing $\psi$ and $\psi'$, respectively.
We query the partition tree $\WDS_{\ph,\ph'}$ with the pair $(\psi,\psi')$  of semi-algebraic ranges to detect 
whether any segment of $L_{\ph,\ph'}$ has one endpoint in $\psi$ and the other in $\psi'$.
We repeat this step for all such pairs of faces of $\pi(\plate)$. By Lemma~\ref{lem:cutcad}, one of these queries 
returns yes if and only if $\plate$ intersects a segment of $L_\pi$.
Furthermore if we return all segments that satisfy the predicate, each segment in the output of any of these 
sub-queries crosses $\plate$ (within $\pi$), and every segment that crosses $\plate$ 
within $\pi$ appears in such an output exactly once over all sub-queries.
As already noted, the cost of a query is $O^*(n^{1/2})$. 
More generally, for a storage parameter $s\in [n,n^\tQ]$, a query can be answered in 
$O^*\biggl ((n/s^{1/\tQ})^{\tfrac{1/2}{1-1/\tQ}} \biggr )$ time using $O^*(s)$ storage and 
preprocessing; see Theorem~\ref{thm:Pi-tradeoff}.
That is, we obtain:
\begin{lemma}
\label{lem:line-wide}
Given a set $L$ of $n$ lines in $\reals^3$ and a partitioning polynomial $F$, 
a data structure of size $O^*(n)$ can be constructed, in $O^*(n)$ expected time, 
so that for a cell $\tau$ of $\reals^3\setminus Z(F)$ and a plate $\plate$ that is wide at $\tau$, 
an intersection query on $L$ (within $\tau$) with $\plate$ can be answered in $O^*(n^{1/2})$ time.
More generally, for a storage parameter $s\in[n,n^\tQ]$, where $\tQ$ is the reduced parametric 
dimension of the query plates, a query can be answered in 
$O^* \biggl ( (n/s^{1/\tQ})^{\tfrac{1/2}{1-1/\tQ}} \biggr)$ time 
using $O^*(s)$ storage and expected preprocessing time.
\end{lemma}
\paragraph{Handling input lines and queries on the zero set.}
Let $L$ be a set of $n$ lines in $\reals^3$ and $F$ a partitioning polynomial, as above.
We compute the $O(1)$ irreducible components of $Z(F)$ using any of the known 
algorithms~\cite{BPR,CLO,GV95,Ka92}, and work with each irreducible component separately.  
Henceforth, without loss of generality, we assume that $Z(F)$ is irreducible.

First consider the case when the lines in $L$ do not lie on $Z(F)$ but the query plate $\plate$ does. 
Since $\plate$ is planar, it follows that $Z(F)$ is a plane. 
Let $E$ be the set of intersection points of the lines of $L$ with $Z(F)$. 
The intersection query on $L$ with a plate 
$\plate \subset Z(F)$ is equivalent to a (planar) range query on $E$ with $\plate$.
Hence, for a storage parameter $s\in[n,n^\tQ]$, where $\tQ$ is the reduced parametric dimension of query plates,
we can construct, in $O^*(s)$ expected time, a data structure of $O^*(s)$ size so that a 
query can be answered in $O^* \biggl((n/s^{1/\tQ})^{\tfrac{1/2}{1-1/\tQ}}\biggr )$ time, 
using Theorem~\ref{thm:Pi-tradeoff}.

Next, we describe the data structure for handling the lines of $L$ that are contained in $Z(F)$.
If $Z(F)$ is not a ruled surface\footnote{%
A \emph{ruled surface} in $\reals^3$ can be described as the set of points swept by a moving straight line. 
A surface is \emph{doubly ruled} if there are two distinct lines through every one of its points 
that lie on the surface. See~\cite{Guth-ruled} for details.
}
then it can contain at most $O(D^2)$ lines 
(this is the Cayley-Salmon theorem; see, e.g.,~\cite{GK},
and see~\cite{Guth-ruled} for a review of ruled surfaces), so only 
$O(1)$ lines of $L$ can lie in $Z(F)$.
For each line $\ell \in L$, we check, during preprocessing, whether $\ell$ is contained in $Z(F)$ 
(say, using B\'ezout's theorem). If the number of such lines does not exceed the Cayley-Salmon threshold, 
then we simply store this set of $O(1)$ lines with $F$. Otherwise, we conclude that $Z(F)$ must be a 
ruled surface and proceed as described below. 
Without loss of generality, assume, for simplicity, that all lines in $L$ lie on $Z(F)$.

If $Z(F)$ is a plane, then intersection searching on $L$ with a plate (which may or may not lie in $Z(F)$)
reduces to an instance of two-dimensional semi-algebraic range searching because the lines lying in $Z(F)$
can be specified by two parameters. Hence, as above, we can answer this query in $O^*(n^{1/2})$ time 
using $O^*(n)$ space and preprocessing.
Consider then the case where the component $Z(F)$ is singly or doubly
ruled by lines. Suppose for specificity that $Z(F)$ is singly ruled; doubly ruled components can be handled 
by a simpler variant of this argument. As observed in \cite{GK}, for example, with the exception of at most 
two lines, all lines that are contained in $Z(F)$ belong to the single 
ruling family, and these lines are parametrized by a single real 
parameter $\theta$. We form the set of the values of $\theta$ that correspond
to the input lines, and preprocess them into a trivial one-dimensional
range searching structure, over the parameter $\theta$, which uses $O^*(n)$ storage and $O(\log n)$
query time. We map a query plate $\plate$ into a range that is a
union of a constant number of intervals along the $\theta$-axis,
representing the values of $\theta$ for which the corresponding line in the ruling crosses~$\plate$.
We then query our structure with each of these intervals.
Testing the (at most) two exceptional lines takes $O(1)$ time.
Putting everything together, we obtain the following:
\begin{lemma}
\label{lem:line-zero}
Given a set $L$ of $n$ lines in $\reals^3$ and a partitioning polynomial $F$, a data structure of size 
$O^*(n)$ can be constructed, in $O^*(n)$ expected time, so that for plates of constant complexity,
a plate-intersection query amid the lines of $L$ that lie in $Z(F)$,
or amid all lines of $L$ with a query plate that lies in $Z(F)$, 
	can be answered in $O^*(n^{1/2})$ time. More generally, for a storage 
	parameter $s\in [n,n^\tQ]$, an intersection query can be answered in 
	$O^* \biggl ((n/s^{1/\tQ})^{\tfrac{1/2}{1-1/\tQ}} \biggr )$
	time using $O^*(s)$ storage and expected preprocessing time.
\end{lemma}
\paragraph*{Analysis of the full procedure.}

For a node $v$, let $S(n_v)$ be the maximum size of the subtree $\Psi_v$ rooted at $v$ 
and constructed on the $n_v$ lines of $L_v$. Since the secondary data structures stored at $v$ use $O^*(n_v)$ 
space, and each cell of $\reals^3 \setminus Z(F_v)$ intersects at most $n_v/D^2$ lines, where $F_v$
is the partitioning polynomial used at $v$, we obtain the following recurrence for $S(n_v)$:
\[
S(n_v) \le 
 \begin{cases*}
	 \displaystyle c_2D^3 \cdot S\left (\frac{n_v}{D^2}\right ) + c_3 n_v^{1+\delta} & for $n_v > n_0$, \\[1mm]
    	  c_4  n_v & for $n_v\le n_0$,
  \end{cases*}
\]
where $c_2, c_3, c_4, \delta$ are constants analogous to those in~\eqref{eq:storage}.
Following the same analysis as in Section~\ref{subsec:data-structure}, 
the solution of the above recurrence is $S(n_v) = O(n_v^{3/2+\eps})$ for any constant $\eps\ge\delta$, 
using the fact that the threshold $n_0$ is a constant. 
The overall size of the data structure is thus $O^*(n^{3/2})$. A similar argument shows that the expected 
preprocessing time is $O^*(n^{3/2})$.

The maximum cost $Q(n_v)$ of a query at $\Psi_v$ obeys the following recurrence:
\[
Q(n_v) \le 
\begin{cases*}
	\displaystyle c_5 D \cdot Q \left (\frac{n_v}{D^2}\right) + c_6 n_v^{1/2+\delta} & for $n_v> n_0$, \\[1mm]
    	  c_7 n_v & for $n_v\le n_0$,
\end{cases*}
\]
where $c_5, c_6, c_7, \delta$ are constants analogous to those in~\eqref{eq:query}.
Following an analysis similar to that in Section~\ref{subsec:plate-query}, the overall query time is $O^*(n^{1/2})$.
In summary, we have shown:

\begin{lemma}
  \label{thm:lines}
  A set $L$ of $n$ lines in $\reals^3$ can be preprocessed into a data structure of size~$O^*(n^{3/2})$, 
  in expected time $O^*(n^{3/2})$, so that, for any query plate $\plate$ of constant complexity, 
  we can perform an intersection query with $\Delta$ in $L$ in $O^*(n^{1/2})$ time.
\end{lemma} 
\paragraph{Space/query-time trade-off.}
We obtain a space/query-time trade-off as in Section~\ref{sec:trade-off}, by combining the above partition tree in
$\reals^3$ with the general semi-algebraic-relation searching technique described in Appendix~\ref{app:multi-level}.
Let $s \in [n,n^{\tQ }]$ be a storage
parameter, where $\tQ $ is the reduced parametric dimension of the plates in $\TT$ for the intersection predicate with lines.
We choose the same threshold value $n_0 \coloneqq n_0(s)$ as in~\eqref{eq:threshold}. We recursively 
construct the partition trees in $\reals^3$ as just described 
until we reach a node $v$ with $n_v \le n_0$, except that we construct the secondary structures at a node $v$ with the storage parameter $s_v$ associated with $v$; $s_v=s$ for the root, and for a child $w$ of $v$, 
$s_w = s_v/D^3$. For a leaf $v$, with $n_v \le n_0$, we proceed as follows:
If $s\le n^{3/2}$, we construct a data structure of size $O^*(n_v)$ 
on $L_v$ for plate-intersection searching as described in Appendix~\ref{subsec:Pi-lin}.
Since a line in $\reals^3$ has parametric dimension~$4$,
a plate-intersection query for $L_v$ is answered in $O^*(n_v^{3/4})$ query time (see Theorem~\ref{lem:Pi-lin}).
Next, if $s>n^{3/2}$,  we construct a data structure on $L_v$ of size $O^*(n_v^\tQ)$ so that a plate-intersection query can be answered in $O^*(1)$ time. 

Following the analysis in Section~\ref{sec:trade-off}, the total size of the data structure is $O^*(s)$,
as is the expected preprocessing time. 
As for the query time, we obtain the following recurrence:
\begin{equation}
  \label{eq:line-plate-query}
  Q(n_v,s_v) \le
  \begin{cases*}
	  \displaystyle c_5 D \cdot Q\left (\frac{n_v}{D^2},\frac{s_v}{D^3}\right ) + 
	  c_6 \left (\frac{n_v}{s_v^{1/\tQ}}\right)^{\frac{1/2}{1-1/\tQ}}n_v^\delta & for $n_v> n_0$, \\[1mm]
	  c_7 n_v^{1/2+\delta} & for $n_v\le n_0$ and $s\le n^{3/2}$,\\[1mm]
	  c_7 n_v^{\delta} & for $n_v\le n_0$ and $s> n^{3/2}$,
  \end{cases*}
\end{equation}
for similar constants $c_5, c_6, c_7$, and $\delta$.
Unlike \eqref{eq:moderate-query} in Section~\ref{sec:trade-off} the recursive term dominates in 
\eqref{eq:line-plate-query}, i.e., the total overhead term at the children of $v$ is larger than that at $v$,
so the overall query time is dominated by the time spent at the leaves of the partition tree. 
The query procedure visits  $O^*((n/n_0)^{1/2})$ leaves and spends $O^*(n_0^{3/4})$ time at each leaf (including the 
time spent at the secondary structures).  Hence, the overall query time is 
$O^*(n^{5/4}/s^{1/2})$ for $s \le n^{3/2}$, and
$O^* \biggl (\frac{n}{s^{1/\tQ}}\biggr)^{\frac{1}{2-3/\tQ}} \biggr )$ for $s>n^{3/2}$.
Hence, we obtain the following result.
\begin{theorem}
  \label{thm:lines-tradeoff}
Let $\TT$ be a family of plates in $\reals^3$ of reduced parametric dimension $\tQ$ 
(with respect to lines). For a storage parameter $s\in[n,n^{\tQ }]$, 
a set $L$ of $n$ lines in $\reals^3$ can be preprocessed into a data structure
of size~$O^*(s)$, in expected time $O^*(s)$, so that, for any query 
plate $\plate\in\TT$, an intersection query with $\plate$ on $L$ can be answered in time
\[
O^* \biggl ( \frac{n^{5/4}}{s^{1/2}} + \biggl (\frac{n}{s^{1/\tQ}}\biggr)^{\frac{1}{2-3/\tQ}} \biggr ).
\]
\end{theorem}
One can easily verify that the first (resp., second) term dominates when $s \le n^{3/2}$ (resp., $s\ge n^{3/2}$).
\subsection{The case of segments}
\label{subsec:seg-input}

Next, we show how we adapt the above data structure to answer plate-intersection queries on a set 
$\E$ of $n$ segments in $\reals^3$. A segment $e=pq$ intersects a plate $\plate$ if and only if 
(i) the endpoints $p$ and $q$ lie on opposite sides of the plane $h_\plate$ supporting $\plate$, 
and (ii) the line $\ell_e$ supporting $e$ intersects $\plate$. 
Let $h_\plate^+, h_\plate^-$ be the two halfspaces bounded by $h_\plate$.
For a line $\ell$ and a plate $\plate$, let $\Pi (\ell, \plate)$ be the predicate that is $1$ 
if and only if $\ell$ intersects $\plate$. The intersection condition between $e$ and $\plate$ 
can thus be expressed as:
\begin{itemize}
\item[(i)] $p \in h_\plate^+$, $q \in h_\plate^-$, and $\Pi(\ell_e,\plate) =1$, or
\item[(ii)] $p \in h_\plate^-$, $q \in h_\plate^+$, and $\Pi(\ell_e,\plate) =1$.
\end{itemize}

Since the intersection condition for segments is obtained by augmenting the intersection condition with 
two polynomial inequalities, each of which uses three parameters of the input segments and of the query plates, 
following the analysis in Section~\ref{subsec:Pi-tradeoff}, we obtain the following result:
\begin{theorem}
  \label{thm:segs-tradeoff}
	Let $\TT$ be a family of plates in $\reals^3$ of reduced parametric dimension $\tQ \ge 3$ (with respect to lines). For a storage parameter 
$s\in[n,n^{\tQ}]$, a set $\E$ of $n$ line segments in $\reals^3$ can be preprocessed into a data structure
of size~$O^*(s)$, in expected time $O^*(s)$, so that an intersection query with a plate $\plate\in\TT$ 
on $\E$ can be answered in time
\[
O^* \biggl ( \frac{n^{5/4}}{s^{1/2}} + \biggl (\frac{n}{s^{1/\tQ }}\biggr)^{\frac{1}{2-3/\tQ}} \biggr ).
\]
\end{theorem} 
\subsection{The case of arcs}
\label{subsec:arc-input}

Finally, let $\G$ be a set of $n$ constant-degree algebraic arcs in $\reals^3$ of reduced 
parametric dimension $\tO$ (with respect to $\TT$), such that any plane contains at most $O(1)$ arcs of $\G$. 
Our overall data structure for answering plate-intersection queries on $\G$ is the same 
as the one in Section~\ref{subsec:line-input}, but we need to adapt various substructures 
so that they can handle arcs (instead of lines). Since the reduced parametric dimension of $\G$ is $\tO$, 
we can construct a data structure of $O^*(n)$ size that can answer a plate-intersection query on $\G$ 
in $O^*(n^{1-1/\tO})$ time. Note that if $\G$ is a set of planar arcs, then we can use the ideas of 
Section~\ref{subsec:planar_imp} to reduce the value of $\tO$ to the reduced parametric dimension of 
the curves supporting the arcs of $\G$. When the threshold value $n_0$ is not a constant, 
we build this structure at each leaf of the partition tree $\Psi$. The more challenging part,
however, is to adapt the secondary data structures $\WDS$ and $\ZDS$, which we describe next.

\paragraph{Answering queries with wide plates.}
Let $\G$ be a set of $n$ arcs as above and $F$ a partitioning polynomial of degree at most $c_1D$,
for an absolute constant $c_1$ and a sufficiently large constant parameter $D$. 
For each cell $\tau$ of $\reals^3\setminus Z(F)$, we build a data structure $\WDS_\tau$ that 
can answer intersection queries on $\G$ with plates that are wide at $\tau$.
We face the following major new issue, which did not arise in the approach described 
in Section~\ref{subsec:line-input}. 
When the input objects were line segments with their endpoints lying on $\bd\tau$, 
Lemma~\ref{lem:cutcad} provided us with a necessary and sufficient condition for a 
query plate $\plate$ to intersect such a segment $e$, namely, that the endpoints of 
$e$ lie in different cells of the arrangement $\pi(\plate)$. However, when the input 
objects are curved arcs, this criterion remains sufficient but in general 
not necessary; see Figure~\ref{fig:curveincad} for an illustration.

\begin{figure}[htb]
  \centering
  \scalebox{0.75}{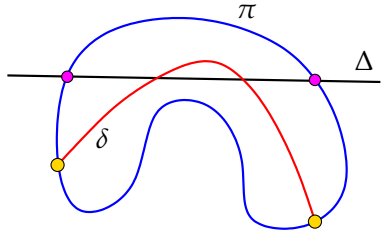}
  \caption{Lemma~\ref{lem:cutcad} may fail when the input consists of curved arcs.
    The arc $\gamma$ has both endpoints in the same cell of $\pi(\plate)$ but it
    still intersects $\plate$.}
  \label{fig:curveincad}
\end{figure}

We therefore use the following different approach, borrowing ideas and tools from Section~\ref{sec:wide}.
Let $\ES_3$, $\ES$, and $\hat F \in \reals[a,b,c,x,y]$ be the same as defined in Section~\ref{sec:wide}.
As before, we construct a CAD $\CAD_5$ of $\ES$ induced by $\hat F$. Let $\CAD_3$ be the projection 
of $\CAD_5$ onto $\ES_3$, and for a point $\xi\in\ES_3$, let $\fiber (\xi)$ be the two-dimensional 
fiber of $\CAD_5$ over $\xi$, and $\fiber^\uparrow(\xi)$ the lifting of $\fiber(\xi)$ to the plane $h_\xi$. 
Recall that $\fiber^\uparrow(\xi)$ is a 
refinement of $\A(F;h_\xi)$ into pseudo-trapezoids. Therefore, for a plate $\plate$ that is wide at $\tau$, 
$\fiber^\uparrow(\plate^*)$ includes a partition of $\plate\cap\tau$ into pseudo-trapezoids.  Each 
pseudo-trapezoid of $\fiber^\uparrow(\xi)$ corresponds to a cell $C$ of $\CAD_5$, denoted by $C^\uparrow(\xi)$,
and has a constant-size discrete label, which is its semi-algebraic representation as defined in~\eqref{eq:encoding}.  
Recall that each cell $C\in \CAD_5$ is associated with a cell $\tau \in \A(F)$, denoted by $\tau_C$.

For each cell $C \in \CAD_5$, let $\Pi_C(\gamma; \xi)$ be the semi-algebraic $C$-intersection 
predicate defined in~\eqref{eq:intersect-pred}, which, for an arc $\gamma$ and a point $\xi\in\ES_3$,
is $1$ if and only if $\xi\in C^\downarrow$ and an intersection point of $\gamma\cap h_\xi$ lies in 
the pseudo-trapezoid $C^\uparrow(\xi)$.  
Using Theorem~\ref{thm:Pi-tradeoff} of the appendix, 
for a storage parameter $s\in [n,n^3]$, we can preprocess $\G$, in $O^*(s)$ expected time, into a data structure 
of size $O^*(s)$ so that a $C$-intersection query with a pseudo-trapezoid $\plate_C$ of a plate that is wide at 
$\tau_C$ can be answered in time 
$O^*\left( (n/s^{1/3})^{\tfrac{1-1/\tO}{2/3}}\right)$, where $\tO$ is the reduced parametric dimension of 
$\G$ for $\Pi_C$.
We again remark that if the arcs in $\G$ are planar, 
then $\tO$ can be taken to be the reduced parametric dimension
of the curves supporting the arcs in $\G$.

For a cell $\tau$ of $\reals^3\setminus Z(F)$, we construct the $C$-intersection searching 
data structure for all CAD cells $C$ with $\tau_C=\tau$, and store it at $\tau$, and we repeat this 
procedure for all cells $\tau$. Since $|\CAD_5|=O(1)$, the total size of the data structure remains $O^*(s)$.

For a query plate $\plate$ that is wide at $\tau$, we answer a plate-intersection query with $\plate$
within $\tau$ by answering $C$-intersection queries on $\G$ for all $C$ with $\tau_C=\tau$,
$\plate^* \in C^\downarrow$, and $C^\uparrow(\plate^*) \subseteq \plate$. The total time in answering 
a query is $O^*\biggl( (n/s^{1/3})^{\tfrac{1-1/\tO}{2/3}} \biggr)$. That is, we have:

\begin{lemma}
\label{lem:arc-wide}
	Let $\TT$ be a family of plates of constant complexity, let $\G$ be a set of $n$ constant-degree
algebraic arcs in $\reals^3$, and let $\tO$ be the  reduced parametric dimension of the arcs in $\G$ or of the
	curves supporting 
	them if they are planar (with respect to $\TT$), 
and let $\tau$ be a cell of $\reals^3\setminus Z(F)$. For a storage parameter $s\in[n,n^3]$, 
$\G$ can be preprocessed into a data structure of size~$O^*(s)$, in expected time $O^*(s)$, 
so that, for any query plate $\plate\in\TT$ that is wide at $\tau$, a plate-intersection query 
	on $\G$ with $\plate$ within $\tau$ can be answered in $O^*\biggl( (n/s^{1/3})^{\tfrac{1-1/\tO}{2/3}} \biggr)$ time.
\end{lemma}
\paragraph{Handling the zero set.}
Consider next the task of handling input arcs or query plates that lie in $Z(F)$. 
As in Section~\ref{subsec:line-input}, without loss of generality, we can assume $Z(F)$ to be irreducible. 
First, assume that $Z(F)$ is a plane. By our general-position assumption, if $Z(F)$ is a plane then it contains only $O(1)$ input arcs, 
and we simply store them and answer a plate-intersection query on them na\"ively, by brute force. 
For arcs of $\G$ not lying in $Z(F)$, an intersection query with a plate of $\TT$ lying in $Z(F)$ 
reduces to a two-dimensional semi-algebraic range query with the query plate,
and thus can be answered in $O^*(n^{1/2})$ time using $O^*(n)$ space.

Next, we consider the case when $Z(F)$ is not a plane. Then no query plate lies in $Z(F)$, and 
we need a data structure for answering arc-intersection queries amid input arcs that lie in $Z(F)$.
So we assume that the arcs of $\G$ lie in $Z(F)$.
We follow the same recursive approach as in Section~\ref{sec:zero-set}.

Specifically, we fix a sufficiently large constant $E=D^{O(1)}$. Using the algorithm in~\cite{AAEZ}, 
we construct a partitioning polynomial $G$, of degree  $O^*(E)$  so that each cell of 
$Z(F)\setminus Z(G)$ is crossed by at most $n/E$ arcs of $\G$. The number of cells in $\A(G;F)$ 
is at most $c_2 E^2$, for a suitable constant $c_2$ that depends on $D$. A plate $\plate$ not lying in $Z(F)$ 
crosses at most $c_3 E$ cells of $\A(G;F)$, for another constant $c_3$, and it is \emph{wide} (i.e., 
the cell does not contain any endpoint of a connected component of $\plate\cap Z(F)$) at all of them except for 
$O(D) = O(1)$ cells where it is \emph{narrow}.

For each one-dimensional cell $\zcell \in \A(G;F)$, we preprocess the input arcs that overlap 
with $\zcell$ (i.e., they lie in $Z(F)\cap Z(G)$) for plate-intersection queries, by preprocessing 
them into a segment tree. Omitting the straightforward details, the size of this data structure is 
$O(n\log n)$ and an intersection query can be answered in $O(\log n)$ time.

For each two-dimensional cell $\zcell$ of $Z(F)\setminus Z(G)$, let $\G_\zcell$ denote the set of 
arcs of $\G$ that intersect $\zcell$.
We build a secondary data structure for answering plate-intersection queries (within $\zcell$) with plates 
that are wide at $\zcell$. This requires constructing a CAD $\ZCAD$ of $\ES$ induced by $\{\hat F, \hat G\}$, 
as in Section~\ref{subsec:z-wide}, and building an arc-intersection-searching data structure for each 
cell $\zCADcell$ of $\ZCAD$. This enables us to work with the plane supporting a query plate at all 
cells where the plate is wide, and thus the parametric dimension of the query plates can be taken
to be $3$ for these cells. Hence, for a storage parameter $s \in [n,n^3]$, the query time is 
$O^*\biggl ((n/s^{1/3})^{\tfrac{1-1/\tO}{2/3}}\biggr)$. 
Finally, we recursively construct the data structure on $\G_\tau$.

To answer a query with a plate $\plate$, we first query the one-dimensional cells $\zcell$ of $\A(G;F)$ 
to answer an intersection query on the arcs that lie in $\zcell$. Next, let $\zcell$ be a two-dimensional 
cell of $\A(G;F)$ that $\plate$ crosses. If $\plate$ is wide at $\zcell$, we use the secondary structure 
stored at $\zcell$ to answer the intersection query within $\zcell$. If $\plate$ is narrow at $\zcell$, 
we recursively search at $\zcell$. There are at most $c_2 D$ cells of $\A(G;F)$ at which $\plate$ is 
narrow, for some absolute constant $c_2  > 0$. Since $E\gg D$, the total query time in answering intersection queries on the arcs lying in 
$Z(F)$ is $O^*\biggl((n/s^{1/3})^{\tfrac{1-1/\tO}{2/3}} \biggr)$, using $O^*(s)$ space and expected preprocessing.
We thus have:
\begin{lemma}
\label{lem:arc-zero}
Let $\TT$ be a family of constant-complexity plates in $\reals^3$, let $F$ be a partitioning polynomial
of constant degree, let $\G$ be a set of $n$ constant-degree algebraic arcs in $\reals^3$  that lie in $Z(F)$, 
	and let $\tO$ be the  reduced parametric dimension of the arcs in $\G$ or of the curves supporting 
	them if they are planar (with respect to $\TT$).
For a storage parameter $s\in[n,n^3]$, $\G$ can be preprocessed into a data structure of size~$O^*(s)$, 
in expected time $O^*(s)$, so that a plate-intersection query with a plate of $\TT$ that does not lie 
	in $Z(F)$ can be answered in $O^*\biggl( (n/s^{1/3})^{\tfrac{1-1/\tO}{2/3}} \biggr)$ time.
\end{lemma}

\paragraph{Putting everything together.}
For a storage parameter $s\in [n,n^\tQ]$, we construct the overall data structure following the approach 
in  Section~\ref{sec:trade-off}; see also the space/query-time trade-off discussion for the case of lines in Section~\ref{subsec:line-input}. 
We construct a partition tree $\Psi$ on $\G$ 
using the polynomial partitioning technique of Guth~\cite{Guth} and choose the same threshold $n_0$ as 
in~\eqref{eq:threshold}. Let $v$ be a node of $\Psi$ associated with a subset $\G_v \subseteq \G$ of $n_v$ arcs and 
a storage parameter $s_v$; $s_v=s$ for the root, and $s_w = s_v/D^3$ for a child $w$ of $v$. If $n_v > n_0$, we associate secondary structures at $v$ of size $O^*(s_v)$ using Lemmas~\ref{lem:arc-wide} and~\ref{lem:arc-zero}.

If $n_v \le n_0$, then $v$ is a leaf of $\Psi$ and we proceed as follows. If $s \le n^{3/2}$, we construct 
a data structure on $\G_v$ of size  $O^*(n_v)$ for answering plate-intersection queries in $O^*(n_v^{1-1/\tO})$ time using Theorem~\ref{lem:Pi-lin}. If $s > n^{3/2}$, we construct a data structure on $\G_v$ of size $O^*(n^\tQ)$ so that a query can be answered in $O^*(1)$ time. The same analysis as in Section~\ref{sec:trade-off} implies that the 
total size of the data structure is $O^*(s)$, and that it can be constructed in $O^*(s)$ expected time.

The query procedure visits $O^*((n/n_0)^{1/2})$ leaves of $\Psi$, and by Lemmas~\ref{lem:arc-wide} 
and~\ref{lem:arc-zero}, the total time spent in querying the secondary structures is 
$O^* ((n/s^{1/3})^{\tfrac{1-1/\tO}{2/3}})$.
For $s \le n^{3/2}$, the query procedure spends $O^*(n_0^{1-1/\tO})$ time at each leaf, so the total time spent at the leaves is 
$O^* ( n^{2- 3/\tO} / s^{1-2/\tO})$, which is dominated by the  time spent on the secondary structures.
For $s>n^{3/2}$, the query procedure visits $O((n/s^{1/\tQ})^{\tfrac{1}{2-3/\tQ}})$ leaves and spends $O^*(1)$ time at 
each leaf.
Putting everything together, we obtain the following summary result.
\begin{theorem}
  \label{thm:arcs}
Let $\TT$ be a family of constant-complexity plates in $\reals^3$, and let $\G$ be a set of 
constant-degree algebraic arcs in $\reals^3$ so that only $O(1)$ of them lie on any plane. 
Let $\tQ$ be the reduced parametric dimensions of \ $\TT$, and let $\tO$ be the reduced parametric dimension of the arcs in $\G$ or of the curves supporting them if they are planar. 
For a storage parameter $s \in [n,n^{\tQ}]$, $\G$ can be preprocessed, in $O^*(s)$ expected time, 
into a data structure of size $O^*(s)$, so that for any query plate $\plate\in \TT$, an 
intersection query on $\G$ with $\plate$ can be answered in time
\[
	O^* \biggl ( (n/s^{1/3})^{\frac{1-1/\tO}{2/3}}+(n/s^{1/\tQ})^{\frac{1}{2-3/\tQ}} \biggr ).
\]
\end{theorem}

\begin{remark*}
  We note that if we formulate the plate-intersection query as semi-algebraic range searching in a 
  straightforward manner, by Theorem~\ref{thm:Pi-tradeoff}, 
  we will obtain a data structure of size $O^*(s)$ with query time
  $O^*\left((n/s^{1/\tQ})^{\tfrac{1-1/\tO}{1-1/\tQ}}\right)$. For the boundary cases of the storage parameter, 
  our query time is the same, but it is better for all intermediate values of $s$. In particular, 
  for $s=n^{3/2}$, the simple approach yields the query time 
  $O^*(n^{\rho})$, with $\rho= (1-\tfrac{1}{\tO})(1-\tfrac{1}{2(\tQ -1)})$, while 
  $\rho$ improves to $\tfrac{3}{4}(1-\tfrac{1}{\tO })$ for our approach, which is indeed smaller for $\tQ > 3$.
\end{remark*}

\section{Plate-Intersection Queries amid Plates}
\label{sec:plate-plate}

Finally, we move to the third type of queries, in which both input and query objects are plates. 
We first focus on the case where both input and query objects are triangles, and later comment on
the extension to the general case.

\subsection{The case of triangles}
\label{subsec:tri-tri}

We wish to preprocess a set $\T$ of $n$ triangles in $\reals^3$ for answering triangle-intersection queries.
The solution is quite simple.
Note that if a query triangle
$\nabla$ intersects an input triangle $\Delta$ (in general position) then their intersection 
is a line segment $e=pq$, where each of the endpoints $p$, $q$ is the intersection point of either
(i) an edge of $\nabla$ with $\Delta$, or (ii) an edge of $\Delta$ with $\nabla$.
(Any combination of (i) and (ii) can occur at the two endpoints $p,q$.)

For intersections of type (i), we construct a data structure using Corollary~\ref{cor:space-query-triangle} 
in Section~\ref{sec:trade-off}, which 
answers a query in time $O^* (n^{5/4}/s^{1/2} + n^{8/9}/s^{2/9})$ using $O^*(s)$ storage.
For intersections of type (ii), we use the data structure in Section~\ref{subsec:seg-input},
which can answer a query in time $O^* ( n^{5/4}/s^{1/2} + n^{4/5}/s^{1/5})$ using $O^*(s)$ storage. Hence, we obtain the following result:
\begin{theorem} \label{thm:tri-count}
Let $\T$ be a set of $n$ triangles in $\reals^3$, and let $s\in [n,n^4]$ be a storage parameter.
$\T$ can be preprocessed, in expected time $O^*(s)$, into a data structure of size $O^*(s)$ so that a 
	triangle-intersection query can be answered in time $O^* (n^{5/4}/s^{1/2} + n^{8/9}/s^{2/9})$.
\end{theorem} 

We note that Theorem~\ref{thm:tri-count} holds for triangle-intersection counting queries as well. If we are 
only interested in triangle-intersection detection or reporting queries, we can use the data structure 
in~\cite{ES} for handling type~(i) intersections to get a slightly improved bound, namely, the structure
can answer a query 
in $O^*(n^{5/4}/s^{1/2}+n^{4/5}/s^{1/5})$ time using $O^*(s)$ storage. We thus obtain the following:

\begin{theorem} \label{thm:trixx}
Let $\T$ be a set of $n$ triangles in $\reals^3$, and let $s\in [n,n^4]$ be a storage parameter.
$\T$ can be preprocessed, in expected time $O^*(s)$, into a data structure of size $O^*(s)$ so that a 
triangle-intersection detection (resp., reporting) query can be answered in time
$O^* ( n^{5/4}/s^{1/2} + n^{4/5}/s^{1/5})$ (resp., $O^* ( n^{5/4}/s^{1/2} + n^{4/5}/s^{1/5} + k)$,
where $k$ is the output size).
\end{theorem} 
\subsection{The case of plates}
\label{subsec:plate-plate}

Consider next the general setup, where both the input and query objects are plates. 
As in Section~\ref{sec:main}, we assume that the input plates are in general position, i.e., 
any plane contains $O(1)$ input plates, and any line is contained in the supporting planes of $O(1)$ input plates.
Furthermore, we assume that the edges of the input and query plates admit a parametric representation as 
described in the beginning of Section~\ref{sec:main}.
Many aspects of the algorithm of Theorem~\ref{thm:trixx} are fairly easy to generalize. 
Let $\nabla$ be a query plate. By our general position assumption, $h_\nabla$, the plane supporting $\nabla$ contains 
$O(1)$ input plates, which can be handled separately. So let $\Delta$ be an input plate that does not lie in $h_\nabla$.
Then $\Delta\cap\nabla$ is the union of $O(1)$ pairwise-disjoint segments, all lying on the intersection
line of the two supporting planes, and an endpoint of each of these segments
is an intersection point of either (i) $\Delta$ with a boundary arc of $\nabla$, 
or (ii) a boundary arc of $\Delta$ with $\nabla$. 
(There is only one intersection segment when both plates are convex.) 
For simplicity, we state the bounds only for the case $s=n^{3/2}$. 

\paragraph*{Intersections of type (i).}
We preprocess $\T$ into a data structure using Theorem~\ref{thm:space-query-tradeoff} in 
Section~\ref{sec:trade-off} with storage parameter $s=n^{3/2}$.
An intersection query with an arc bounding $\nabla$ takes
$O^*\biggl(n^{\frac{2\tQ -3}{3(\tQ -1)}}\biggr)$ time, where $\tQ$ is the reduced parametric 
dimension of the curves supporting the edges of query plates.
\paragraph*{Intersections of type (ii).}

For this case we use the technique presented in Section~\ref{subsec:seg-input}.
That is, we apply Theorem~\ref{thm:arcs} to the boundary arcs (edges) of the input plates,
and obtain a data structure of size $O^*(n^{3/2})$ (constructed in expected time $O^*(n^{3/2})$),
which supports plate-intersection queries in $O^*\left(n^{\frac{3(\tO -1)}{4\tO}}\right)$ time,
where $\tO \ge 3$ is the reduced parametric dimension of the curves supporting the edges of input plates.
Combining the bound in Theorem~\ref{thm:arcs} with the one in Theorem~\ref{thm:space-query-tradeoff},
we obtain:
\begin{theorem} \label{thm:ext}
  Let $\TT$ be a family of constant-complexity plates in $\reals^3$,
  let $\T$ be a set of $n$ constant-complexity plates (not necessarily from the family $\TT$) 
  in $\reals^3$, and let $\tQ$ (resp.\ $\tO$)
  be the reduced parametric dimension of the curves supporting the edges of 
  plates in $\TT$ (resp.\ $\T$) with respect to the intersection predicate 
  with $\T$ (resp.\ $\TT$).  $\T$ can be preprocessed, in expected time $O^*(n^{3/2})$, into a 
  data structure of size $O^*(n^{3/2})$ so that an intersection query amid $\T$ with a plate 
  in $\TT$ can be answered in time $O^*(n^\rho)$, where
  $\rho = \max \left\{ \frac{2\tQ -3}{3(\tQ -1)},\; \frac{3(\tO -1)}{4\tO} \right\}$,
  for $\tO=\tQ =t \ge 3$, $\rho=\frac{3(t-1)}{4t}$.
  For counting queries, it counts the number of 
  intersection segments between the query plate and the input plates.
\end{theorem}
\section{Segment-Intersection Queries for Spherical Caps}
\label{sec:caps}

In the final result of this study, we consider segment-intersection queries for a set $\S$ of 
$n$ spherical caps in $\reals^3$; a spherical cap is a portion of a sphere cut off by a halfspace.
This case is different from the previous cases in that the input objects are not flat. We use this
case to illustrate that our techniques can also be applied to non-flat objects. As earlier, we 
construct a partition tree $\Psi$ on $\S$ based on the polynomial partitioning technique of Guth~\cite{Guth}, 
and the main challenge is to answer 
segment intersection queries on spherical caps that are wide at a cell of the respective polynomial partition.
We go over the steps of the algorithms presented in Sections~\ref{sec:main}--\ref{sec:trade-off}
and discuss the modifications needed for the new problem.

\paragraph{Constructing the CAD}
In the spirit of the technique in Sections~\ref{sec:main}--\ref{sec:trade-off},
we replace the caps by their containing spheres by constructing a CAD, as follows.
Let $\ES_4 \coloneqq \reals^4$ denote the $(a,b,c,r)$- space of all spheres in $\reals^3$
(where a point $(a,b,c,r)$ represents the sphere centered at $(a,b,c)$ and of radius $r$; strictly speaking, $\ES_4 = \reals^3 \times \reals_{\ge 0}$ 
but for simplicity and without loss of generality, we allow spheres of negative radius). 
For a point $\xi = (a,b,c,r)\in \ES_4$,
let $S_\xi\colon (x-a)^2+(y-b)^2+(z-c)^2=r^2$ denote the sphere defined by $\xi$.
For a spherical cap $\Delta$, let $\Delta^* \in \ES_4$ be the point corresponding to the sphere that contains $\Delta$.
Set $\ES \coloneqq \ES_4 \times \reals^3$.
Define the polynomial $\hat G \in \reals[a,b,c,r,x,y,z]$ by 
\[ 
\hat G (a,b,c,r,x,y,z) \coloneqq (x-a)^2 + (y-b)^2 + (z-c)^2 - r^2 . 
\]
Let $F\in\reals[x,y,z]$ be a partitioning polynomial of degree $c_1D$ for some constant
$c_1>0$ and some sufficiently large constant $D>0$. We construct a CAD $\CAD_7$ of $\ES$ induced by 
$\{F,\hat G\}$, with coordinates $a,b,c,r,x,y,z$ in the reverse-elimination order (starting with $z$).

The \emph{base} CAD $\CAD_4$ is the  projection of $\CAD_7$ onto $\ES_4$. For a cell $C\in\CAD_7$, 
let $C^\downarrow$ be its projection onto $\ES_4$. For each point $\xi\in\ES_4$, the three-dimensional 
fiber of $\CAD_7$ over $\xi\in\ES_4$, denoted by $\fiber(\xi)$, is a refinement of $\A(\{S_\xi,F\})$.
As before, each cell of $\fiber(\xi)$ is the cross-section of a cell $C$ of $\CAD_7$ over $\xi$, 
denoted by $C(\xi)$, 
and thus has a constant-size semi-algebraic 
encoding, which depends only on $C$. Again, this encoding will be used in the subsequent range searching 
step. Each cell $C\in\CAD_7$ is associated with a cell $\tau = \tau_C$ of $\A(F)$, such that for all 
$\xi\in C^\downarrow$, $C(\xi)$ is contained in $\tau$. 
We will be mostly interested in two-dimensional cells 
of $\fiber(\xi)$ that are contained in $S_\xi$,  referred to as \emph{pseudo-trapezoids}.
These pseudo-trapezoids are cross-sections of the cells of $\CAD_7$ that lie in $Z(\hat G)$ and 
$\xi\in C^\downarrow$. 

We define wide caps and narrow caps in full analogy to the definitions for plates. If a cap $\Delta$ 
is wide at a cell $\tau$ of $\reals^3 \setminus Z(F)$ then $\fiber(\Delta^*)$ contains a partition 
of $\Delta\cap\tau$ into pseudo-trapezoids, all disjoint from the relative boundary of $\Delta$.

For each cell $C\in\CAD_7$ that lies in $Z(\hat G)$, let $\S_C$ be the set of spherical caps $\Delta$ 
such that $\Delta^*\in C^\downarrow$, $\Delta$ is wide at $\tau_C$, and the pseudo-trapezoid 
$C(\xi) \subseteq \Delta$. 

\paragraph{The range searching mechanism.}
For each cell $C$ of the CAD, we define a semi-algebraic \emph{$C$-intersection predicate} $\Pi(e, \xi)$,
similar to~\eqref{eq:intersect-pred}, which is $1$ if and only if the segment $e$ intersects the sphere $S_\xi$ and 
one of the intersection points lies in the pseudo-trapezoid $C(\xi)$. We preprocess the set $\S_C$ for 
$C$-intersection queries. The parametric dimensions of the spheres corresponding to the 
caps in $\S_C$ and of segments in $\reals^3$ are $4$ and $6$, respectively. Hence, for a storage 
parameter $s\in [n,n^6]$, $\S_C$ can be preprocessed into a data structure of $O^*(s)$ size, in
$O^*(s)$ expected time, so that a $C$-intersection query on $\S_C$ can be answered in 
$O^*((n/s^{1/6})^{9/10})$ time (see Theorem~\ref{thm:Pi-tradeoff} in the appendix).
We build such a data structure for each cell $C\in\CAD_7$ that lies in $Z(\hat G)$.

For a query segment $e$ and a cell $\tau \in \reals^3\setminus Z(F)$, we answer a segment-intersection query on 
the spherical caps  of $\S$ that are wide at $\tau$ by performing $C$-intersection queries for all 
cells $C$ such that $\tau_C=\tau$ and $e\cap \tau_C \ne \emptyset$.

\paragraph{The overall performance.}
Following the same analysis as in Section~\ref{sec:main} and observing that 
the query time of the secondary structure is $O^*(n^{27/40})$ 
for $s=n^{3/2}$, we obtain a data structure of size $O^*(n^{3/2})$ with
$O^*(n^{27/40})$ overall query time. 
To obtain a space/query-time trade-off, we combine the above partition tree with the 
machinery developed in Section~\ref{sec:trade-off}, namely we choose the same threshold $n_0$ as in~\eqref{eq:threshold}. For a leaf $v$ of $\Psi$ with $n_v \le n_0$ spherical caps, we proceed as follows:
For $s \le n^{3/2}$, we construct a data structure on $\S_v$ of size $O^*(n_v)$ using Theorem~\ref{lem:Pi-lin}.
Since the parametric dimension of a spherical cap is $\tO = 7$ 
(four for its sphere and three for the halfspace that cuts the cap off its sphere), 
a segment-intersection query on $\S_v$ can be answered in $O^*(n^{6/7})$ 
time. For $s > n^{3/2}$, we construct a data structure of size $O^*(n^6)$ (since the parametric dimension of segments is $6$) so that a segment-intersection query can be answered in $O^*(1)$ time. Following the analysis in
Section~\ref{sec:trade-off},
we obtain a data structure of size and expected preprocessing cost
$O^*(s)$, with query time $O^*(n^{11/7}/s^{5/7}+n^{9/10}/s^{3/20})$. 
\begin{theorem}
\label{thm:caps}
Let $\S$ be a set of $n$ spherical caps in $\reals^3$. For a storage parameter $s\in [n,n^6]$, $\S$ can be 
preprocessed, in $O^*(s)$ expected time, into a data structure of size $O^*(s)$, so that a segment 
	intersection query on $\S$ can be answered in $O^*(n^{11/7}/s^{5/7}+n^{9/10}/s^{3/20})$ time.
\end{theorem}

\begin{remarks*}
(i) Our main goal in considering segment-intersection queries amid spherical caps 
was to demonstrate the versatility of our technique. We did not make an effort to optimize the bounds. 
For example, it might be possible to improve the reduced parametric dimension of segments to $4$,
by removing the effect of its endpoints (as in Section~\ref{subsec:planar_imp}), or to improve the 
parametric dimension of spherical caps (ideally to $4$).

(ii) We considered segment-intersection queries amid spherical caps, but the technique 
extends 
to arc-intersection queries amid spherical caps or more broadly amid other types of 
surface patches. The bounds one would obtain are similar to those in Section~\ref{sec:trade-off},
and depend on the (reduced) parametric dimensions of the surface patches and query arcs.
\end{remarks*}

\section{Discussion}
\label{sec:disc}

In this paper we presented a general technique for answering intersection queries amid planar objects, 
which also extends to non-flat objects in some cases. Our main observation is that the CAD construction 
facilitates the passage from an input consisting of surface patches (such as triangles, disks or spherical 
caps) to the full surfaces containing them (such as planes or spheres). This leads to an improvement (often 
significant) in the reduced parametric dimension of the input objects, which in turn leads to improved 
performance bounds for the resulting algorithms. 

This paper raises many open issues that would be interesting to pursue. We mention a few here:
\begin{itemize}
\item 
Can our approach be extended to answer intersection queries where both input and query objects 
are non-flat surface patches such as spherical caps? Our idea in 
Section~\ref{sec:plate-plate} of working with the boundary arcs of input/query objects does 
not work in this setting, 
as the intersection curve of two spherical caps may be a closed 
curve that lies in the interior of both caps. However, it might be possible to reduce the problem to the sphere-intersection problem among a set of spheres by using a variant of the CAD construction.
\item 
  Can our idea of replacing the query planar arc with the curve supporting the arc (see Section~\ref{sec:trade-off}) be extended to non-planar arcs?
  
\item 
  Currently the arc-intersection counting query counts the number of intersection points between the query arc and the input plates. Can it be extended to count the number of plates intersected by the query arc?

  \item 
    Can the recent lower bounds on semi-algebraic range queries~\cite{AC21,AC22} and on intersection searching~\cite{AC23} be extended to prove that (some of) the data structures presented in this paper are near optimal?
\end{itemize} 

\begin{acks}
  We thank Peyman Afshani for sharing with us problems that have motivated our study 
  of segment-intersection searching amid spherical caps. We also thank Ovidiu Daescu 
  for suggesting the problems studied in the latter part of the paper.
  Moreover, we are grateful to Saugata Basu, Eric Kaltofen, and J. Rafael Sendra for useful discussions.
  Finally, thanks are due to the anonymous
  reviewers of the conference version and of a preliminary version of this work for their insightful comments, 
  which helped us improve the presentation.

  Work by P.~K.~A. partially supported by NSF grants CCF-20-07556,  CCF-22-23870, and IIS-24-02823,
  and by the US-Israel Binational Science Foundation Grant 2022131.
  Work by B.~A. partially supported by NSF Grants CCF-15-40656 and CCF-20-08551,
  and by Grant~2014/170 from the US-Israel Binational Science Foundation.
  Work by E.~E. partially supported by the US-Israel Binational Science Foundation Grant 2022131, %
  and by Grants 824/17 and 800/22 from the Israel Science Foundation.
  Work by M.~K. partially supported by Grant 495/23 from the Israel Science Foundation,
  and by Grants 2019715 and CCF-20-08551 from the US-Israel Binational Science Foundation/US National Science Foundation.
  Work by M.~S. partially supported by Grants 260/18 and 495/23 from the Israel Science Foundation.
\end{acks}
\bibliographystyle{abbrv}%
\bibliography{circshoot}

\newpage

\appendix
\section*{Appendix}

\section{Answering Semi-Algebraic Relation Queries}
\label{app:multi-level}

In this appendix we present a multi-level data structure for answering \emph{semi-algebraic-relation queries}, defined by a semi-algebraic predicate, by recursively composing partition trees based on polynomial-partitioning methods.
The concept of a multi-level data structure goes back (at least) to Bentley~\cite{Be80} who used 
multi-level range trees for orthogonal range searching. 
Dobkin and Edelsbrunner~\cite{DE87} used multi-level data structures in 
the context of partition trees. Over the last four decades, multi-level data structures, based on geometric cuttings and 
simplicial partitions, have been extensively used to answer queries formulated as a conjunction of linear inequalities, 
see, e.g., the survey papers~\cite{Mat94,AE99,Ag:rs}. Using the semi-algebraic range searching data structure 
by Agarwal and Matou\v{s}ek~\cite{AM94},
multi-level partition trees have been developed for answering queries 
that are formulated as conjunctions of polynomial inequalities~\cite{AE99,Mat93,Mat94}, but they lead to 
a weaker bound (e.g., $O^*(n^{1-\tfrac{1}{2d-4}})$ query time, instead of $O^*(n^{1-\tfrac{1}{d}})$,
using $O^*(n)$ space, for $d> 4$). 
In principle, partition trees based on the recently developed algorithmic polynomial-partitioning technique, 
as developed in~\cite{AAEZ,MP}, can also be composed to 
construct multi-level data structures, and thereby yield considerably more efficient data structures
than those available from \cite{AM94}. However, both the construction and the analysis are more subtle,
due to the complicated nature of the partitions constructed in \cite{AAEZ,MP}.
Since we are unaware of such constructions having been described in the literature and these data structures are
extensively used throughout this paper,
we present them and analyze their performance in this appendix, in a fairly comprehensive manner,
for the sake of completeness.
We also believe that the general machinery presented here will find additional applications
beyond those in the present work.

Let $\oset$ be a family of geometric objects (e.g., points, segments, balls, simplices) such that 
each object $\obj \in \oset$ can be specified by a vector of $t$ real values, 
for some constant $t$, and thus can be represented 
as a point $\obj^*$ in a $t$-dimensional real vector space $\reals^t$, which we refer to as the 
\emph{object} space (sometimes also called the \emph{data} space) and denote by $\ospace$. 
Similarly, let $\qset$ be a family of \emph{query} objects, where each query object $\query$ 
can be represented as a point $\query^*$ in $\reals^{t'}$, where $t'$ is the (another constant) number of parameters 
needed to specify a query. Let $\qspace$ denote the parametric space $\reals^{t'}$
of query objects. We refer to the dimensions 
$t$ and $t'$ of $\ospace$ and $\qspace$ as the respective \emph{parametric dimensions} of $\oset$ and $\qset$.
Let $\Pi\colon \ospace\times \qspace \rightarrow \{0,1\}$ be a \emph{semi-algebraic predicate} defined as a Boolean formula on
a set of polynomial inequalities. 
Without loss of generality, we can assume that $\Pi$ is of the form
\begin{equation}
\label{eq:semi-pred}
	\Pi (\xx, \yy) = \bigvee_{i=1}^r \bigwedge_{j=1}^{k_i} (g_{ij} (\xx, \yy) \ge 0),\quad\text{for } \xx\in \ospace,\;\yy\in\qspace ,
\end{equation}
where each $g_{ij} \in \reals[\xx,\yy]$ is a polynomial over the joint space $\ospace\times\qspace$. 
With a slight abuse of notation, for a pair $\obj\in\oset$ and $\query\in\qset$, we will use 
$\Pi(\obj,\query)$ to denote $\Pi(\obj^*,\query^*)$. We say that $\Pi$ has \emph{constant complexity} if 
$\sum_i k_i$ is a constant and the degrees of all the polynomials $g_{ij}$ are also bounded by some constant. 
For a query object $Q\in\qset$, let
\[ 
\Qout_\Pi(\query) \coloneqq \{\obj\in\oset \mid \Pi(\obj,\query) = 1 \} .
\]

Clearly, $\Qout_\Pi(\query)$ is a semi-algebraic set of constant complexity for each $\query\in\qspace$.
Our goal is to preprocess $\oset$ into a data structure 
so that for a query object $\query\in\qset$, a desired aggregate statistics on the set 
$\Qout_\Pi(\query)$ can be computed quickly. We refer to this task as a 
\emph{$\Pi$-query}. 
We use the standard semi-group model: let $(\Sigma,+)$ be a semigroup. 
Each object $\obj\in\oset$ has a weight $w(\obj)\in \Sigma$. For a query $\query\in\qset$, 
the goal is to compute the sum 
\[
	\ph (\query) \coloneqq \ph_\Pi (\query, \Sigma) = \sum_{\obj\in\Qout_\Pi(\query)} w(\obj) .
\]
For example, counting queries can be answered by
choosing the semigroup to be $(\mathbb{N},+)$, where $+$ denotes
the standard integer addition, and setting $w(p)=1$ for every $p\in S$;
detection queries by choosing the semigroup to be $(\{0,1\}, \vee)$
and setting $w(p)=1$ for every $p$; and reporting queries by choosing the semigroup to
be $(2^S,\cup)$ and setting $w(p) = \{p\}$.

For a query object $\query$, define the semi-algebraic set
\begin{equation}
	\label{eq:ospace}
	\pregion{\query}_\Pi \coloneqq \{ \xx\in\ospace \mid \Pi(\xx,\query^*) = 1 \} .
\end{equation}
Then a $\Pi$-query on $\oset$ with $\query$ can be formulated as a semi-algebraic range query 
with $\pregion{\query}_\Pi$ on the set $\oset^* = \{\obj^*\mid \obj\in\oset\}$ (in the object space).

Alternatively, we can map an object $\obj\in\oset$ to the semi-algebraic region $\pregion{\obj}_\Pi$ 
in the query space, given by
\begin{equation}
	\label{eq:qspace}
	\pregion{\obj}_\Pi \coloneqq  \{ \yy\in\qspace \mid \Pi(\obj^*,\yy) = 1 \} .
\end{equation}
A $\Pi$-query can now be formulated as a \emph{point-enclosure} query (in the query space) 
with $\query^*$ on the collection of regions
$\pregion{\oset}_\Pi \coloneqq \{ \pregion{\obj}_\Pi \mid \obj\in\oset\}$.
That is, we query with a point, and seek the aggregate weight of the regions that contain it.
Roughly speaking, the first approach leads to  an $O^*(n)$-size data structure for answering $\Pi$-queries while the second approach leads to a data structure with $O^*(1)$ query time.

The following two lemmas, taken respectively from \cite{MP} and \cite{AAEZ},
lead to partition trees for answering semi-algebraic range and point-enclosure queries:

\begin{lemma}[Matou\v{s}ek and Pat\'akov\'a \cite{MP}]
	\label{lem:mp}
	Let $V$ be an algebraic variety of dimension $k\ge 1$ in $\reals^d$ such that all of its irreducible components have 
	dimension $k$ as well, and the degree of every polynomial defining $V$ is at most $E$. 
	Let $S \subset V$ be a set of $n$ points, and let $D>1$ be a parameter. 
	There exists a polynomial $g\in\reals[x_1,\ldots,x_d]$ of degree at most 
	$E^{d^{O(1)}}D^{1/k}$ that does not vanish identically on 
	any of the irreducible components  of $V$ (i.e., $V\cap Z(g)$ has dimension at most $k-1$), and each cell of 
	$V\setminus Z(g)$ contains at most 
	$n/D$ points of $S$.  Assuming $D, E, d$ are constants, the polynomial~$g$, a semi-algebraic representation of the 
	cells in  $V\setminus Z(g)$, and the points of $S$ lying in each cell can be computed in $O(n)$ expected time.
\end{lemma}

The second lemma, proved in~\cite{AAEZ}, generalizes the above result to semi-algebraic sets,
albeit with a somewhat weaker claim.
\begin{lemma}[Agarwal~\etal~\cite{AAEZ}]
	\label{lem:aaez}
	Let $V$ be an algebraic variety of dimension $k\ge 1$ in $\reals^d$, defined by polynomials of degree at most $E$. 
	Let $S$ be a multiset of $n$ semi-algebraic sets in $\reals^d$, each of complexity at most $b$, and let
	$D>1$ be a parameter. 
There exists a polynomial $g\in\reals[x_1,\ldots,x_d]$ of degree $E^{d^{O(1)}} D$ so that 
	$V\cap Z(g)$ has dimension at most $k-1$; $V\backslash Z(g)$ is partitioned into a set $\Omega$ of 
	$O(E^{d^{O(1)}} D^k)$ 
 connected semi-algebraic cells, each of complexity $(ED)^{O(d^4)}$, so that each 
	cell of $\Omega$ is crossed by (that is, intersected by, but not contained in) at most $n/D$ sets of $S$. 
        Assuming $D$, $E$, $d$, and $b$ are constants, the polynomial $g$, a semi-algebraic representation 
        of the cells in $\Omega$, and the elements of $S$ crossing each cell of $\Omega$ can be computed in 
        $O(n)$ randomized expected time.
\end{lemma}

The first (resp., second) lemma is our main tool for constructing a partition tree for answering 
semi-algebraic range queries (resp., point-enclosure queries). 
We combine them to obtain a trade-off between the size of the data structure and its query time. 
The overall performance of the data structure can be improved by using the following observation.
Often, the parametric dimensions $\dim(\ospace)$ of $\oset$ and $\dim(\qspace)$ of $\qset$ might be large,
but each polynomial in the predicate~$\qpred$
uses only few of the parameters that define $\ospace$ and $\qspace$. 
For example, a triangle in $\reals^2$ requires six parameters, 
but many queries on triangles (e.g., triangle-intersection queries) can be expressed by 
disjunctions 
 and conjunctions of polynomial inequalities where each inequality uses only two parameters 
of each of the involved triangles (the coordinates of one vertex of the triangle or the two coefficients defining 
the line supporting one of its edges). In this case one can construct a multi-level data structure,
each of whose levels consists of a two-dimensional partition tree, rather than a six-dimensional
tree; see, e.g.,~\cite{DE87,AE99,Mat94}. 

We say that the \emph{reduced parametric dimension} of $\ospace$ (with respect to $\Pi$) is $\tO$ if 
each polynomial $g_i$ in $\Pi$ uses at most $\tO$ of the $\dim (\ospace)$ parameters of an object. 
Similarly, we define the reduced parametric dimension of $\qspace$ with respect to $\Pi$ and denote 
it by $\tQ$. Let $\xx_i$ (resp. $\yy_i$) denote the subset of the variables of $\xx$ (resp., $\yy$) 
used in $g_i$, and let $\ospace_i$ (resp., $\qspace_i$) be the 
subspace of $\ospace$ (resp., $\qspace$) spanned by $\xx_i$ (resp., $\yy_i$). For a data object $\obj$, 
let $\obj_i^*$ denote the projection of $\obj$ 
onto $\ospace_i$, and similarly define $\query_i^*$ for 
a query object $\query$. Each $g_i$ is defined over the corresponding subspace $\ospace_i\times \qspace_i$
of $\ospace \times \qspace$, of dimension at most $\tO +\tQ $. 
Let 
\begin{align*}
\pregion{\query}_i & = \{\xx_i \in \ospace_i \mid g_i(\xx_i,\query_i^*) \ge 0\}\quad\text{and} \quad
\pregion{\obj}_i  = \{ \yy_i \in \qspace_i \mid g_i (\obj_i^*,\yy_i) \ge 0\} .
\end{align*}

For simplicity, we describe the multi-level data structure for the case when $\Pi$ consists of only conjunctions, i.e., $\Pi$ is of the form
\begin{equation}
\label{eq:semi-pred-conj}
	\Pi (\xx, \yy) = \bigwedge_{i=1}^{k} (g_{i} (\xx, \yy) \ge 0),\quad\text{for } \xx\in \ospace,\;\yy\in\qspace .
\end{equation}
The disjunctions in a general predicate
$\Pi$ can be handled by constructing a separate data structure for each disjunct. That is,
if $\Pi(\xx,\yy) = \Pi^{(1)} (\xx, \yy) \vee \Pi^{(2)} (\xx,\yy)$, say, then we construct separate data 
structures for $\Pi^{(1)}$ and $\Pi^{(2)}$, query each of them with the query object, and aggregate their 
answers. This na\'{\i}ve approach works for detection, reporting, and even some semi-group aggregation
(e.g., max or min) queries but it does not work for counting queries,
as some objects may be counted more than once. 
To handle general aggregation queries,
we build a data structure for 
$\Pi^{(1)}$, augment it with one for $\Pi^{(2)}$, and modify the query procedure, as described below, to
ensure that each input object satisfying the query predicate is 
included in the sum exactly once. 

\paragraph{Overview of the data structure.}
For $i\in [1,k]$, let $\pi_i$ and $\Pi_i$ be the semi-algebraic predicates 
\begin{equation}
	\label{eq:prefix-pred}
	\pi_i(\xx,\yy) = (g_i(\xx_i,\yy_i) \ge 0) \quad \mbox{and}\quad 
	\Pi_i (\xx,\yy) = \bigwedge_{j=1}^i \pi_j(\xx,\yy) .
\end{equation}
A standard multi-level data structure for answering $\Pi$-queries recursively builds $k$ levels of 
a partition tree. For convenience, we index the levels in reverse order, so we refer to the topmost 
level as the level-$k$ tree, and to the bottommost level as the level-$1$ tree.  
The level-$i$ partition tree is constructed to \emph{extract}, from a current so-called \emph{canonical subset},
the objects that satisfy the polynomial inequality $g_i$ with respect to a query object. Thus the top $k-i+1$ 
levels of the data structure together extract the objects that satisfy $\bigwedge_{j\ge i}\pi_j(\xx,\yy)$.
Each node $v$ of any level-$i$ partition tree (except for $i=0$) recursively builds another
partition tree of $i-1$ levels for answering $\Pi_{i-1}$-queries,
and it is attached to $v$ as one of its ``secondary'' data structures. 

Since we construct our partition trees using polynomial partitioning, there is an additional complication,
because Lemmas~\ref{lem:mp} and~\ref{lem:aaez} do not provide any guarantees on the partitioning of the
points that lie on the zero set of the 
corresponding partitioning polynomial, so  we have to handle the zero set separately. Nevertheless, the lemmas
provide us with the means of doing this, as they are formulated to apply to point sets and regions that
lie on a variety, of any dimension. This leads to two nested recursions---the outer one recurses on the 
index of the query subpredicate, as above. For each outer recursive level, the inner recursion is on the 
dimension of the variety to which the input or query objects are mapped as points. 

We now describe the overall data structure in detail. We first present a data structure of size $O^*(n)$ 
and then describe how to improve the query time by increasing the size of the structure. Since the predicate 
$\Pi$ is fixed, we omit it from the subscripts in our notation. For a set $\F \subset \reals[x_1,\ldots,x_d]$ 
of polynomials, let $Z(\F)=\bigcap_{F\in \F} Z(F)$ denote the real variety defined by $\F$. 
For $\F=\emptyset$, $Z(\F) = \reals^d$.

\subsection{Near-linear-size data structure}
\label{subsec:Pi-lin}

\paragraph{Data structure.}

We construct a multi-level partition tree $\Psi$ on $\oset$ for answering semi-algebraic range queries 
with the sets $\pregion{\query}$, for $\query\in\qset$. At each recursive step, we are at a node $v$ 
of $\Psi$, associated with some canonical subset $\oset_v\subseteq \oset$, and with two indices
$0 \le i \le k$, the index of the corresponding sub-predicate $\Pi_i$, and 
$1 \le t \le t_i = \dim(\ospace_i)$, the dimension of the variety (of constant degree)
containing (the points representing the objects of) $\oset_v$. 
The task at $v$ is to construct a recursive partition tree $\Psi_v^{(i,t)}$ 
for answering $\Pi_i$-queries on $\oset_v$.
We refer to $\Psi_v^{(i,t)}$ as an \emph{$(i,t)$-level partition tree}.

More precisely, we have a triple $(\oset_v,\F,i)$, where $\F$ is a set of $O(1)$ polynomials of
constant degree in $\reals[\xx_i]$ and $\oset_v \subseteq \oset$ is the canonical subset associated
with $v$, so that the set $\oset_{v,i}^* \coloneqq \{ \obj_i^* \mid \obj\in\oset_v\}$ of the points 
representing the objects of $\oset_v$ is contained in $Z(\F)$. Our goal at $v$ 
is to answer $\Pi_i$-queries on $\oset_v$ (that is, on $\oset_{v,i}^*$).
Set $n_v \coloneqq |\oset_v|$. The node $v$ is associated with a cell 
$\tau_v\subseteq Z(\F)$ of a suitable polynomial partition $F$ 
(that is, a connected component of $Z(\F)\setminus Z(F)$). 
Initially, $\oset_v = \oset$, $\F=\emptyset$, 
$i=k$, and $\tau_v=\ospace_i$. 

For $i=0$, $\Psi_v^{(i,t)}$ is a singleton node that stores $w(\oset_v)$. 
So assume $i\ge 1$. We fix a threshold parameter $n_0 \coloneqq n_0(i,t)$, where $n_0=1$ for $t=1$. 
Consider first the case $t=1$. In this case, $\oset_{v,i}^*$ lies on a one-dimensional curve 
$Z(\F)$. 
For simplicity, assume that $Z(\F)$ is a connected irreducible curve (the general
case is handled by partitioning $Z(\F)$ into 
connected, irreducible 
components and handling each of them separately).
We sort $\oset_{v,i}^*$ along $Z(\F)$ and construct a one-dimensional range tree~\cite{dBCKO08} on that set.  
Each node $z$ of the range tree is associated with a subarc $\tau_z$ of $Z(\F)$ and a subset 
$\oset_z \subseteq \oset_v$, such that $\oset_{z,i}^* \subset \tau_z$. If $z$ is a leaf, we simply 
store $\oset_z$ at $z$. Otherwise, we construct an $(i-1, t_{i-1})$-level structure 
$\Psi_z^{(i-1,t_{i-1})}$ for the subproblem $(\oset_z,\emptyset,i-1)$ and attach it 
to $z$ as a secondary structure.

Next, assume $i\ge 1$ and $t>1$. If $n_v \le n_0$ then $v$ is a leaf and we simply store $\oset_v$ 
at $v$. So assume that $n_v > n_0$. 
We choose a sufficiently large constant 
$D \coloneqq D(i,t)$ (in particular, $D(i,t-1)\gg D(i,t)$),
and apply Lemma~\ref{lem:mp}, which
yields a partitioning polynomial $F_v$ for the point set $\oset_{v,i}^*$ with respect to $Z(\F)$
that satisfies the properties in the lemma, i.e., 
the degree of $F_v$ is $O(D^{\tfrac1t+\delta})$, for an arbitrarily small constant $\delta>0$, and each
cell of $Z(\F) \setminus Z(F_v)$ contains at most $n_v/D$ points of $\oset^*_{v,i}$.
We attach two secondary structures to $v$---one $(i,t-1)$-level data structure $\Psi^{(i,t-1)}_v$, and another 
$(i-1,t_{i-1})$-level data structure $\Psi^{(i-1,t_{i-1})}_v$---as described below.

Let $\oset_v^0 \coloneqq \{ \obj \in \oset_v \mid \obj_i^* \in Z(F_v)\}$. 
We recursively construct an $(i,t-1)$-level data structure $\Psi^{(i,t-1)}_v$ for the 
subproblem $(\oset_v^0, \F\cup\{F_v\}, i)$ and attach it to $v$ as one of its secondary structures. 
We also construct an $(i-1,t_{i-1})$-level data structure $\Psi^{(i-1,t_{i-1})}_v$ for the 
subproblem $(\oset_v, \emptyset, i-1)$ and attach it to $v$ as another secondary structure. 
Let $\tau$ be a cell of $Z(\F) \setminus Z(F_v)$. We create a child $z_\tau$ of $v$. 
We then compute a semi-algebraic representation of $\tau$ and store it at $z_\tau$.
Set $\oset_\tau \coloneqq \{ \obj\in\oset_v \mid \obj_i^* \in \tau\}$. 
We recursively construct an $(i,t)$-level subtree $\Psi_\tau^{k,i}$ for the subproblem 
$(\oset_\tau, \F, i)$ and store it as a subtree of $\Psi_v^{(i,t)}$ rooted at $z_\tau$.
This completes the description of the overall data structure.

\begin{figure}[htb]
   \centering
   \includegraphics[scale=0.8]{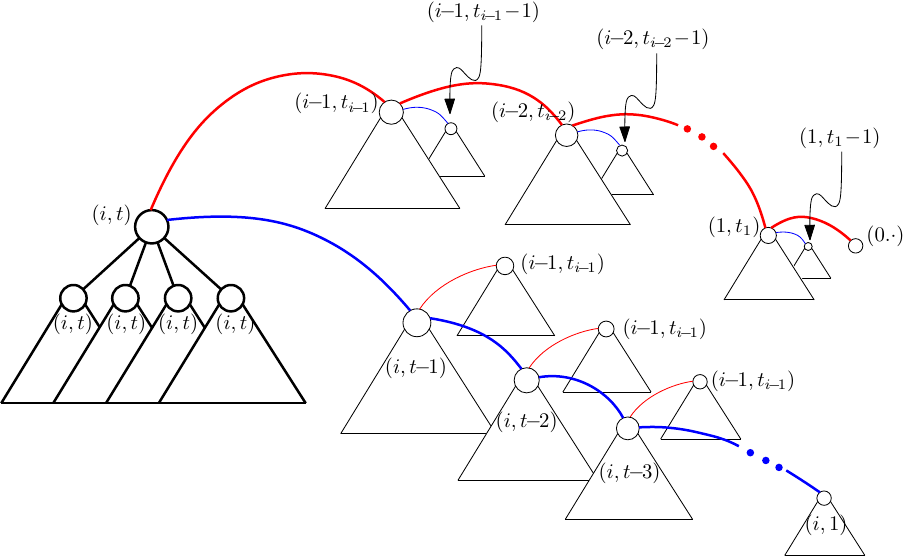}
   \caption{Schematic diagram of the $O^*(n)$-size data structure $\Psi_v^{(i,t)}$, constructed 
     on a $(i,t)$-level node $v$. Red and blue threads illustrate recursion on $i$ and $t$,
     respectively.}
   \label{fig:schem1}
\end{figure}

\paragraph{Query procedure.}
Let $\query\in\qset$ be a query object. Roughly speaking, for each $i\in[0,k]$, using $\Psi$, 
we compute the subset 
\[
	\Qout_\Pi^{(i)} (\query) \coloneqq 
	\{ \obj \in \oset \mid \bigwedge_{j=i+1}^k \left( g_j(\obj^*, \query^*)\ge 0 \right)\},
\]
the output of a $\Pi_i$-query on $\query$,
as 
a partition into 
a small number of pairwise-disjoint \emph{canonical} subsets $\C_1, \ldots, \C_u$, 
associated with respective $(i, \cdot)$-level nodes $v_1, \ldots, v_u$ of $\Psi$.
For $i=0$, $\Qout_\Pi^{(0)}(\query) = \Qout_\Pi(\query)$, so
these canonical subsets form the output of the $\Pi$-query for $\query$, and we simply add 
their prestored weights (in the corresponding semigroup).
For $i\ge 1$, for each $j \le u$, we recursively search in the secondary structure 
$\Psi^{(i-1,t_{i-1})}_{v_j}$ for answering the extended $\Pi_{i-1}$-query on~$\C_j$ with~$\query$.

In more detail, we traverse $\Psi$ in a top-down manner, and maintain a partial sum $\mu$. 
Initially, $\mu=0$ and we start at the root of $\Psi$. 
Suppose we are at a node $v$ of an $(i,t)$-level tree. If $i=0$, we simply add $w_v = w(\oset_v)$ to $\mu$. 
So assume $i\ge 1$. If $v$ is a leaf of the tree, we scan the set $\oset_v$ and add to $\mu$ 
the weights of those objects $\obj\in\oset_v$ for which $\Pi_i(\obj,\query)=1$.
If $v$ is an internal node, three cases can arise. Let 
$\pregion{\query}_i =\{\xx_i \in \ospace_i \mid g_i(\xx_i,\query_i^*)\ge 0\}$
be the semi-algebraic set defined above for the $i$-th inequality.
If $\tau_v\cap\pregion{\query}_i=\emptyset$, we do not continue the processing at $v$.
If $\tau_v \subseteq \pregion{\query}_i$,
we recursively search in the secondary structure $\Psi_v^{(i-1,t_{i-1})}$ with $\query$. 
Finally, if $\bd\pregion{\query}^{(i)} \cap \tau_v\ne \emptyset$, which is equivalent to $Z(g_i) \cap \tau_i \ne \emptyset$, we first recursively visit
the secondary structure $\Psi_v^{(i,t-1)}$ (to search in the set $\oset_v^0$),
and then recursively search at every child of~$v$. 

The canonical subsets associated with the level-$0$ nodes that are visited by the query procedure, i.e., 
the nodes whose weights are added to produce $\mu$, induce a partition of $\Qout_\Pi(Q)$. 
In contrast, the canonical subsets associated with the $(i,\cdot)$-level nodes $v$
visited by the query procedure for which $\tau_v\cap\pregion{\query}=\emptyset$ induce a partition of 
$\Qout_{\overline{\Pi}_i}(\query)$ where
\begin{equation}
	\label{eq:one-neg}
	\overline{\Pi}_i (\xx,\yy) = \Pi_{i-1} (\xx,\yy) \wedge \neg \pi_i(\xx,\yy).
\end{equation}
By collecting such canonical subsets at all levels, the same query procedure can also be used to compute a 
partition of ${\Qout}_{\neg\Pi}(\query) = \{ \obj\in\OS \mid \Pi(\obj,\query)=0\}$ into a few canonical subsets.
\paragraph{Analysis.} 
We now analyze the size and the query time of $\Psi$. For a node $v$ at some $(i,t)$-level, let 
$S(n_v,i,t)$ be the maximum size of the partition tree constructed on at most $n_v$ objects with these
parameters. We note that $S(n_v,0,t) = O(1)$ and that for $n_v \le n_0$ and $i\ge 1$ we have
$S(n_v,i,t) = O(n_v) = O(1)$ too. For $t=1$, the corresponding tree is a balanced binary tree, so $S(n_v,i,1)$ 
satisfies the following recurrence (for $n_v > n_0(i,1)$):
\begin{equation}
	S(n_v, i, 1) \le 2 S(n_v/2, i, 1) + S(n_v, i-1, t_{i-1}).
\end{equation}
Since $S(n_v,i,t)=\Omega(n_v)$,
the solution of the above recurrence is 
\[
S(n_v,i,1) = O(\log n_v) \cdot S(n_v,i-1,t_{i-1}) .
\]
Finally, for $t>1$ and $n_v > n_0$, we store at $v$ two
secondary structures of levels $(i,t-1)$ and $(i-1,t_{i-1})$, and we recursively construct an
$(i,t)$-level subtree on a set of at most $n_v/D$ objects for each of the $O(D)$ children of $v$. 
Hence, in this case, the recurrence is
\begin{equation}
	S(n_v,i,t) \le c_1 D \cdot S(n_v/D,i,t) + S(n_v,i-1,t_{i-1}) + S(n_v,i,t-1),
\end{equation}
where $t_{i-1} = \dim\left(\ospace_{i-1}\right)$ and $c_1 \coloneqq c_1(i,t)$ is a constant.
Using induction, it easily follows that, for any arbitrarily small constant $\eps>0$, there exists a constant 
$A \coloneqq A(i,t,\eps)$ such that the solution of the above recurrence is 
\[
S(n_v,i,t) \le A n_v^{1+\eps} .
\]
That is, $S(n_v,i,t) = O^*(n_v)$ for each $i$ and $t$, where the constant of proportionality 
also depends on $i$ and $t$.

Next, let $Q(n_v,i,t)$ be the maximum time spent by a query at a node $v$ of some $(i,t)$-level
that stores $n_v$ objects. We have $Q(n_v,0,t) = O(1)$, and $Q(n_v,i,t) = O(n_v)$ for $n_v \le n_0(i,t)$.  
For $t=1$, the analysis for a one-dimensional range tree implies that the query procedure visits nodes 
along $O(1)$ paths of the tree and the subproblem size at a node is at most half of that of its parent. 
Since $Q(n_v,i,t_i) = \Omega(n^\eps)$, for some constant $\eps>0$, we obtain 
\[
Q(n_v,i,1) = O(1) \cdot Q(n_v, i-1, t_{i-1}).
\]
Finally, consider $i\ge 1$, $t>1$ and $n_v > n_0(i,t)$. Since the degree of the partitioning polynomial is 
$O(D^{1/t+\delta})$, for a suitably small constant $\delta>0$,
and the degree of $g_i$ is constant, $Z(g_i)$ intersects $O(D^{1-1/t+\delta})$ cells of 
$Z(\F)\setminus Z(F_v)$~\cite{BPR}, which 
leads to the following recurrence:
\begin{equation}
	Q(n_v,i,t) \le c_2 D^{1-1/t+\delta} Q(n_v/D,i,t) + Q(n_v,i-1,t_{i-1}) + Q(n_v,i,t-1) ,
\end{equation}
where $c_2 \coloneqq c_2(i,t)$ is a constant. 
The first term in the above recurrence follows from \cite{MP}, and the second and third terms correspond to the 
query procedure visiting the secondary structures stored at $v$.
Again, it can easily be shown that, for any arbitrarily small constant $\eps> \delta$, there exists a 
constant $B \coloneqq B(i,t,\eps)$ such that the solution of the above recurrence is 
\[
Q(n_v,i,t) \le B n_v^{1-1/\tO +\eps} ,\qquad\text{where}\quad \tO  = \max_{1\le i\le k} t_i .
\]
That is, $Q(n_v,i,t) = O^*(n_v^{1-1/\tO})$ for each $i$ and $t$, where the constant of proportionality 
depends on $i$ and $t$. The same analysis implies that the subset of input objects that satisfy the query predicate
can be represented as the union of $ O^*(n_v^{1-1/\tO})$ pairwise-disjoint canonical subsets.

Finally, as mentioned above, the data structure can be adapted to handle disjunctions for general aggregation 
queries such as counting queries. More precisely, suppose, for concreteness,
that $\Pi(\xx,\yy) = \Pi_1(\xx,\yy) \vee \Pi_2(\xx,\yy)$, 
where each $\Pi_i$ is composed of only conjunctions. Set $\Pi'_2(\xx,\yy)=(\neg\Pi_1(\xx,\yy))\wedge 
\Pi_2(\xx,\yy)$.
We first build a multi-level data structure $\Psi_1$ 
for $\Pi_1$. At each node $v$ of $\Psi_1$ (at all levels), we construct a data structure $\Psi_2^v$ 
for $\Pi_2$-queries on the 
corresponding subset of input objects. To answer a $\Pi$-query, we first compute $\mu_1$,
the sum of weights of objects in $\Qout_{\Pi_1}(\query)$. As mentioned above, the same procedure can also compute 
a partition of $\neg\Qout_{\Pi_1}(\query)$ into a few canonical subsets. Each such canonical subset $\C_v$ is associated with a node $v$ of the partition tree. We query $\Psi_2^v$ with $\query$ and obtain  the sum of weights of points 
in $\C_v\cap \Qout_{\Pi_2}(\query)$.  
The sum of the weights over all canonical subsets of $\Qout_{\neg \Pi_1}(\query)$, denoted by $\mu_2$, returns the 
weight of $\Qout_{\Pi'_2}(\query)$. We return $\mu_1+\mu_2$. It is easily seen that no object is included 
multiple times in the sum.

We thus have the following result.
\begin{theorem}
	\label{lem:Pi-lin}
	Let $\oset$ be a set of $n$ geometric objects of constant complexity with reduced parametric dimension~$\tO$,
        and let $\qset$ be a family of query objects of constant complexity. Let 
	$\Pi \colon \oset\times\qset \rightarrow \{0,1\}$ be a semi-algebraic predicate of constant complexity. 
	$\oset$ can be preprocessed, in $O^*(n)$ randomized expected time, into a data structure of size $O^*(n)$,
        so that a $\Pi$-query on $\oset$ (with respect to any semigroup $\Sigma$)
	with an object in $\qset$ can be answered in $O^*(n^{1-1/\tO})$ time. 
	The subset of input objects that satisfy the query predicate can be represented as the union of 
	$O^*(n^{1-1/\tO})$ pairwise-disjoint canonical subsets.
\end{theorem}
\subsection{Space/query-time trade-off}
\label{subsec:Pi-tradeoff}

Next, we show how the query time can be improved by increasing the size of the data structure. 
We now define $t_i = \dim(\qspace_i)$ and set $\tQ \coloneqq \max_{1\le i \le k} t_i$. That is,
$\tQ$ is the reduced parametric dimension of $\qset$ with respect to $\Pi$. Let $s\in [n,n^{\tQ}]$ be 
a \emph{storage parameter}; 
the data structure we build will have size (and expected preprocessing cost) $O^*(s)$.

\paragraph{Data structure.}
We now work in the query space $\qspace$ and construct a multi-level partition tree $\qtree$ on 
$\pregion{\oset} \coloneqq \{ \pregion{\obj} \mid \obj \in \oset\}$ for answering 
point-enclosure queries with a query point $\query^*\in\qspace$, using Lemma~\ref{lem:aaez}. When the size of a subproblem falls 
below some threshold (now in general not a constant), 
we switch to the object space and build a data structure on the
current canonical set of objects using Theorem~\ref{lem:Pi-lin}. 
If we set the threshold to a constant and simply store the $O(1)$ objects at each leaf, 
we obtain a data structure of size $O^*(n^{\tQ})$ with $O^*(1)$ query time, see below.
Our structure interpolates between these two extreme performance bounds, the one just mentioned and the 
one in Theorem~\ref{lem:Pi-lin}. 

At each recursive step, we are at a node $v$ of $\qtree$, and we construct an $(i,t)$-level 
partition tree $\qtree^{(i,t)}$, for $0 \le i \le k$ and $1 \le t \le t_i$, for answering 
$\Pi_i$-queries on some canonical subset $\oset_v \subseteq \oset$, with a query object $\query$,
such that $\query_i^*$ lies on some given $t$-dimensional variety of constant complexity. 
More precisely, we now have a $4$-tuple
$(\oset_v, s_v, \F, i)$ where $\F$ is a set of $O(1)$ constant-degree polynomials in $\reals[\yy_i]$, 
$\oset_v \subseteq \oset$ is a canonical set associated with $v$, and $s_v$ is a storage parameter. 
We wish to construct a data structure of size $O^*(s_v)$ on $\oset_v$, in expected $O^*(s_v)$
preprocessing time, for answering efficiently $\Pi_i$-queries on $\oset_v$ with a query object $\query$ 
such that $\query_i^* \in Z(\F)$; see below for the analysis of the efficiency of the query. 
As before, $v$ is associated with a semi-algebraic cell 
$\tau_v \subset Z(\F)$, obtained from a suitable polynomial partitioning. 
Initially, $\oset_v = \oset$, $s_v=s$, $\F=\emptyset$, $i=k$, and $\tau_v=\qspace^{(i)}$. 

For $i=0$, $\qtree^{(0,t)}$ is a singleton node that stores $w(\oset_v)$. So assume $i\ge 1$. 

If $t=1$, $Z(\F)$ is a one-dimensional curve, and we consider query objects $\query$ for which
$\query_i^*$ lies on $Z(\F)$. As before, it suffices to consider situations in which $Z(\F)$
is connected. For an object $\obj\in\oset$, let $\pregion{\obj}\sqcap Z(\F)$ 
denote the set of connected components (arcs) of $\pregion{\obj}\cap Z(\F)$. We compute 
$\I_v \coloneqq \bigcup \{ \pregion{\obj}_i \sqcap Z(\F) \mid \obj \in \oset_v\}$, which is a 
set of $O(|\oset_v|)$ arcs on $Z(\F)$, and build a segment tree $\qtree_v^{(i,1)}$ on 
$\I_v$~\cite{dBCKO08}. 
Each node $z$ of $\qtree_v^{(i,1)}$ is associated with an arc 
$\tau_z \subseteq Z(\F)$, a subset $\oset_z\subseteq \oset_v$ of objects $\obj$ for which an endpoint of 
$\pregion{\obj}_i \cap Z(\F)$ lies in $\tau_z$,  and another subset $\C_z \subseteq \oset_v$, such that, for each
$\pregion{\obj} \in \C_z$ we have
$\tau_z \subseteq \pregion{\obj}_i$ but $\tau_{p(z)} \not\subset \pregion{\obj}_i\cap Z(\F)$ (where $p(z)$ is the parent of $z$).
If $z$ is a node of depth $\delta$ in the segment tree $\qtree_v^{(i,1)}$,
we construct an $(i-1, t_{i-1})$-level structure $\qtree_z^{(i-1,t_{i-1})}$ 
on $\C_z$ with space parameter $s_v/2^\delta$, i.e., for the subproblem $(\C_z, s_v/2^\delta, \emptyset, i-1)$, 
and attach $\qtree_z^{(i-1,t_{i-1})}$ to $z$ as its secondary structure. 

Finally, assume that $i\ge 1$ and $t>1$. 
We fix a threshold parameter $n_0 \coloneqq (n^t/s)^{\tfrac{1}{t-1}}$.
If $n_v \le n_0$, $v$ is a leaf of $\qtree$, and we construct a data structure $\Psi_v$ of size 
$O^*(n_v)$, as described in Section~\ref{subsec:Pi-lin}, for the subproblem $(\oset_v,\emptyset,i)$,
using Theorem~\ref{lem:Pi-lin} (the storage parameter is irrelevant in this case).
So assume that $n_v > n_0$. Set 
$\pregion{\oset}_{v,i} \coloneqq \{ \pregion{\obj}_i \subseteq \qspace_i \mid \obj\in\oset_v\}$, which is
a family of semi-algebraic sets in $\qspace_i$.
We choose a sufficiently large constant $D \coloneqq D(i,t)$
(in particular, $D(i,t-1)\gg D(i,t)$), 
and compute a 
partitioning polynomial $F_v$ for $\pregion{\oset}_{v,i}$ with respect to $Z(\F)$,
using Lemma~\ref{lem:aaez} and the notations therein. Taking $D$ to be sufficiently large, we write
the degree of $F_v$ as $O(D^{1+\delta})$, for an arbitrary constant $\delta>0$, and the number of cells 
in $Z(\F)\setminus Z(F_v)$ is $O(D^{t+\delta})$, 
where each cell is crossed by at most $n_v/D$ regions.
We recursively construct an $(i,t-1)$-level data structure $\qtree^{(i,t-1)}_v$ for the 
subproblem $(\oset_v, s_v, \F\cup\{F_v\}, i)$ and attach it to $v$ as one of its secondary structures. 
For each cell $\tau$ of $Z(\F) \setminus Z(F_v)$, we create a child $z_\tau$ of $v$. 
We construct a semi-algebraic representation of $\tau$ and store it at $z_\tau$.
Set $\oset_\tau = \{ \obj\in\oset_v \mid \bd\pregion{\obj}_i \cap \tau \ne \emptyset\}$, and
$\C_\tau = \{ \obj\in\oset_v \mid \tau \subseteq \pregion{\obj}_i\}$.
We construct an $(i-1,t_{i-1})$-level data structure $\qtree_\tau^{(i-1,t_{i-1})}$ for the subproblem
$(\C_\tau,s_v, \emptyset, t_{i-1})$ and store 
it as a secondary structure at the child $z_\tau$. Finally, we recursively construct at $z_\tau$
an $(i,t)$-level subtree $\qtree_\tau^{i,t}$ for the subproblem 
$(\oset_\tau, s_v/D^t, \F, i)$.

As for the linear-size data structure, two secondary structures are attached to each node $v$---an $(i,t-1)$-level data structure on $\oset_v^0$ 
and an $(i-1,t_{i-1})$-level data structure on the subset of objects $\obj\in\oset_{p(v)}$ for which 
$\tau_v \subseteq \pregion{\obj}_i$. 
This completes the description of the data structure.

\begin{figure}[htb]
  \centering
  \includegraphics[scale=0.8]{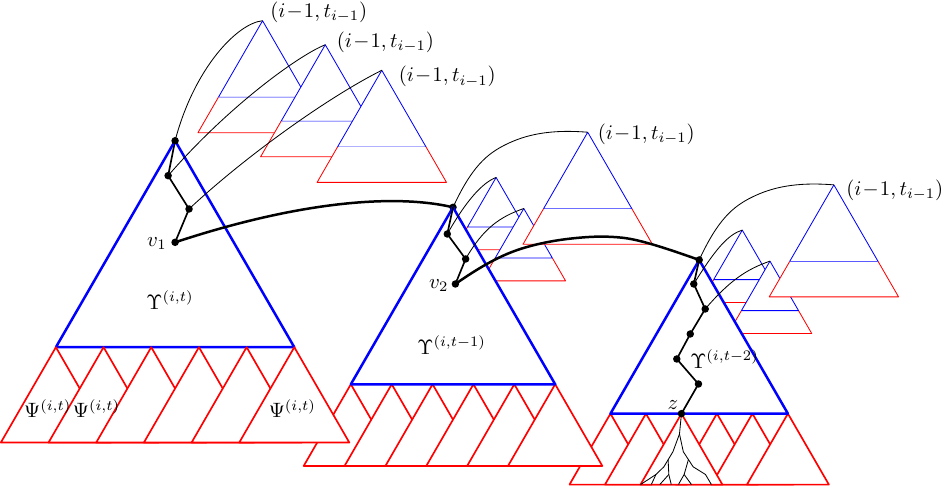}
  \caption{Schematic diagram of $\qtree^{(i,t)}$. The query procedure visits a path in level 
    $(i,\cdot)$ trees (shown in black), recursively visits nodes in $(i-1,\cdot)$-level 
    secondary structures attached to each node on the path, and eventually visits a subtree of 
    the small-size data structure $\Psi$ attached to a leaf (e.g., node $z$ in the figure).}
  \label{fig:schem2}
\end{figure}
\paragraph{Query procedure.}
Let $\query\in\qset$ be a query object. The basic idea is the same as in Section~\ref{subsec:Pi-lin}: 
we compute a partition of $\Qout_\Pi(\query)$ into canonical subsets, which are associated with (some of)
the nodes that the query procedure reaches. 
Since we allow more storage, the number of canonical subsets in the partition becomes smaller. 

In more detail, we traverse $\qtree$ in a top-down manner, and maintain a partial sum $\mu$. 
Initially, $\mu=0$ and we start at the root of $\qtree$. 
Suppose we are at a node $v$ of an $(i,t)$-level tree $\qtree_v^{(i,t)}$, such that $\query_i^* \in \tau_v$. 
If $i=0$, we add the weight $w_v = w(\oset_v)$ to $\mu$ and back up. So assume $i\ge 1$. 

We recursively query the $(i-1,t_{i-1})$-level secondary structure 
$\qtree^{(i-1,t_{i-1})}$ with $\query$ and add the resulting weight to $\mu$.
Next, if $v$ is a leaf, we query the structure $\Psi_v$ stored at $v$,
using Theorem~\ref{lem:Pi-lin}, and add the resulting aggregate weight to $\mu$. 
Finally, assume that $v$ is an internal node. 
If $t=1$, we check which of the two intervals associated with the children of $v$ contains $\query_i^*$,
and recursively search there. If $t>1$, we check whether the point $\query_i^*$ lies on $Z(F_v)$. 
If the answer is yes, we recursively search with $\query$ in the $(i,t-1)$-level structure $\qtree^{(i,t-1)}$ 
stored at $v$. Otherwise we find the child $z$ of $v$ such that $\query_i^* \in \tau_z$,
and recursively search at $z$.
Adding to $\mu$ the weights resulting from each recursive call, we obtain the aggregate output of the query.

\paragraph{Analysis.} 
We now analyze the size and the query time of $\qtree$. For a node $v$ at some $(i,t)$-level, 
let $S(n_v, s_v, i, t)$ be the maximum size of the partition tree constructed at $v$ 
(i.e., on a set of at most $n_v$ objects, with storage parameter $s_v$, and parameters $i$ and $t$). 
We note that $S(n_v,s_v,0,t) = O(1)$. %
For $t=1$, the structure is a segment tree, which is basically a balanced binary tree and each object is stored at $O(\log n_v)$ nodes, so it is easily seen that
$$S(n_v,s_v,i,1) = O(\log n_v) \cdot S(n_v,s_v,i-1, t_{i-1}).$$

Finally, for $i\ge 1, t>1, n_v > n_0$, we store two
secondary structures of levels $(i,t-1)$ and $(i-1,t_{i-1})$, and 
we construct an $(i,t)$-level subtree on a set of at most $n_v/D$ objects for each child of $v$, 
leading to the following recurrence:
\begin{equation}
	S(n_v, s_v,i,t) \le c_3 D^{t+\delta} \cdot S(n_v/D,s_v/D^t,i,t) + S(n_v,s_v,i-1,t_{i-1}) + S(n_v,s_v,i,t-1),
\end{equation}
where $c_3 \coloneqq c_3(i,t)$ is a constant.
By induction on $n_v, i, t$, we can show that for any arbitrarily small constant $\eps>\delta$, 
the overall size of the data structure is $O(sn^\eps)=O^*(s)$.
The analysis of the expected preprocessing cost is nearly identical, and yields the same bound $O^*(s)$.

Next, let $Q(n_v,s_v,i,t)$ denote the maximum query time (for at most $n_v$ objects, 
storage parameter $s_v$, and parameters $i$ and $t$). As before, $Q(n_v,s_v,0,t) = O(1)$. 

For $i \ge 1$ and $t=1$, 
the query procedure follows a single path in the segment tree until it reaches a leaf, so 
\[
	Q(n_v,s_v,i, 1) \le \sum_{j=0}^{\log (n_v/n_0)} Q(n_v/2^j, s_v/2^j, i-1,t_{i-1}) + O^*(n_0^{1-1/\tO}).
\]
Next, consider the case $i\ge 1$ and $t>1$. 
For $n_v \le n_0$, by Theorem~\ref{lem:Pi-lin}, 
the query time is $O(n_v^{1-1/\tO+\eps})$, for an arbitrarily small constant $\eps>0$,
so assume that $n_v > n_0$.
Recall that the query procedure visits the $(i-1,t_{i-1})$-level secondary structure attached to $v$, 
and it either visits a child of $v$ or the $(i,t-1)$-level secondary 
structure attached to $v$, depending on whether the point $\query_i^*$ does not or does lie on $Z(F_v)$.
We thus obtain the following recurrence:
\begin{equation}
	\label{eq:qtime-tradeoff}
	Q(n_v,s_v,i,t) \le Q(n_v,s_v,i-1,t_{i-1}) +  \max \Bigl\{Q(n_v/D,s_v/D^t,i,t),\, Q(n_v,s_v,i,t-1)\Bigr\}.
\end{equation}
We claim that the overall query time is 
\begin{equation} \label{eq:query-bd}
	Q(n_v,s_v,i,t) \le c_i n_0^{1-1/\tO+\eps} \log^{i-1} (n_v/n_0) (\log (n_v/n_0)+t) ,
\end{equation}
for any $\eps>0$ and a suitable sequence of constant coefficients $c_i$.
For $n_v > n_0$, by 
induction on $n_v, i$, and $t$, using~\eqref{eq:qtime-tradeoff}, and using the fact that the query time is 
$O(n_v^{1-1/\tO+\eps})$ for $n_v \le n_0$, we obtain for $t>1$ (the case of $t=1$ is simpler)
\begin{align*}
	Q(n_v,s_v,i,t) & \le c_{i-1} n_0^{1-1/\tO+\eps} \log^{i-2} (n_v/n_0)(\log (n_v/n_0)+t_{i-1})\\
			&\quad + c_i n_0^{1-1/\tO+\eps} \log^{i-1} (n_v/n_0) 
				\max \biggl \{ \log \frac{n_v/D}{n_0} + t,\, 
						\log \frac{n_v}{n_0} + t-1 \biggr \} \\
			& \le  c_i n_0^{1-1/\tO+\eps} \log^{i-1} (n_v/n_0) 
			\left ( 1 + 	\max \biggl \{ \log \frac{n_v/D}{n_0} + t,\,
				\log \frac{n_v}{n_0} + t-1 \biggr \} \right ) \\
				&\qquad \mbox{(assming $c_i \ge t c_{i-1}$)}\\
			& \le  c_i n_0^{1-1/\tO+\eps} \log^{i-1} (n_v/n_0) (\log(n_v/n_0)+t) .
\end{align*}

Hence, the overall query time is
\[
  O^*(n_0^{1-1/\tO}) = O^* \biggl ( (n^{\tQ }/s)^{\frac{1-1/\tO}{\tQ -1}} \biggr ) 
	= O^* \biggl ( (n/s^{1/\tQ})^{\frac{1-1/\tO}{1-1/\tQ}} \biggr ) .
\]
We thus have the main result of this appendix:
\begin{theorem}
	\label{thm:Pi-tradeoff}
	Let $\oset$ be a set of $n$ geometric objects of constant complexity,
        and let $\qset$ be a family of query objects of constant complexity. 
	Let $\Pi\colon \oset\times\qset \rightarrow \{0,1\}$ be a semi-algebraic predicate of constant complexity,
        and let $\tO, \tQ$ be the respective reduced parametric dimensions of $\oset$ and $\qset$
	with respect to $\Pi$. For a storage parameter $s\in[n,n^{\tQ}]$,
	$\oset$ can be preprocessed, in $O^*(s)$ expected time, into a data structure of size $O^*(s)$,
	so that a $\Pi$-query on $\oset$ with an object $\query\in\qset$ can be answered 
	in $O^*\biggl((n/s^{1/\tQ})^{\tfrac{1-1/\tO}{1-1/\tQ}} \biggr)$ time.
\end{theorem}

If we perform $m\ge 1$ $\Pi$-queries on $\oset$ with $m$ objects of $\qset$, 
then the total expected time spent, including the time to construct the data structure, is 
\[O^* \biggl (s+m(n/s^{1/\tQ})^{\tfrac{1-1/\tO}{1-1/\tQ}} \biggr ).\]
By choosing the storage parameter 
$$s = \min \left\{ n^\tQ, \max \left\{n, 
	m^{\tfrac{1-1/\tQ}{1-1/\tO\tQ}}n^{\tfrac{1-1/\tO}{1-1/\tO\tQ}}\right\}\right\},$$
we obtain the following:
\begin{corollary}
	\label{cor:Pi-tradeoff}
	Let $\oset$ be a set of $n$ geometric objects of constant complexity,
        and let $\qset$ be a family of query objects of constant complexity. 
	Let $\Pi\colon \oset\times\qset \rightarrow \{0,1\}$ be a semi-algebraic predicate of constant complexity,
        and let $\tO, \tQ$ be the respective reduced parametric dimensions of $\oset$ and $\qset$
	with respect to $\Pi$. For any given $m>0$, $m$ $\Pi$-queries on $\OS$ with $m$ objects of $\qset$ can be 
	performed in a total of 
	$$O^*\biggl (m^{\tfrac{1-1/\tQ}{1-1/\tO\tQ}}n^{\tfrac{1-1/\tO}{1-1/\tO\tQ}}+m+n \biggr )$$
	expected time. For $\tO=\tQ=t$, the total expected time is 
	$$O^* \biggl ((mn)^{\tfrac{1}{1+1/t}}+m+n \biggr ).$$
\end{corollary}
Corollary~\ref{cor:Pi-tradeoff} can be used to compute a compact representation of 
the collection of pairs of objects that satisfy a semi-algebraic predicate: 
Let $\Aset$ and $\Bset$ be two sets of geometric objects, let $\Pi \colon \Aset\times\Bset \rightarrow \{0,1\}$ be a 
semi-algebraic predicate of constant complexity, and let 
$\Aset\Pi\Bset = \{ (\Aobj,\Bobj) \mid \Pi(\Aobj,\Bobj)=1\}$ be the collection of pairs that satisfy the predicate. 
A \emph{biclique cover} of $\Aset\Pi\Bset$ is a family $\F=\{(\Aset_1,\Bset_1), \ldots, (\Aset_r,\Bset_r)\}$ 
such that (i) for any $1 \le i \le r$ and for any $(\Aobj,\Bobj)\in \Aset_i\times\Bset_i$, $\Pi(\Aobj,\Bobj)=1$, 
and (ii) for any $(\Aobj,\Bobj) \in \Aset\times\Bset$ with $\Pi(\Aobj,\Bobj)=1$, there is an 
index (not necessarily unique) $j\le r$ 
such that $(\Aobj,\Bobj)\in \Aset_j\times\Bset_j$. The size of the biclique cover $\F$ is 
$\sum_{i=1}^r \left( |\Aset_i|+|\Bset_i| \right)$. 
If every pair of $\Aset\Pi\Bset$ appears only once in $\F$, then we refer to $\F$ as a 
\emph{biclique partition} of $\Aset\Pi\Bset$. A biclique partition of $\Aset\Pi\Bset$ can be computed 
as follows: We preprocess $\Aset$ into a data structure using Theorem~\ref{lem:Pi-lin}, and we query it with each object in $\Bset$.
Let $\Aset_1, \ldots, \Aset_r$ be the family of canonical subsets constructed by the data structure. By construction,
the size of the data structure is proportional to $\sum_{i=1}|\Aset_i|$.
For an object $\Bobj\in\Bset$, let $\C(\Bobj)$ be the subfamily of the canonical subsets that represent 
the set $\Qout_\Pi(\Bobj)$. For $1 \le i \le r$, we set $\Bset_i\subseteq\Bset$ to be the subset of objects for which
$\Aset_i$ is in $\C(\Bobj)$. Then $\{(\Aset_i,\Bset_i) \mid 1 \le i \le r\}$ is the desired biclique partition, and 
$\sum_i |\Bset_i| = \sum_{\Bobj\in\Bset} |\C(\Bobj)|$.
By Corollary~\ref{cor:Pi-tradeoff}, we obtain the following:

\begin{corollary}
	\label{cor:biclique2} 
	Let $\Aset$ and $\Bset$ be two sets of geometric objects of sizes $n$ and $m$, respectively, 
	$\Pi  \colon \Aset\times\Bset \rightarrow \{0,1\}$ a semi-algebraic predicate of constant complexity, and 
	$\tA, \tB$ the reduced parametric dimension of $\Aset$ and $\Bset$, respectively. Then a 
	biclique partition of $\Aset\Pi\Bset$ of size  
	$O^*\biggl (m^{\tfrac{1-1/\tB}{1-1/\tA\tB}}n^{\tfrac{1-1/\tA}{1-1/\tA\tB}}+m+n \biggr )$
	can be computed in expected time 
	$O^*\biggl (m^{\tfrac{1-1/\tB}{1-1/\tA\tB}}n^{\tfrac{1-1/\tA}{1-1/\tA\tB}}+m+n \biggr )$.
	For $\tA=\tB=t$ and $m=n$, the bounds become $O^*(n^{\tfrac{2}{1+1/t}})$.
\end{corollary}

Special cases of Corollary~\ref{cor:biclique2} have been used for many  geometric 
optimization problems as well as for representing geometric graphs compactly~\cite{AV00,Mat91,AE99}. 
For example, it implies that the edges in the intersection graph of a set of $n$ 
segments or unit disks (resp.\ disks of arbitrary radii) in the plane can be represented as a 
biclique partition of size $O^*(n^{4/3})$ (resp.\ $O^*(n^{3/2})$) as $t=2$ (resp. $t=3$) in these setups.
The partition can be computed within the same time bound. 

Let $\Pi$ be a predicate of the form~\eqref{eq:semi-pred-conj} defined by a conjunction of 
polynomial inequalities.
Using the fact that the query procedure can also compute a partition of $\Qout_{\overline{\Pi}_i}(\query)$ 
(see \eqref{eq:one-neg}), 
for any $i \le k$, into a few canonical subsets, we can generalize the above corollary as follows:

\begin{corollary}
  \label{cor:biclique3} 
  Let $\Aset$ and $\Bset$ be two sets of geometric objects of sizes $n$ and $m$, respectively, 
  $\Pi  \colon \Aset\times\Bset \rightarrow \{0,1\}$ a semi-algebraic predicate of constant complexity of 
  the form $\Pi(A,B) = \bigwedge_{i=1}^k \pi_i(A,B)$, where each $\pi_i$ is defined by a 
  single polynomial inequality, and 
  $\tA, \tB$ the reduced parametric dimension of $\Aset$ and $\Bset$, respectively, with respect to $\Pi$. 
  Then a partition $\{(\Aset_1, \Bset_1), \ldots, (\Aset_u,\Bset_u)\}$ of $\Aset \times \Bset$ of size  
  $O^*\biggl (m^{\tfrac{1-1/\tB}{1-1/\tA\tB}}n^{\tfrac{1-1/\tA}{1-1/\tA\tB}}+m+n \biggr )$
  can be computed in expected time 
  $O^*\biggl (m^{\tfrac{1-1/\tB}{1-1/\tA\tB}}n^{\tfrac{1-1/\tA}{1-1/\tA\tB}}+m+n \biggr )$ such that for every 
  $1 \le i \le u$ and for every $1 \le j \le k$ either $\Pi_j(A,B)=1$ for all 
  $(A,B) \in \Aset_i \times \Bset_i$ or $\Pi_j(A,B) =0$ for all $(A,B) \in \Aset_i \times \Bset_i$.
  For $\tA=\tB=t$ and $m=n$, the bounds become $O^*(n^{\tfrac{2}{1+1/t}})$.
\end{corollary}

We conclude this discussion by making a few final remarks:
\smallskip
\begin{remarks*}
  (i) In the above data structure, we first construct a partition tree in the query space. 
  When the size of the subproblem falls below a certain threshold, we construct a partition 
  tree in the object space. We switch between the query and object spaces only once. 
  In some applications (see e.g.,~\cite{AS02,Mat93}), however, one obtains a better performance by 
  switching back and forth between the two spaces at each step.  Namely, we perform one step of polynomial 
  partitioning in the object space followed by one step of polynomial partitioning in the query space, 
  and then repeat. 

  (ii) The query time in the first data structure (see Theorem~\ref{lem:Pi-lin}) only 
  depends on $\tO$ (and not on $\tQ$). 
  Therefore if multiple polynomial inequalities use the same subset of parameters of the input 
  objects, we can combine these inequalities into a single predicate $\Pi_i^O(\xx_i,\yy)$. More 
  generally, the first data structure can handle a predicate of the form 
  \[ \Pi(\xx,\yy) = \bigwedge_{i=1}^k \Pi_i^O (\xx_i, \yy) , \]
  where each $\Pi_i^O$ is an arbitrary semi-algebraic predicate of constant complexity 
  (possibly containing disjunctions).
  Similarly, the top portion of the second data structure, which is constructed in the query space,
  can handle an equally general predicate of the form
  \[ \Pi(\xx,\yy) = \bigwedge_{i=1}^k \Pi_i^Q (\xx, \yy_i) ,\]
  by lumping together several subpredicates that use the same subset of the query parameters.
  We note that the compressed predicate could reduce the number of levels in the data structure, 
  and that each sub-predicate could have disjunctions. 
  Although this observation improves the performance of the data structure by only 
  at most a polylogarithmic factor, it often simplifies the data structure considerably.
  
  For example, suppose we wish to construct a linear-size data structure to perform 
  planar double-wedge range queries, where the query objects are 
  double wedges $W$, bounded by two lines $\ell_1$ and $\ell_2$, that do not contain the vertical line 
  (the segment-intersection query for a set of lines in $\reals^2$ can be formulated as 
  double-wedge range query using the standard duality transform). 
  Regarding $\ell_1, \ell_2$ as linear functions,
  the intersection predicate in this case is
  $$\Pi(\xx) = ((\ell_1(\xx)\ge 0) \wedge (\ell_2(\xx) \le 0))\vee 
  ((\ell_1(\xx)\le 0)\wedge (\ell_2(\xx)\ge 0)).$$
  Instead of constructing two separate two-level partition trees, we can construct a single 
  one-level partition tree on the input points. The query procedure recursively visits the children of a node 
  $v$ if the boundary $\bd W$ (i.e., one of the lines bounding $W$) intersects the cell $\tau_v$. 
  
  \smallskip

  (iii) In either of the two data structures, when the size of a subproblem falls below the threshold $n_0$ 
  (i.e., at a leaf), one could build a completely different data structure for answering queries on 
  the subproblem (such as our construction of a data structure based on polynomial partitioning in 
  $3$-space in Section~\ref{sec:trade-off}).
  Similarly, one could 
  replace the level-$1$ data structure with a different one that is more efficient. Such an approach sometimes 
  leads to better query time for detection (emptiness) and reporting queries, as in the case of 
  halfspace-emptiness queries. See, e.g.,~\cite{MoSh,ShSh}.
\end{remarks*}

\end{document}